\RequirePackage{fix-cm}
\documentclass[smallextended, envcountsame]{svjour3}
\smartqed  % flush right qed marks, e.g. at end of proof

\usepackage[T1]{fontenc}
\usepackage[utf8]{inputenc}

\usepackage{calc}
\usepackage{amsmath}
\usepackage{amssymb}
\usepackage{mathtools}
\usepackage{bm}
\usepackage{csquotes}
\usepackage{mathrsfs}
\usepackage{varwidth}
\usepackage{caption}
\usepackage{subcaption}
\usepackage{rotating}
\usepackage{enumitem}
\usepackage{adjustbox}

\usepackage{xspace}

\newcommand*{\GAut}{$\mathscr{G}$-automaton\xspace}
\newcommand*{\GAuta}{$\mathscr{G}$-automata\xspace}

\usepackage{tikz}
\usetikzlibrary{automata, positioning, matrix, calc, shapes, decorations.pathreplacing}
\pgfdeclaredecoration{pencilline}{initial}{
  \state{initial}[width=+\pgfdecoratedinputsegmentremainingdistance,auto corner on length=1mm,]{
    \pgfpathcurveto%
    {% From
      \pgfqpoint{\pgfdecoratedinputsegmentremainingdistance}
      {\pgfdecorationsegmentamplitude}
    }
    {%  Control 1
      \pgfmathrand
      \pgfpointadd{\pgfqpoint{\pgfdecoratedinputsegmentremainingdistance}{0pt}}
      {\pgfqpoint{-\pgfdecorationsegmentaspect\pgfdecoratedinputsegmentremainingdistance}%
        {\pgfmathresult\pgfdecorationsegmentamplitude}
      }
    }
    {%TO 
      \pgfpointadd{\pgfpointdecoratedinputsegmentlast}{\pgfpoint{1pt}{1pt}}
    }
  }
  \state{final}{}
}

\spdefaulttheorem{fact}{Fact}{\bfseries}{\itshape}

\usepackage{tabularx}
\renewcommand{\problem}[3][]{%
  \par\vspace{0.125cm plus 0.1cm minus 0.05cm}\noindent\hskip\parindent\begin{tabularx}{\textwidth-\parindent}{@{}lX}%
    \if\relax\detokenize{#1}\relax%
    \else%
    \textnormal{\textbf{Constant:}}&#1\\%
    \fi%
    \textnormal{\textbf{Input:}}&#2\\%
    \textnormal{\textbf{Question:}}&#3\\%
  \end{tabularx}\vspace{0.125cm plus 0.1cm minus 0.05cm}\par%
}
\newcommand{\function}[2]{%
  \par\vspace{0.125cm plus 0.1cm minus 0.05cm}\begin{tabularx}{\textwidth-2\parindent}{lX}%
    \textnormal{\textbf{Input:}}&#1\\%
    \textnormal{\textbf{Output:}}&#2\\%
  \end{tabularx}\vspace{0.125cm plus 0.1cm minus 0.05cm}\par%
}

\newlength{\edgelength}
\newcommand{\trans}[4]{%
  \begin{tikzpicture}[auto, shorten >=1pt, >=latex, baseline=(l.base), inner sep=0pt, outer xsep=0.3333em]
    \node[outer xsep=0pt] (l) {\ensuremath{#1}};%
    \setlength{\edgelength}{\widthof{\scriptsize\ensuremath{#2/#3}}+0.5cm}%
    \node[base right=\edgelength of l] (r) {\ensuremath{#4}};%
    \path[->] (l.mid east) edge node[inner sep=0pt] {\scriptsize\ensuremath{#2/#3}} (r.mid west);%
    \node[inner sep=0pt] at (l.south west) {};% This is an ugly hack and I do not know why it is needed; however, the spacing is wrong otherwies in the definition of \delta'' in "Encoding over Two Letters"
  \end{tikzpicture}%
}

\newcommand{\transa}[3]{%
  \begin{tikzpicture}[auto, shorten >=1pt, >=latex, baseline=(l.base), inner sep=0pt, outer xsep=0.3333em]
    \node (l) {\ensuremath{#1}};%
    \setlength{\edgelength}{\widthof{\scriptsize\ensuremath{#2}}+0.5cm}%
    \node[base right=\edgelength of l] (r) {\ensuremath{#3}};%
    \path[->] (l.mid east) edge node[inner xsep=0pt, inner ysep=0.2em] {\scriptsize\ensuremath{#2}} (r.mid west);%
  \end{tikzpicture}%
}

\newsavebox{\blankboxdisplay}
\savebox{\blankboxdisplay}{\hspace{0.1ex}\tikz[baseline=0.1em]{%
    \node [shape=rectangle, anchor=south, draw, solid, inner sep=0pt, minimum width=1ex, minimum height=0.9em] (char) {};}%
  \hspace{0.1ex}}
\newsavebox{\blankboxtext}
\savebox{\blankboxtext}{\hspace{0.1ex}\tikz[baseline=0.1em]{%
    \node [shape=rectangle, anchor=south, draw, solid, inner sep=0pt, minimum width=1ex, minimum height=0.9em] (char) {};}%
  \hspace{0.1ex}}
\newsavebox{\blankboxscript}
\savebox{\blankboxscript}{\scriptsize\hspace{0.1ex}\tikz[baseline=0.1em]{%
    \node [shape=rectangle, anchor=south, draw, solid, inner sep=0pt, minimum width=1ex, minimum height=0.9em] (char) {};}%
  \hspace{0.1ex}}
\newsavebox{\blankboxscriptscript}
\savebox{\blankboxscriptscript}{\tiny\hspace{0.1ex}\tikz[baseline=0.1em]{%
    \node [shape=rectangle, anchor=south, draw, solid, inner sep=0pt, minimum width=1ex, minimum height=0.9em] (char) {};}%
  \hspace{0.1ex}}
\newcommand{\blank}{\makeatletter%
  \ensuremath{\mathchoice%
    {\text{\usebox\blankboxdisplay}}%
    {\text{\usebox\blankboxtext}}%
    {\text{\usebox\blankboxscript}}%
    {\text{\usebox\blankboxscriptscript}}}%
  \makeatother}

\newsavebox{\circledOnebox}
\savebox{\circledOnebox}{\tikz[baseline=(s.base)]{\node[draw, circle, solid, inner sep=0.1mm] (s) {$1$};}}
\newcommand*{\circledOne}{\usebox{\circledOnebox}}

\newsavebox{\circledZerobox}
\savebox{\circledZerobox}{\tikz[baseline=(s.base)]{\node[draw, circle, solid, inner sep=0.1mm, outer sep=0pt] (s) {$0$};}}
\newcommand*{\circledZero}{\usebox{\circledZerobox}}

\newcommand*{\inbox}[1]{\tikz[baseline=(s.base)]{\node[draw, rectangle, solid, inner sep=0.5mm] (s) {$#1$};}}
\newsavebox{\boxedOnebox}
\savebox{\boxedOnebox}{\inbox{1}}
\newcommand*{\boxedOne}{\usebox{\boxedOnebox}}

\newsavebox{\boxedZerobox}
\savebox{\boxedZerobox}{\inbox{0}}
\newcommand*{\boxedZero}{\usebox{\boxedZerobox}}

\newcommand{\boxedAlph}{\{ \, \boxedZero, \boxedOne \, \}}

\newcommand{\checksub}[1]{\ensuremath{\checkmark_{\kern-0.35em#1}}}
\newcommand{\eqA}{=_{\kern-0.2emA_5}}
\newcommand{\sub}[1]{_{\kern-0.25ex #1}}

\newcommand*{\ComplexityClass}[1]{\ensuremath{\mathsf{#1}}}
\newcommand*{\PSPACE}{\ComplexityClass{PSPACE}\xspace}
\newcommand*{\EXPSPACE}{\ComplexityClass{EXPSPACE}\xspace}
\newcommand{\NC}{\ComplexityClass{NC}\xspace}

\newcommand{\LOGSPACE}{\ComplexityClass{LOGSPACE}\xspace} %
\newcommand{\DLINTIME}{\ComplexityClass{DLINTIME}\xspace}
\newcommand{\DLINSPACE}{\ComplexityClass{DLINSPACE}\xspace}

\newcommand*{\DecProblem}[1]{\texttt{\textsc{#1}}}

\DeclareMathOperator{\id}{id}
\usepackage{bbold}
\newcommand{\idGrp}{\mathbb{1}}
\newcommand{\revbin}{\overleftarrow{\operatorname{bin}}}
\newcommand{\PPre}{\operatorname{PPre}}

\journalname{Theory of Computing Systems}

\spnewtheorem*{xproof}{}{\itshape}{\rmfamily}% the label is assigned later

\renewenvironment{proof}[1][\proofname]
{\renewcommand\xproofname{#1}\xproof}
{\endxproof}
\usepackage[
  pdfpagelabels,
  hypertexnames=true,
  plainpages=false,
  naturalnames=false]{hyperref}
\usepackage{color}
\definecolor{darkblue}{rgb}{0,0.1,0.5}
\hypersetup{colorlinks,
  linkcolor=darkblue,
  anchorcolor=darkblue,
  citecolor=darkblue}

\usepackage[nameinlink, capitalise]{cleveref}
\makeatletter
\newcommand{\labelx}[1]{
  \relax
  \ifmmode
    \label{#1}
  \else
    \ifnum\pdfstrcmp{\@currenvir}{document}=0
      \label{#1}
    \else
      \label[\@currenvir]{#1}
    \fi
  \fi
}
\makeatother
\makeatletter
\if@cref@capitalise
\crefname{fact}{Fact}{Facts}
\else
\crefname{fact}{fact}{facts}
\fi
\makeatother
\Crefname{fact}{Fact}{Facts}

\begin{document}
  \title{An Automaton Group with \PSPACE-Complete Word Problem}
  
  \author{Jan Philipp Wächter \and
    Armin Weiß\thanks{The second author was funded by DFG project DI 435/7-1.}
  }
  
  
  \institute{Jan Philipp Wächte \and Armin Weiß \at
    Universität Stuttgart,
    Institut für Formale Methoden der Informatik (FMI),
    Universitätsstraße 38,
    70569 Stuttgart, Germany\\
    \email{$\{ \text{jan-philipp.waechter}, \text{armin.weiss} \}$@fmi.uni-stuttgart.de}
  }
  
  \date{Received: date / Accepted: date}
  
  \maketitle
  
  \begin{abstract}
    We construct an automaton group with a $\PSPACE$-complete word problem, proving a conjecture due to Steinberg. Additionally, the constructed group has a provably more difficult, namely $\EXPSPACE$-complete, compressed word problem and acts over a binary alphabet. Thus, it is optimal in terms of the alphabet size. Our construction directly simulates the computation of a Turing machine in an automaton group and, therefore, seems to be quite versatile. It combines two ideas: the first one is a construction used by D'Angeli, Rodaro and the first author to obtain an inverse automaton semigroup with a $\PSPACE$-complete word problem and the second one is to utilize a construction used by Barrington to simulate Boolean circuits of bounded degree and logarithmic depth in the group of even permutations over five elements.
  \end{abstract}

  \begin{section}{Introduction}
    The word problem is one of Dehn's fundamental algorithmic problems in group theory \cite{dehn11}: given a word over a finite set of generators for a group, decide whether the word represents the identity in the group. While, in general, the word problem is undecidable \cite{nov55,boone59}, many classes of groups have a decidable word problem. Among them is the class of automaton groups.
    In this context, the term \emph{automaton} refers to finite state, letter-to-letter transducers. In such automata, every state $q$ induces a length-preserving, prefix-compatible action on the set of words, where an \emph{input word} $u$ is mapped to the \emph{output word} obtained by reading $u$ starting in $q$. The group or semigroup generated by the automaton is the closure under composition of the actions of the different states and a (semi)group arising in this way is called an \emph{automaton (semi)group}.

    The interest in automaton groups was stirred by the observation that many groups with interesting properties are automaton groups. Most prominently, the class contains the famous Grigorchuk group \cite{Grigorchuk80} (which is the first example of a group with sub-exponential but super-polynomial growth and admits other peculiar properties, see \cite{grigorchuk2008groups} for an accessible introduction). 
    There is also a quite extensive study of algorithmic problems in automaton (semi)groups: the conjugacy problem and the isomorphism problem (here the automaton is part of the input) -- the other two of Dehn's fundamental problems -- are undecidable for automaton groups \cite{sunic2012conjugacy}. Moreover, for automaton semigroups, the order problem could be proved to be undecidable \cite[Corollary~3.14]{gillibert2014finiteness}. Recently, this could be extended to automaton groups \cite{gillibert2018automaton} (see also \cite{bartholdi2017wordAndOrderProblems}). On the other hand, the undecidability result for the finiteness problem for automaton semigroups \cite[Theorem~3.13]{gillibert2014finiteness} could not be lifted to automaton groups so far.

    The undecidability results show that the presentation of groups using automata is still quite powerful. Nevertheless, it is not very difficult to see that the word problem for automaton groups is decidable. One possible way is to show an upper bound on the length of an input word on which a state sequence\footnote{In order to avoid ambiguities, we call a word over the states of the automaton a \emph{state sequence}. So, in our case the input for the word problem is a state sequence.} not representing the identity of the group acts non-trivially. In the most general setting, this bound is $|Q|^n$ where $Q$ is the state set of the automaton and $n$ is the length of the state sequence. 
    Another viewpoint is that one can use a non-deterministic guess and check algorithm to solve the word problem. This algorithm uses linear space proving that the word problem for automaton (semi)groups is in \PSPACE. This approach seems to be mentioned first by Steinberg \cite[Section~3]{steinberg2015some} (see also \cite[Proposition~2 and~3]{dangeli2017complexity}).
    In some special cases, better algorithms or upper bounds are known: for example, for contracting automaton groups (and this includes the Grigorchuk group), the witness length is bounded logarithmically \cite{nekrashevych2005self} and the problem, thus, is in \LOGSPACE; other examples of classes with better upper bounds or algorithms include automata with polynomial activity \cite{bondarenko2012growth} or Hanoi Tower groups \cite{bondarenko2014wordProblem}.
    On the other hand, Steinberg conjectured that there is an automaton group with a \PSPACE-complete word problem \cite[Question~5]{steinberg2015some}. As a partial solution to his problem, an inverse automaton semigroup with a \PSPACE-complete word problem has been constructed in \cite[Proposition~6]{dangeli2017complexity}\footnote{In fact, the semigroup is generated by a partial, invertible automaton. A priori, this seems to be a stronger statement than that the semigroup is inverse and also an automaton semigroup. That is why the cited paper uses the term \enquote{automaton-inverse semigroup}. Only later, it was shown that the two concepts actually coincide \cite[Theorem~25]{structurePart}.}. In this paper, our aim is to finally prove the conjecture for groups.
    
    First, however, we give a simpler proof for the weaker statement that the uniform word problem for automaton groups (where the group, represented by its generating automaton, is part of the input) is $\PSPACE$-complete in \cref{sec:uniformWP}. This simpler proof uses the same ideas as the main proof and should facilitate understanding the latter.

    For the main result, \cref{thm:nonuniformPSPACE}, we adopt the construction used by D'Angeli, Rodaro and the first author from \cite[Proposition~6]{dangeli2017complexity}. This construction uses a master reduction and directly encodes a Turing machine into an automaton. Already in \cite[Proposition~6]{dangeli2017complexity}, it was also used to show that there is an automaton group whose word problem with a rational constraint (which depends on the input) is \PSPACE-complete. To get rid of this rational constraint, we apply an idea used by Barrington \cite{barrington89boundedWidth} to transform $\NC^1$-circuits (circuits of bounded fan-in and logarithmic depth) into bounded-width polynomial-size branching programs. Similar ideas predating Barrington have been attributed to Gurevich (see \cite{mak84}) and given by Mal'cev \cite{malcev62} but also by Mauerer and Rhodes \cite{maurer1965property} as well as Krohn, Maurer and Rhodes \cite{krohn1966realizing}. Nevertheless, this paper is fully self-contained and no previous knowledge of either \cite{dangeli2017complexity} or \cite{barrington89boundedWidth} is needed. Barrington's construction uses the group $A_5$ of even permutations over five elements. However, there is a wide variety of other groups that can be used in a similar way. For example the Grigorchuk group (see \cite{BartholdiFLW20}) or the free group in three generators (see \cite{Robinson93phd}). Both examples have the advantage that they are -- in contrast to $A_5$ -- generated by an automaton over a binary alphabet (see \cite{Aleshin83,VorobetsV07} for the free group). We will describe our construction in a general way (independent of the group). Afterwards, we can plug in, for example, the mentioned free group in three generators and obtain an automaton group over a binary alphabet with a \PSPACE-complete word problem.

    Finally, in \cref{sec:compressedWP}, we also investigate the compressed word problem for automaton groups. Here, the (input) state sequence is given as a so-called \emph{straight-line program} (a context-free grammar which generates exactly one word). Alternatively, the compressed word problem can also be considered as a circuit value problem where the input gates of the circuit are labeled by group elements and the inner gates compute the product of group elements. See \cite{lohrey2014compressed} or \cite[Chapter~4]{bassino2020complexity} for more information on the compressed word problem. By uncompressing the input sequence and applying the above mentioned non-deterministic linear-space algorithm, one can see that the compressed word problem can be solved in \EXPSPACE. Thus, the more interesting part is to prove that this algorithm cannot be improved significantly: we show that there is an automaton group with an \EXPSPACE-hard compressed word problem. This result is interesting because, by taking the disjoint union of the two automata, we obtain a group whose (ordinary) word problem is \PSPACE-complete and whose compressed word problem is \EXPSPACE-complete and, thus, provably more difficult by the space hierarchy theorem \cite[Theorem~6]{stearns1965hierarchies} (or e.\,g.\ \cite[Theorem~7.2, p.~145]{papadimitriou97computational} or \cite[Theorem~4.8]{arora2009computational}). To the best of our knowledge, this was the first known example of such a group. Meanwhile, however, Bartholdi, Figelius, Lohrey and the second author found other examples of groups whose compressed word problem is provably harder than their (ordinary) word problem \cite{BartholdiFLW20}. These example include the Grigorchuk group, Thompsons' groups and several others defined in terms of wreath products.
    
    Other explicit previous results on the compressed word problem for automaton groups do not seem to exist. However, it was observed by Gillibert \cite{gillibert2019personal} that the proof of \cite[Proposition~6]{dangeli2017complexity} also yields an automaton semigroup with an \EXPSPACE-complete compressed word problem in a rather straightforward manner. For the case of groups, it is possible to adapt the construction used by Gillibert to prove the existence of an automaton group with an undecidable order problem \cite{gillibert2018automaton} slightly to obtain an automaton group with a \PSPACE-hard compressed word problem \cite{gillibert2019personal}.
    
%    \paragraph{Acknowledgements.}
    This work is an extended version of the conference paper \cite{WachterW20}. The most notable extensions are the simpler proof for the uniform word problem (which was omitted from the conference paper for space reasons) and the extension of the construction to use a binary alphabet instead of one with five elements (which is novel).
    Parts of the presentation are taken from the first author's doctoral thesis \cite{Wachter20diss}.
  \end{section}

  \begin{section}{Preliminaries}
    \paragraph{Words and Alphabets with Involution.}
    In this paper, we use $A \uplus B$ for the disjoint union of the sets $A$ and $B$. Additionally, we use common notations from formal language theory. In particular, we use $\Sigma^*$ to denote the set of words over an alphabet $\Sigma$ including the empty word $\varepsilon$. If we want to exclude the empty word, we write $\Sigma^+$. For any alphabet $Q$, we define a natural involution between $Q$ and a disjoint copy $Q^{-1} = \{ q^{-1} \mid q \in Q \}$ of $Q$: it maps $q \in Q$ to $q^{-1} \in Q^{-1}$ and vice versa. In particular, we have $(q^{-1})^{-1} = q$. The involution extends naturally to words over $Q^{\pm 1} = Q \cup Q^{-1}$: for $q_1, \dots, q_n \in Q^{\pm 1}$, we set $(q_n \dots q_1)^{-1} = q_1^{-1} \dots q_n^{-1}$. This way, the involution is equivalent to taking the group inverse if $Q$ is a generating set of a group. To simplify the notation, we write $Q^{\pm *}$ for $(Q^{\pm 1})^*$.
    
    A word $u$ is a \emph{prefix} of a word $v$ if there is some word $x$ with $v = ux$. It is a \emph{proper} prefix if $x$ is non-empty. The set of proper prefixes of some word $w$ is $\PPre w$ and $\PPre L$ for some language $L$ is $\PPre L = \bigcup_{w \in L} \PPre w$.
    
    \paragraph{Turing Machines and Complexity.}
    We assume the reader to be familiar with basic notions of complexity theory (see \cite{papadimitriou97computational} or \cite{arora2009computational} for standard text books on complexity theory) such as configurations for Turing machines, computations and reductions in logarithmic space (\LOGSPACE) as well as complete and hard problems for \PSPACE and the class \EXPSPACE. For the class of problems (or functions) solvable (or computable) in deterministic linear space, we write \DLINSPACE. 
    We only consider deterministic, single-tape machines and write their configurations as words $c_0 \dots c_{i - 1} p c_i \dots c_{n - 1}$ where the $c_j$ are symbols from the tape alphabet and $p$ is a state. In this configuration, the machine is in state $p$ and its head is over the symbol $c_i$. By convention, we assume that the tape symbols at positions $-1$ and $n$ are blank symbols.
    
    Using suitable normalizations, we can assume that every Turing machine admits a simple function which describes its transitions:
    \begin{fact}[Folklore]\labelx{fct:turingNormalization}
      Consider a deterministic Turing machine with state set $P$ and tape alphabet $\Delta$. After a straightforward transformation of the transition function and states, we can assume that the symbol $\gamma_i^{(t + 1)}$ at position $i$ of the configuration at time step $t + 1$ only depends on the symbols $\gamma_{i - 1}^{(t)}, \gamma_i^{(t)}, \gamma_{i + 1}^{(t)} \in \Gamma$ at position $i - 1$, $i$ and $i + 1$ at time step $t$. Thus, we may always assume that there is a function $\tau: (P \uplus \Delta)^3 \to P \uplus \Delta$ mapping the symbols $\gamma_{i - 1}^{(t)}, \gamma_i^{(t)}, \gamma_{i + 1}^{(t)} \in P \uplus \Delta$ to the uniquely determined symbol $\gamma_i^{(t + 1)}$ for all $i$ and $t$.
    \end{fact}
    \begin{proof}[Proof Idea]
      The only problem appears if the machine moves to the left: if we have the situation $abp\textcolor{gray}{c}$ or $abp\textcolor{gray}{d}$ (where the positions $i-1$, $i$ and $i + 1$ are in black) and the machine moves to the left in state $p$ when reading a $c$ but does not move when reading a $d$, then the new value for the second symbol does not only depend on the symbols right next to it; we can either be in the situation $\textcolor{gray}{a}p'\textcolor{gray}{bc'}$ or $\textcolor{gray}{a}b\textcolor{gray}{p'd'}$ (where position $i$ is in black). To circumvent the problem, we can introduce intermediate states. Now, instead of moving to the left, we go into an intermediate state (without movement). In the next step, we move to the left (but this time the movement only depends on the state and not on the current symbol).\qed
    \end{proof}
    
    \paragraph{Automata.}
    We use the word \emph{automaton} to denote what is more precisely called a letter-to-letter, finite state transducer. Formally, an automaton is a triple $\mathcal{T} = (Q, \Sigma, \delta)$ consisting of a finite set of \emph{states} $Q$, an input and output alphabet $\Sigma$ and a set $\delta \subseteq Q \times \Sigma \times \Sigma \times Q$ of \emph{transitions}. For a transition $(p, a, b, q) \in Q \times \Sigma \times \Sigma \times Q$, we usually use the more graphical notation $\trans{p}{a}{b}{q}$ and, additionally, the common way of depicting automata
    
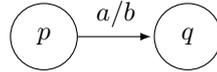
\begin{center}
      \begin{tikzpicture}[auto, shorten >=1pt, >=latex]
        \node[state] (p) {$p$};
        \node[state, right=of p] (q) {$q$};
        
        \draw[->] (p) edge node {$a / b$} (q);
      \end{tikzpicture}
    \end{center}
    where $a$ is the \emph{input} and $b$ is the \emph{output}.
    We will usually work with \emph{deterministic} and \emph{complete} automata, i.\,e.\ automata where we have
    \[
      d_{p, a} = \left| \{ \trans{p}{a}{b}{q} \in \delta \mid b \in \Sigma, q \in Q \} \right| = 1
    \]
    for all $p \in Q$ and $a \in \Sigma$. In other words, for every $a \in \Sigma$, every state has exactly one transition with input $a$.
    
    A \emph{run} of an automaton $\mathcal{T} = (Q, \Sigma, \delta)$ is a sequence 
    \begin{center}
      \begin{tikzpicture}[auto, shorten >=1pt, >=latex, baseline=(q0.base)]
        \node (q0) {$q_0$};
        \node[right=of q0] (q1) {$q_1$};
        \node[right=of q1] (dots) {$\dots$};
        \node[right=of dots] (qn) {$q_n$};
        
        \path[->] (q0) edge node {$a_1 / b_1$} (q1)
                  (q1) edge node {$a_2 / b_2$} (dots)
                  (dots) edge node {$a_n / b_n$} (qn)
        ;
      \end{tikzpicture}
    \end{center}
    of transitions from $\delta$. It \emph{starts} in $q_0$ and \emph{ends} in $q_n$. Its \emph{input} is $a_1 \dots a_n$ and its \emph{output} is $b_1 \dots b_n$. If $\mathcal{T}$ is complete and deterministic, then, for every state $q \in Q$ and every word $u \in \Sigma^*$, there is exactly one run starting in $q$ with input $u$. We write $q \circ u$ for its output and $q \cdot u$ for the state in which it ends. This notation can be extended to multiple states. To avoid confusion, we usually use the term \emph{state sequence} instead of \enquote{word} (which we reserve for input or output words) for elements $\bm{q} \in Q^*$. Now, for states $q_1, q_2, \dots, q_\ell \in Q$, we set $q_\ell \dots q_2 q_1 \circ u = q_\ell \dots q_2 \circ (q_1 \circ u)$ inductively. If the state sequence $\bm{q} \in Q^*$ is empty, then $\bm{q} \circ u$ is simply $u$.
    
    This way, every state $q \in Q$ (and even every state sequence $\bm{q} \in Q^*$) induces a map $\Sigma^* \to \Sigma^*$ and every word $u \in \Sigma^*$ induces a map $Q \to Q$. If all states of an automaton induce bijective functions, we say it is \emph{invertible} and call it a \emph{\GAut}. In this case, we let the function induced by the inverse $q^{-1}$ of a state $q$ be the inverse of the function induced by $q$. For a \GAut $\mathcal{T}$, all bijections induced by the states generate a group (with composition as operation), which we denote by $\mathscr{G}(\mathcal{T})$. A group is called an \emph{automaton group} if it arises in this way. Clearly, $\mathscr{G}(\mathcal{T})$ is generated by the maps induced by the states of $\mathcal{T}$ and, thus, finitely generated.
    
    \begin{example}\labelx{ex:addingMachine}
      The typical first example of an automaton generating a group is the \emph{adding machine} $\mathcal{T} = (\{ q, \id \}, \{ 0, 1 \}, \delta)$:
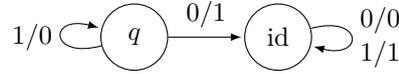
      \begin{center}
        \marginbox{0pt 1ex 0pt 1ex}{%
        \begin{tikzpicture}[auto, shorten >=1pt, >=latex, baseline=(id.base)]
          \node[state] (q) {$q$};
          \node[state, right=of q] (id) {$\id$};
          \draw[->] (q) edge[loop left] node {$1/0$} (q)
                        edge node {$0/1$} (id)
                    (id) edge[loop right] node[align=left] {$0/0$\\$1/1$} (id);
        \end{tikzpicture}}
      \end{center}
      It obviously is deterministic and complete and, therefore, we can consider the map induced by the state $q$. We have $q^3 \circ 000 = q^2 \circ 100 = q \circ 010 = 110$. From this example, it is easy to see that the action of $q$ is to increment the input word (which is interpreted as a reverse/least significant bit first binary representation $\revbin(n)$ of a number $n$). The inverse is accordingly to decrement the value. As the other state $\id$ acts like the identity, we obtain that the group $\mathscr{G}(\mathcal{T})$ generated by $\mathcal{T}$ is isomorphic to the infinite cyclic group.
    \end{example}
    
    Similar to extending the notation $\bm{q} \circ u$ to state sequences, we can also extend the notation $q \cdot u$. For this, it is useful to introduce \emph{cross diagrams}, another notation for transitions of automata. For a transition $\trans{p}{a}{b}{q}$ of an automaton, we write the cross diagram given in \cref{fig:crossDiagram}. Multiple cross diagrams can be combined into a larger one. For example, the cross diagram in \cref{fig:combinedCrossDiagram} indicates that there is a transition $\trans{q_{i,j-1}}{a_{i - 1, j}}{a_{i,j}}{q_{i,j}}$ for all $1 \leq i \leq n$ and $1 \leq j \leq m$. Typically, we omit unneeded names for states and abbreviate cross diagrams. Such an abbreviated cross diagram is depicted in \cref{fig:abbreviatedCrossDiagram}. If we set $\bm{p} = q_{n, 0} \dots q_{1, 0}$, $u = a_{0, 1} \dots a_{0, m}$, $v = a_{n, 1} \dots a_{n, m}$ and $\bm{q} = q_{n, m} \dots q_{1, m}$, then it indicates the same transitions as in \cref{fig:combinedCrossDiagram}. It is important to note here, that the right-most state in $\bm{p}$ is actually the one to act first.
    \begin{figure}%
      \centering%
      \begin{subfigure}[t]{0.15\linewidth}%
        \centering%
        \begin{tikzpicture}[baseline=(m-2-1.base)]
          \matrix[matrix of math nodes, text height=1.25ex, text depth=0.25ex] (m) {
            & a & \\
            p & & q \\
            & b & \\
          };
          \foreach \i in {1} {
            \draw[->] let
              \n1 = {int(2+\i)}
            in
              (m-2-\i) -> (m-2-\n1);
            \draw[->] let
              \n1 = {int(1+\i)}
            in
              (m-1-\n1) -> (m-3-\n1);
          };
        \end{tikzpicture}%
        \caption{Cross diagrams}\label{fig:crossDiagram}%
      \end{subfigure}\hfill%
      \begin{subfigure}[t]{0.6\linewidth}%
        \centering%
        \begin{tikzpicture}[baseline=(m-4-1.east)]
          \matrix[matrix of math nodes, text height=1.25ex, text depth=0.25ex, ampersand replacement=\&] (m) {
                     \& a_{0, 1}     \&          \& \dots \&              \& a_{0, m}     \&     \\
            q_{1, 0} \&              \& q_{1, 1} \& \dots \& q_{1, m - 1} \&              \& q_{1, m} \\
                     \& a_{1, 1}     \&          \&       \&              \& a_{1, m}     \&     \\
              \vdots \& \vdots       \&          \&       \&              \& \vdots       \& \vdots \\
                     \& a_{n - 1, 1} \&          \&       \&              \& a_{n - 1, m} \&     \\
            q_{n, 0} \&              \& q_{n, 1} \& \dots \& q_{n, m - 1} \&              \& q_{n, m} \\
                     \& a_{n, 1}     \&          \& \dots \&              \& a_{n, m}     \&     \\
          };
          \foreach \j in {1, 5} {
            \foreach \i in {1, 5} {
              \draw[->] let
                \n1 = {int(2+\i)},
                \n2 = {int(1+\j)}
              in
                (m-\n2-\i) -> (m-\n2-\n1);
              \draw[->] let
                \n1 = {int(1+\i)},
                \n2 = {int(2+\j)}
              in
                (m-\j-\n1) -> (m-\n2-\n1);
            };
          };
        \end{tikzpicture}%
        \caption{Multiple crosses combined in one diagram}\label{fig:combinedCrossDiagram}%
      \end{subfigure}\hfill%
      \begin{subfigure}[t]{0.15\linewidth}%
        \centering%
        \begin{tikzpicture}[baseline=(m-2-1.base)]
          \matrix[matrix of math nodes, text height=1.25ex, text depth=0.25ex, ampersand replacement=\&] (m) {
                   \& u \& \\
            \bm{p} \&   \& \bm{q} \\
                   \& v \& \\
          };
          \foreach \i in {1} {
            \draw[->] let
              \n1 = {int(2+\i)}
            in
              (m-2-\i) -> (m-2-\n1);
            \draw[->] let
              \n1 = {int(1+\i)}
            in
              (m-1-\n1) -> (m-3-\n1);
          };
        \end{tikzpicture}%
        \caption{Abbreviated cross diagram}\label{fig:abbreviatedCrossDiagram}%
      \end{subfigure}
      \caption{Combined and abbreviated cross diagrams}
    \end{figure}
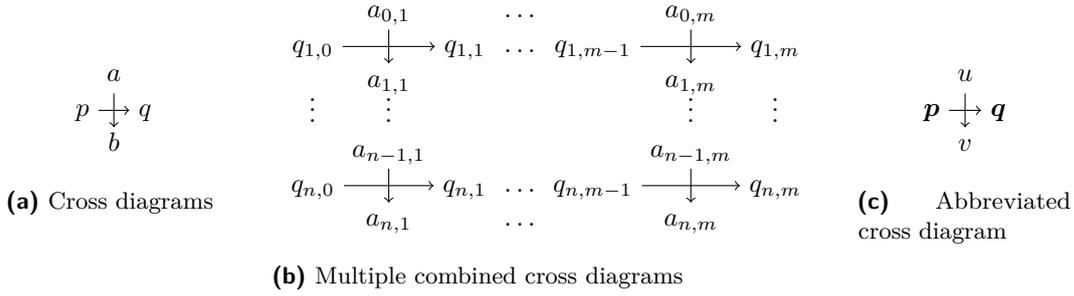
    
    If we have the cross diagram from \cref{fig:abbreviatedCrossDiagram}, we set $\bm{p} \cdot u = \bm{q}$. This is the same, as setting $q_n \ldots q_1 \cdot u = [q_n \dots q_2 \cdot (q_1 \circ u)] (q_1 \cdot u)$ inductively and, with the definition from above, we already have $\bm{p} \circ u = v$.
    
    \paragraph{Group Theory and Word Problem.}
    For the neutral element of a group, we write $\idGrp$. We write $\bm{p} =_G \bm{q}$ or $\bm{p} = \bm{q}$ in $G$ if two words $\bm{p}$ and $\bm{q}$ over the generators (and their inverses) of a group evaluate to the same group element. Typically, $G$ will be an automaton group generated by a \GAut $\mathcal{T} = (Q, \Sigma, \delta)$ and $\bm{p}$ and $\bm{q}$ are state sequences from $Q^{\pm *}$.
    
    The \emph{word problem} of a group $G$ generated by a finite set $Q$ is the decision problem:
    \problem
      [the group $G$]
      {$\bm{q} \in Q^{\pm *}$}
      {is $\bm{q} = \idGrp$ in $G$?}\noindent
    In addition, if $\mathcal{C}$ is a class of groups, we also consider the  \emph{uniform} word problem for $\mathcal{C}$. Here, the group $G \in \mathcal{C}$ is part of the input (in a suitable representation).
    
    \paragraph{Balanced Iterated Commutators.}
    For elements $h$ and $g$ of a group $G$, we write $g^h$ for the conjugation $h^{-1} g h$ of $g$ with $h$ and $[h, g]$ for the commutator $h^{-1} g^{-1} hg$. Both notations extend naturally to words over the group generators: if $G$ is generated by $Q$, then we let $\bm{p}^{\bm{q}} = \bm{q}^{-1} \bm{p} \bm{q}$ and $[\bm{q}, \bm{p}] = \bm{q}^{-1} \bm{p}^{-1} \bm{q} \bm{p}$ for $\bm{p}, \bm{q} \in Q^{\pm *}$.
    
    Commutators can be used to simulate logical conjunctions in groups since we have $[g, h] = \idGrp$ if $g = \idGrp$ or $h = \idGrp$. To create a $D$-ary logical conjunction, we can nest multiple binary conjunctions in a tree of logarithmic depth and we can use the same idea with commutators.\footnote{The usage of commutators to compute Boolean functions has been formalized by the notion of $G$-programs (see \cite{barrington89boundedWidth}). Here, we do not use this formalism because it cannot be applied to the compressed word problem.}
    
    \begin{definition}\labelx{def:balancedCommutator}
      Let $Q$ be an alphabet and $\alpha, \beta: \mathbb{N} \to Q^{\pm *}$. For $\bm{p}_0, \dots, \bm{p}\sub{D - 1} \in Q^{\pm*}$ with $D = 2^d$, we define the word $B_{\beta, \alpha}[\bm{p}\sub{D - 1}, \dots, \bm{p}_0]$ inductively on $d$ by
      \begin{align*}
        B_{\beta, \alpha}[\bm{p}_0] &= \bm{p}_0 \\
        B_{\beta, \alpha}[\bm{p}_{2D - 1}, \dots, \bm{p}_0] &= \left[ B_{\beta, \alpha}[\bm{p}_{2D - 1}, \dots, \bm{p}\sub{D}]^{\beta(d)},\: B_{\beta, \alpha}[\bm{p}\sub{D - 1}, \dots, \bm{p}_{0}]^{\alpha(d)} \right] \text{.}
      \end{align*}
    \end{definition}
    We also write $B_{\beta, \alpha}$ for $\alpha \in Q^{\pm *}$ or $\beta \in Q^{\pm *}$, where we identify $\alpha$ and $\beta$ with the constant maps $n \mapsto \alpha$ and $n \mapsto \beta$, respectively.
    
    One part of using $B_{\beta, \alpha}[\bm{p}\sub{D - 1}, \dots, \bm{p}_0]$ as a logical conjunction is that it collapses to the neutral element if one of the $\bm{p}_i$ is equal to the neutral element.
    \begin{fact}\labelx{fct:commutatorCollapses}
      Let $G$ be a group generated by some (finite) set $Q$, $\alpha, \beta: \mathbb{N} \to Q^{\pm *}$ and $\bm{p}_0, \dots, \bm{p}\sub{D - 1} \in Q^{\pm*}$ for some $D = 2^d$. If there is some $i$ with $\bm{p}_i = \idGrp$ in $G$, we have $B_{\beta, \alpha}[\bm{p}\sub{D - 1}, \dots, \bm{p}_0] = \idGrp$ in $G$.
    \end{fact}
    \begin{proof}
      We show the fact by induction on $d$. For $d = 0$ (or, equivalently, $D = 1$), there is nothing to show. So, consider the step from $d$ to $d + 1$ (or, equivalently, from $D$ to $2D$). If we have $0 \leq i < D$, then, by induction, we have
      \begin{align*}
        B_{\beta, \alpha}[\bm{p}_{2D - 1}, \dots, \bm{p}_0]
        &= \bigl[ 
          {\underbrace{B_{\beta, \alpha}[\bm{p}_{2D - 1}, \dots, \bm{p}\sub{D} ]}_{{}= \bm{p}'}}^{\beta(d)},
          {\underbrace{B_{\beta, \alpha}[\bm{p}\sub{D - 1}, \dots, \bm{p}_{0}]}_{{} =_G \, \idGrp}}^{\alpha(d)} \bigr] \\
        &=_G
          \begin{array}[t]{c@{\;}c@{}l}
            \beta(d)^{-1} \bm{p}'^{-1} \beta(d) & \alpha(d)^{-1} \idGrp \alpha(d) \\ \beta(d)^{-1} \bm{p}' \beta(d) & \alpha(d)^{-1} \idGrp \alpha(d) &{}=_G \idGrp \text{.}
          \end{array}
      \end{align*}
      The case $D \leq i < 2D$ is symmetric.\qed
    \end{proof}
    
    The other part of using $B_{\beta, \alpha}[\bm{p}\sub{D - 1}, \dots, \bm{p}_0]$ as a logical conjunction~-- namely that $B_{\beta, \alpha}[\bm{p}\sub{D - 1}, \dots, \bm{p}_0]$  is nontrivial for proper choices of the $\bm{p}_{i}$ -- depends on the actual underlying group.
    \begin{example}\labelx{ex:conjunctionInA5}
      A simple way to use $B_{\beta, \alpha}$ as a logical conjunction is the group $A_5$ of even permutations over five elements. It was used by Barrington \cite{barrington89boundedWidth} to convert logical circuits of bounded fan-in and logarithmic depth (so-called \emph{$\NC^1$-circuits}) to bounded-width, polynomial-size branching programs.
      
      Since $A_5$ is finite, we can use the entire group as a finite generating set. We let $\sigma = (1 3 2 5 4)$, $\alpha = (2 3) (4 5)$ and $\beta = (2 4 5)$. A straightforward calculation shows that we have
      \[
        \sigma = [ \sigma^\beta, \sigma^\alpha ]
      \]
      for this choice (compare to \cite[Lemma~1 and~3]{barrington89boundedWidth}), which allows us to use $B_{\beta, \alpha}$ as a $D$-ary logical conjunction. In fact, it is even quite simple, because $\alpha$ and $\beta$ are constant (and do not depend on the level $d$).
      
      The idea is that we consider $\idGrp$ as \texttt{false} and $\sigma$ as \texttt{true}. Now, to use $B_{\beta, \alpha}$ as a logical conjunction, we show
      \[
       B_{\beta, \alpha}[g\sub{D - 1}, \dots, g_0] \eqA
       \begin{cases}
         \sigma & \text{if } g_0 = \dots = g\sub{D - 1} = \sigma\\
         \idGrp & \text{otherwise}
       \end{cases}
      \]
      for all $g_0, \dots, g\sub{D - 1} \in \{ \id, \sigma \} \subseteq A_5$ where $D = 2^d$. We show the case $g_0 = \dots = g\sub{D - 1} = \sigma$ by induction on $d$ and the case $g_i = \idGrp$ for some $i$ follows by \cref{fct:commutatorCollapses}. For $d = 0$ (or, equivalently, $D = 1$), there is nothing to show. So, consider the step from $d$ to $d + 1$ (or, equivalently, from $D$ to $2D$), where we have
      \[
        B_{\beta, \alpha}[g_{2D - 1}, \dots, g_0] = \bigl[ {\underbrace{B_{\beta, \alpha}[g_{2D - 1}, \dots, g\sub{D}]}_{{} \eqA \sigma}}^\beta, {\underbrace{B_{\beta, \alpha}[g\sub{D - 1}, \dots, g_{0}]}_{{} \eqA \sigma}}^\alpha \bigr] \eqA \sigma
      \]
      by induction and the choice of $\sigma$, $\alpha$ and $\beta$.
    \end{example}
    
    Normal commutators and conjugation are compatible in the sense that we have $[h, g]^k = [h^k, g^k]$ for all elements $g, h, k$ of a group. The commutators from \cref{def:balancedCommutator} satisfy a similar compatibility that we will exploit below.
    \begin{fact}\labelx{fct:conjugatedCommutator}
      Let $G$ be a group generated by some (finite) set $Q$, $\alpha, \beta: \mathbb{N} \to Q^{\pm *}$ and $\gamma \in Q^{\pm *}$ such that $\gamma$ commutes with $\alpha(d)$ and $\beta(d)$ in $G$ for all $d$. Then, we have for all $\bm{p}_0, \dots, \bm{p}\sub{D - 1} \in Q^{\pm *}$ and with $D = 2^d$
      \[
        B_{\beta, \alpha}[\bm{p}\sub{D - 1}, \dots, \bm{p}_0]^\gamma = B_{\beta, \alpha}[\bm{p}\sub{D - 1}^\gamma, \dots, \bm{p}_0^\gamma] \text{ in } G \text{.}
      \]
    \end{fact}
    \begin{proof}
      We prove the statement by induction on $d$. For $d = 0$ (or, equivalently, $D = 1$), we have $B_{\beta, \alpha}[\bm{p}_0]^\gamma = \bm{p}_0^\gamma = B_{\beta, \alpha}[\bm{p}_0^\gamma]$ and, for the step from $d$ to $d + 1$ (or, equivalently, from $D$ to $2D$), we have in $G$:
      \begin{multline*}
        B_{\beta, \alpha}[\bm{p}_{2D - 1}, \dots, \bm{p}_0]^\gamma \\
        \begin{alignedat}[t]{2}
          &= \left[ B_{\beta, \alpha}[\bm{p}_{2D - 1}, \dots, \bm{p}\sub{D}]^{\beta(d)},\: B_{\beta, \alpha}[\bm{p}\sub{D - 1}, \dots, \bm{p}_{0}]^{\alpha(d)} \right]^\gamma &\text{(by definition)} \\
          &= \left[ B_{\beta, \alpha}[\bm{p}_{2D - 1}, \dots, \bm{p}\sub{D}]^{\beta(d) \gamma},\: B_{\beta, \alpha}[\bm{p}\sub{D - 1}, \dots, \bm{p}_{0}]^{\alpha(d) \gamma} \right] &\quad \text{($[h, g]^k = [h^k, g^k]$)} \\
          &= \left[ B_{\beta, \alpha}[\bm{p}_{2D - 1}, \dots, \bm{p}\sub{D}]^{\gamma \beta(d)},\: B_{\beta, \alpha}[\bm{p}\sub{D - 1}, \dots, \bm{p}_{0}]^{\gamma \alpha(d)} \right] \\
          &&\llap{\text{($\gamma$ comutes with $\alpha(d), \beta(d)$)}} \\
          &= \left[ B_{\beta, \alpha}[\bm{p}_{2D - 1}^\gamma, \dots, \bm{p}\sub{D}^\gamma]^{\beta(d)},\: B_{\beta, \alpha}[\bm{p}\sub{D - 1}^\gamma, \dots, \bm{p}_{0}^\gamma]^{\alpha(d)} \right] & \text{(by induction)} \\
          &= B_{\beta, \alpha}[\bm{p}_{2D - 1}^\gamma, \dots, \bm{p}_0^\gamma] &\text{\qed}
        \end{alignedat}
      \end{multline*}
    \end{proof}
    
    For our later reductions, it will be important to compute $B_{\beta, \alpha}[\bm{p}\sub{D - 1}, \dots, \bm{p}_0]$ in logarithmic space, which is possible due to the logarithmic nesting depth.
    \begin{lemma}\labelx{fct:BIsLogspaceComputable}
      If the functions $\alpha$ and $\beta$ are computable in \DLINSPACE (where the input is given in binary), we can compute $B_{\beta, \alpha}[\bm{p}\sub{D - 1}, \dots, \bm{p}_0]$ on input of $\bm{p}_0, \dots, \bm{p}\sub{D - 1}$ in logarithmic space.
    \end{lemma}
    \begin{proof}
      We give a sketch for a (deterministic) algorithm which computes the symbol at position $i$ of $B_{\beta, \alpha}[\bm{p}\sub{D - 1}, \dots, \bm{p}_0]$ (as a word where we consider $\alpha(d)$, $\beta(d)$, the $\bm{p}_i$ and their inverses as letters) in logarithmic space. Later, we can expand this by filling in the actual state sequences over $Q^{\pm *}$ (which can be computed in space logarithmic in $D$). For $D = 2^d$ (with $d > 0$), we have
      \[
        B_{\beta, \alpha}[\bm{p}\sub{D - 1}, \dots, \bm{p}_0] =
          \begin{array}[t]{r@{\,}c@{\,}l}
            \beta(d - 1)^{-1} & B_{\beta, \alpha}[\bm{p}\sub{D - 1}, \dots, \bm{p}_{\frac{D}{2}}]^{-1} & \beta(d - 1) \\
            \alpha(d - 1)^{-1} & B_{\beta, \alpha}[\bm{p}_{\frac{D}{2} - 1}, \dots, \bm{p}_{0}]^{-1} & \alpha(d - 1) \\
            \beta(d - 1)^{-1} & B_{\beta, \alpha}[\bm{p}\sub{D - 1}, \dots, \bm{p}_{\frac{D}{2}}] & \beta(d - 1) \\
            \alpha(d - 1)^{-1} & B_{\beta, \alpha}[\bm{p}_{\frac{D}{2} - 1}, \dots, \bm{p}_{0}] & \alpha(d - 1)
          \end{array}
      \]
      and the length $\ell(D)$ of $B_{\beta, \alpha}[\bm{p}\sub{D - 1}, \dots, \bm{p}_0]$ (again as a word where we consider $\alpha(d)$, $\beta(d)$, the $\bm{p}_i$ and their inverses as letters) is given by $\ell(1) = 1$ and $\ell(D) = 8 + 4 \ell(\frac{D}{2})$. This yields
      \[
        \ell(D) = \left( \sum_{i = 0}^{d - 1} 4^i \cdot 8 \right) + 4^d = 8 \frac{4^d - 1}{3} + 4^d = \frac{8}{3} (D^2 - 1) + D^2 = \frac{11}{3} D^2 - \frac{8}{3}
      \]
      and, thus, that the length of $B_{\beta, \alpha}[\bm{p}\sub{D - 1}, \dots, \bm{p}_0]$ is polynomial in $D$. Therefore, we can iterate the algorithm for all positions $1 \leq i \leq \ell(D)$ to output $B_{\beta, \alpha}[\bm{p}\sub{D - 1}, \dots, \bm{p}_0]$ entirely.
      
      To compute the symbol at position $i$, we first check whether $i$ is the first or last position (notice that we need the exact value of $\ell(D)$ for testing the latter). In this case, we know that it is $\beta(d - 1)^{-1}$ or $\alpha(d - 1)$, respectively. Similarly, we can do this for the positions in the middle and at one or three quarters. If the position falls into one of the four recursion blocks, we use two pointers into the input: \texttt{left} and \texttt{right}. Depending on the block, \texttt{left} and \texttt{right} either point to $\bm{p}_0$ and $\bm{p}_{\frac{D}{2} - 1}$ or to $\bm{p}_{\frac{D}{2}}$ and $\bm{p}\sub{D - 1}$. Additionally, we also store whether we are in an inverse block or a non-inverse block and keep track of $d$ as a binary number. We need to decrease $d$ by one in every recursion step. Note that $d = \log D$ as a binary number has length $\log \log D$ (up to constants) and can, thus, certainly be stored in logarithmic space.\footnote{In fact, it would suffice if $\alpha$ and $\beta$ were computable in space $2^n$ here. However, we will not need this more general statement and the hypothesis is not sufficient when we consider the compressed word problem later on.} From now on, we disregard the input left of \texttt{left} and right of \texttt{right} (and do appropriate arithmetic on $i$) and can proceed recursively. If we need to perform another recursive step, we update the variables \texttt{left} and \texttt{right} (instead of using new ones). Therefore, the whole recursion can be done in logarithmic space.\qed
    \end{proof}
    
    Normally, we cannot simply apply $B_{\beta, \alpha}$ to a cross diagram as the output interferes and we could get into different states. However, it is possible if the rows of the cross diagram act like the identity (as we then can clearly re-order rows without interference):
    \begin{fact}\labelx{fct:commutatorInCrossDiagrams}
      Let $\mathcal{T} = (Q, \Sigma, \delta)$ be some \GAut, $u \in \Sigma^*$, $\alpha, \beta: \mathbb{N} \to Q^{\pm *}$ and $\bm{p}_0, \dots, \bm{p}\sub{D - 1} \in Q^{\pm *}$ be state sequences with $D = 2^d$ such that we have the cross diagrams
      \begin{center}
        \begin{tikzpicture}
          \matrix[matrix of math nodes, text height=1.25ex, text depth=0.25ex] (m) {
                     & u & \\
            \bm{p}_0 &   & \bm{q}_0 \\
                     & u & \\
             \vdots  & \vdots  & \vdots \\
                     & u & \\
            \bm{p}\sub{D - 1} &   & \bm{q}\sub{D - 1} \\
                     & u & \\
          };
          
          \foreach \j in {1, 5} {
            \foreach \i in {1} {
              \draw[->] let
                \n1 = {int(2+\i)},
                \n2 = {int(1+\j)}
              in
                (m-\n2-\i) -> (m-\n2-\n1);
              \draw[->] let
                \n1 = {int(1+\i)},
                \n2 = {int(2+\j)}
              in
                (m-\j-\n1) -> (m-\n2-\n1);
            };
          };
          
          \matrix[matrix of math nodes, text height=1.25ex, text depth=0.25ex, right=4cm of m.north, anchor=north] (m2) {
                   & u & \\
            \alpha(d) &   & \alpha'(d) \\
                   & u & \\
          };
          \draw[->] (m2-2-1) -> (m2-2-3)
                    (m2-1-2) -> (m2-3-2);

          \matrix[matrix of math nodes, text height=1.25ex, text depth=0.25ex, right=4cm of m.south, anchor=south] (m3) {
                  & u & \\
            \beta(d) &   & \beta'(d) \\
                  & u & \\
          };
          \draw[->] (m3-2-1) -> (m3-2-3)
                    (m3-1-2) -> (m3-3-2);
                    
          \node[anchor=base] at ($(m2-3-2.south)!0.5!(m3-1-2)$) {and};
        \end{tikzpicture}
      \end{center}
      for all $d$ where $\alpha'(d), \beta'(d), \bm{q}_0, \dots, \bm{q}\sub{D - 1} \in Q^{\pm *}$.
      Then, we also have
      \begin{center}
        \begin{tikzpicture}[baseline=(m-2-1)]
          \matrix[matrix of math nodes, text height=1.25ex, text depth=0.25ex] (m) {
            & u & \\
            B_{\beta, \alpha}[\bm{p}\sub{D - 1}, \dots, \bm{p}_0] &   & B_{\beta', \alpha'}[\bm{q}\sub{D - 1}, \dots, \bm{q}_0] \text{.} \\
            & u & \\
          };
          
          \foreach \j in {1} {
            \foreach \i in {1} {
              \draw[->] let
                \n1 = {int(2+\i)},
                \n2 = {int(1+\j)}
              in
                (m-\n2-\i) -> (m-\n2-\n1);
              \draw[->] let
                \n1 = {int(1+\i)},
                \n2 = {int(2+\j)}
              in
                (m-\j-\n1) -> (m-\n2-\n1);
            };
          };
        \end{tikzpicture}
      \end{center}
    \end{fact}
  \end{section}
  
  \begin{section}{Uniform Word Problem}\label{sec:uniformWP}
    We start by showing that the uniform word problem for automaton groups is \PSPACE-complete. Although this also follows from the non-uniform case proved below, it uses the same ideas but allows for a simpler construction. This way, it serves as a good starting point and makes understanding the more complicated construction below easier.
    
    The general idea is to reduce the $\PSPACE$-complete \cite[Lemma~3.2.3]{kozen77lower}\footnote{That the alphabet may be assumed to contain exactly four elements can be seen easily. In fact, we may even assume it to be of size three or -- with suitable encoding -- size two.} \DecProblem{DFA Intersection Problem}\footnote{Typically, DFA is an abbreviation for \enquote{deterministic finite automaton}. However, as it is common in the setting of this paper, we use the term \emph{automaton} to refer to what is more precisely a transducer (we have an output). Therefore, we use the term \emph{acceptor} to refer to automata without output.}
    \problem
      [$\Gamma = \{ a_1, \dots, a_4 \}$]
      {$D \in \mathbb{N}$ and deterministic finite acceptors\newline
       $\mathcal{A}_0 = (P_0, \Gamma, \tau_0, p_{0, 0}, F_0), \dots,$\newline
       $\mathcal{A}\sub{D - 1} = (P\sub{D - 1}, \Gamma, \tau\sub{D - 1}, p_{0, D - 1}, F\sub{D - 1})$}
      {is $\bigcap_{i = 0}^{D - 1} L(\mathcal{A}_i) = \emptyset$?}\noindent
    in logarithmic space to the uniform word problem for automaton groups. The output automaton basically consists of the input acceptors and obviously stores in the states the information whether an input word was accepted or not. Finally, we use a logical conjunction based on the commutator $B_{\beta, \alpha}$ to extract whether there is an input word that gets accepted by all acceptors. While we could use other groups for the logical conjunction, we will stick for now to the group $A_5$ from \cref{ex:conjunctionInA5} for the sake of simplicity.
    
    \begin{theorem}\labelx{thm:uniformPSPACE}
      The uniform word problem for automaton groups
      \problem
        [$\Sigma = \{ a_1, \dots, a_5 \}$]
        {a \GAut $\mathcal{T} = (Q, \Sigma, \delta)$, $\bm{p} \in Q^{\pm *}$}
        {is $\bm{p} = \idGrp$ in $\mathscr{G}(\mathcal{T})$?}
      \noindent{}(even over a fixed alphabet with five elements) is \PSPACE-complete.
    \end{theorem}
    \begin{proof}
      It is known that the uniform word problem for automaton groups\footnote{In fact, even the corresponding problem for automaton semigroups is in \PSPACE.} is in \PSPACE (using a guess and check algorithm; see \cite{steinberg2015some} or \cite[Proposition 2]{dangeli2017complexity}). Thus, we only have to show that the problem is \PSPACE-hard. As already mentioned, we reduce the \DecProblem{DFA Intersection Problem}
      \problem
        [$\Gamma = \{ a_1, \dots, a_4 \}$]
        {$D \in \mathbb{N}$ and deterministic finite acceptors\newline
         $\mathcal{A}_0 = (P_0, \Gamma, \tau_0, p_{0, 0}, F_0), \dots,$\newline
         $\mathcal{A}\sub{D - 1} = (P\sub{D - 1}, \Gamma, \tau\sub{D - 1}, p_{0, D - 1}, F\sub{D - 1})$}
        {is $\bigcap_{i = 0}^{D - 1} L(\mathcal{A}_i) = \emptyset$?}\noindent
      to the uniform word problem for automaton groups in logarithmic space. As mentioned above, Kozen \cite[Lemma~3.2.3]{kozen77lower} showed that this problem is \PSPACE-hard.
      
      For the reduction, we need to map the acceptors $\mathcal{A}_0 = (P_0, \Gamma, \tau_0, p_{0, 0}, F_0),\allowbreak \dots,\allowbreak \mathcal{A}\sub{D - 1} = (P\sub{D - 1}, \Gamma, \tau\sub{D - 1}, p_{0, D - 1}, F\sub{D - 1})$ to an automaton $\mathcal{T} = (Q, \Sigma, \delta)$ and a state sequence $\bm{p} \in Q^{\pm *}$. Without loss of generality, we can assume $D = 2^d$ here for some $d$ (otherwise, we can just duplicate one of the input acceptors until we reach a power of two, which can be done in logarithmic space).
      
      We assume the state sets $P_i$ to be pairwise disjoint and set $P = \{ \id, \sigma \} \uplus \bigcup_{i = 0}^{D - 1} P_i$ and $F = \bigcup_{i = 0}^{D - 1} F_i$. Additionally, we set $\Sigma = \{ a_1, \dots, a_4 \} \uplus \{ \$ \}$ and assume that the elements $\sigma, \alpha, \beta \in A_5$ (from \cref{ex:conjunctionInA5}) act as the corresponding permutations on $\Sigma$. For the transitions, we set
      \begin{align*}
        \delta ={}& \bigcup_{i = 0}^{D - 1} \left\{ \trans{p}{a}{a}{q} \mid \transa{p}{a}{q} \in \tau_i \right\} \cup \\
                  & \left\{ \trans{r}{\$}{\$}{\id} \mid r \in P \setminus F \right\} \cup \left\{ \trans{f}{\$}{\$}{\sigma} \mid f \in F \right\} \cup \\
                  & \left\{ \trans{\id}{a}{a}{\id} \mid a \in \Sigma \right\} \cup \left\{ \trans{\sigma}{a}{\sigma(a)}{\sigma} \mid a \in \Sigma \right\}
        \text{.}
      \end{align*}
      Thus, we take the union of the acceptors and extend it into an automaton by letting all states act like the identity. With the new letter $\$$ (the ``end-of-word'' symbol), we go to $\id$ for non-accepting states and to $\sigma$ for accepting ones. Finally, we have the state $\id$, which acts like the identity, and the state $\sigma$, whose action is to apply the permutation $\sigma$ to all letters of the input word, justifying the re-use of the name $\sigma$. Finally, we define $\mathcal{T}$ as the (disjoint) union of the just defined automaton $(P, \Sigma, \delta)$ with the automaton
      
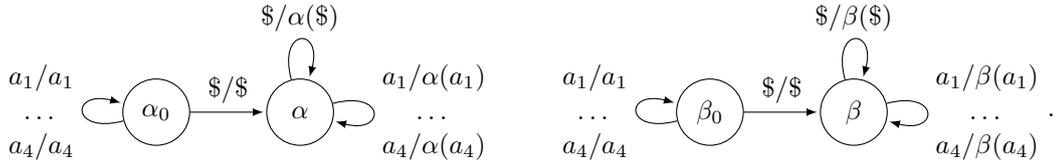
\begin{center}
        \begin{tikzpicture}[auto, shorten >=1pt, >=latex, baseline=(h0.base)]
          \node[state] (k0) {$\alpha_0$};
          \node[state, right=of k0] (k) {$\alpha$};
          \node[state, right=2cm of k] (h0) {$\beta_0$};
          \node[state, right=of h0] (h) {$\beta$};
          
          \draw[->] (k0) edge[loop below] node[align=center] {$a_1 / a_1$\\$\dots$\\$a_4 / a_4$} (k0)
                    (k0) edge node {$\$ / \$$} (k)
                    (k) edge[loop below] node[align=center] {$a_1 / \alpha(a_1)$\\$\dots$\\$a_4 / \alpha(a_4)$} (k)
                    (k) edge[loop above] node {$\$ / \alpha(\$)$} (k)
                    (h0) edge[loop below] node[align=center] {$a_1 / a_1$\\$\dots$\\$a_4 / a_4$} (h0)
                    (h0) edge node {$\$ / \$$} (h)
                    (h) edge[loop below] node[align=center] {$a_1 / \beta(a_1)$\\$\dots$\\$a_4 / \beta(a_4)$} (h)
                    (h) edge[loop above] node {$\$ / \beta(\$)$} (h)
          ;
        \end{tikzpicture}.
      \end{center}
      Notice that $\mathcal{T}$ is deterministic, complete and invertible and that all states except $\sigma, \alpha$ and $\beta$ act like the identity on words not containing $\$$. Also note that on input of $\mathcal{A}_0, \dots, \mathcal{A}\sub{D - 1}$, the automaton can clearly be computed in logarithmic space.
      
      For the state sequence, we set $\bm{p} = B_0[p_{0, D - 1}, \dots, p_{0, 0}]$ where we use $B_0$ as a short-hand notation for the balanced commutator $B_{\beta_0, \alpha_0}$ defined in \cref{def:balancedCommutator}. Observe that, by \cref{fct:BIsLogspaceComputable}, we can compute $\bm{p}$ in logarithmic space ($\alpha_0$ and $\beta_0$ are even constant functions in our setting).
      
      This completes our description of the reduction and it remains to show its correctness. If there is some $w \in \bigcap_{i = 1}^{\ell} L(\mathcal{A}_i)$, we have to show $\bm{p} \neq_{\mathscr{G}(\mathcal{T})} \idGrp$. We have
      \begin{center}
        \begin{tikzpicture}[baseline=(m-6-5.base)]
          \matrix[matrix of math nodes, text height=1.25ex, text depth=0.25ex] (m) {
                        & w      &             & \$ &        \\
            p_{0, 0}    &        & q_{f, 0}    &    & \sigma \\
                        & w      &             & \$ &        \\
            \vdots      & \vdots &             & \vdots &    \\
                        & w      &             & \$ &        \\
            p_{0, D - 1} &        & q_{f, D - 1} &    & \sigma \\
                        & w      &             & \$ &        \\
          };
          \foreach \j in {1, 5} {
            \foreach \i in {1, 3} {
              \draw[->] let
                \n1 = {int(2+\i)},
                \n2 = {int(1+\j)}
              in
                (m-\n2-\i) -> (m-\n2-\n1);
              \draw[->] let
                \n1 = {int(1+\i)},
                \n2 = {int(2+\j)}
              in
                (m-\j-\n1) -> (m-\n2-\n1);
            };
          };
        \end{tikzpicture}
      \end{center}
      where all $q_{f, i} \in F_i$ are final states. Thus, by \cref{fct:commutatorInCrossDiagrams}, we also have
      \begin{center}
        \begin{tikzpicture}[baseline=(m-2-5.base)]
          \matrix[matrix of math nodes, text height=1.25ex, text depth=0.25ex] (m) {
                        & w      &             & \$ &        \\
            \bm{p} = B_0[p_{0, D - 1}, \dots, p_{0, 0}] &        & B_0[q_{f, D - 1}, \dots, q_{f, 0}] &    & B[\underbrace{\sigma, \dots, \sigma}_{D \text{ times}}] \\
                        & w      &             & \$ &        \\
          };
          \foreach \j in {1} {
            \foreach \i in {1, 3} {
              \draw[->] let
                \n1 = {int(2+\i)},
                \n2 = {int(1+\j)}
              in
                (m-\n2-\i) -> (m-\n2-\n1);
              \draw[->] let
                \n1 = {int(1+\i)},
                \n2 = {int(2+\j)}
              in
                (m-\j-\n1) -> (m-\n2-\n1);
            };
          };
        \end{tikzpicture},
      \end{center}
      where we used $B$ as an abbreviation for $B_{\beta, \alpha}$. Without loss of generality, we may assume $\sigma(a_1) \neq a_1$ and, since we have $B[\sigma, \dots, \sigma] = \sigma$ in $A_5$ and also in $\mathscr{G}(\mathcal{T})$ (see \cref{ex:conjunctionInA5}), we obtain $\bm{p} \circ w \$ a_1 = w \$ \sigma(a_1) \neq w \$ a_1$.
      
      If, on the other hand, $\bigcap_{i = 0}^{D - 1} L(\mathcal{A}_i) = \emptyset$, we have to show $\bm{p} =_{\mathscr{G}(\mathcal{T})} \idGrp$. For this, let $w \in \Sigma^*$ be arbitrary. If $w$ does not contain any $\$$, we do not need to show anything since, by construction, only the states $\sigma$, $\alpha$ and $\beta$ act non-trivially on these words and they can only be reached after reading a $\$$. If $w$ contains $\$$, we can write $w = u \$ v$ with $u \in \{ a_1, \dots, a_4 \}^*$. Since the intersection is empty, there is some $i \in \{ 0, \dots, D - 1 \}$ with $u \not \in L(\mathcal{A}_i)$ and we obtain
      \begin{center}
        \begin{tikzpicture}[baseline=(m-6-5.base)]
          \matrix[matrix of math nodes, text height=1.25ex, text depth=0.25ex] (m) {
                        & u      &          & \$ &        \\
            p_{0, 0}    &        & q_0      &    & g_0 \\
                        & u      &          & \$ &        \\
            \vdots      & \vdots &          & \vdots &    \\
                        & u      &          & \$ &        \\
            p_{0, i}    &        & q_i      &    & g_i\rlap{${} = \id$} \\
                        & u      &          & \$ &        \\
            \vdots      & \vdots &          & \vdots &    \\
                        & u      &          & \$ &        \\
            p_{0, D - 1} &        & q_{D - 1} &    & g_{D - 1} \\
                        & u      &          & \$ &        \\
          };
          \foreach \j in {1,5,9} {
            \foreach \i in {1, 3} {
              \draw[->] let
                \n1 = {int(2+\i)},
                \n2 = {int(1+\j)}
              in
                (m-\n2-\i) -> (m-\n2-\n1);
              \draw[->] let
                \n1 = {int(1+\i)},
                \n2 = {int(2+\j)}
              in
                (m-\j-\n1) -> (m-\n2-\n1);
            };
          };
        \end{tikzpicture}
      \end{center}
      where $g_0, \dots, g\sub{D - 1} \in \{ \id, \sigma \}$. Again, we also obtain the cross diagram\newline
      {\centering% This is ugly but creates a better page break
       \resizebox{\linewidth}{!}{%
        \begin{tikzpicture}[baseline=(m-2-5.base)]
          \matrix[matrix of math nodes, text height=1.25ex, text depth=0.25ex, ampersand replacement=\&] (m) {
                        \& u      \&             \& \$ \&        \\
            \bm{p} = B_0[p_{0, D - 1}, \dots, p_{0, 0}] \&        \& B_0[q\sub{D - 1}, \dots, q_i, \dots, q_0] \&    \& B[g\sub{D - 1}, \dots, g_i, \dots, g_0] \\
                        \& u      \&             \& \$ \&        \\
          };
          \foreach \j in {1} {
            \foreach \i in {1, 3} {
              \draw[->] let
                \n1 = {int(2+\i)},
                \n2 = {int(1+\j)}
              in
                (m-\n2-\i) -> (m-\n2-\n1);
              \draw[->] let
                \n1 = {int(1+\i)},
                \n2 = {int(2+\j)}
              in
                (m-\j-\n1) -> (m-\n2-\n1);
            };
          };
        \end{tikzpicture}}
      }
      by \cref{fct:commutatorInCrossDiagrams} but, this time, we have $B[g\sub{D - 1}, \dots, g_i, \dots, g_0] = \idGrp$ in $A_5$ since $g_i = \id$ (see \cref{fct:commutatorCollapses}). Accordingly, we have $\bm{p} \circ u \$ v = u \$ v$.\qed
    \end{proof}
  \end{section}

  \begin{section}{Non-Uniform Word Problem}\label{sec:nonuniformWP}
    In this section, we are going to lift the result from the previous section to the non-uniform case. We show:
    \begin{theorem}\labelx{thm:nonuniformPSPACE}
      There is an automaton group with a \PSPACE-complete word problem:
      \problem
        [a \GAut $\mathcal{T} = (Q, \Sigma, \delta)$ with $|\Sigma| = 2$]
        {$\bm{q} \in Q^{\pm *}$}
        {is $\bm{q} = \idGrp$ in $\mathscr{G}(\mathcal{T})$?}
    \end{theorem}
    In order to prove this theorem, we are going to adapt the construction used in \cite[Proposition~6]{dangeli2017complexity} to show that there is an inverse automaton semigroup with a \PSPACE-complete word problem and that there is an automaton group whose word problem with a single rational constraint is \PSPACE-complete. The main idea is to do a reduction directly from a Turing machine accepting an arbitrary $\PSPACE$-complete problem.
    
    Let $M$ be such a deterministic, polynomially space-bounded Turing machine with input alphabet $\Lambda$, tape alphabet $\Delta$, blank symbol $\blank$, state set $P$, initial state $p_0$ and accepting states $F \subseteq P$. Thus, for any input word of length $n$, all configurations of $M$ are of the form $\blank \Delta^\ell P \Delta^m \blank$ with $\ell + 1 + m = s(n)$ for some polynomial $s$. This makes the word problem of $M$
    \problem
      [the \PSPACE machine $M$]
      {$w \in \Lambda^*$ of length $n$}
      {does $M$ reach a configuration with a state from $F$ from the initial configuration $\blank p_0 w \blank^{s(n) - n - 1} \blank$?}\noindent
    \PSPACE-complete and we will eventually do a co-reduction to the word problem of a \GAut $\mathcal{T} = (Q, \Sigma, \delta)$. In fact, we will not work with the Turing machine $M$ directly but instead use the transition function $\tau: \Gamma^3 \to \Gamma$ for $\Gamma = P \uplus \Delta$ from \cref{fct:turingNormalization}.
    
    Our automaton $\mathcal{T}$ operates in two modes. In the first mode, which we will call the \enquote{TM mode}, it interprets its input word as a sequence of configurations of $M$ and verifies that the configuration sequence constitutes a valid computation. This verification is done by multiple states (where each state is responsible for a different verification part) and the information whether the verification was successful is stored in the state, \textbf{not} by manipulating the input word. So we have \emph{successful states} and \emph{fail states}. Upon reading a special input symbol, the automaton will switch into a second mode, the \enquote{commutator mode}. More precisely, successful states go into a dedicated \enquote{okay} state $r$ and the fail states go into a state which we call $\id$ for reasons that become apparent later. Finally, to extract the information from the states, we use the iterated commutator from \cref{def:balancedCommutator}.
    
    Eventually, we want the alphabet of the automaton to only contain two letters. However, we will first describe how the TM mode of the automaton works for more letters (using an automaton $\mathcal{T}'$) and then how we can encode this automaton over two letters. Finally, we will extend this encoding with the commutator mode part of the automaton to eventually obtain $\mathcal{T}$.
    
    \paragraph{Generalized Check-Marking.}
    The idea for the TM mode is similar to Kozen's approach for showing that the \DecProblem{DFA Intersection Problem} is \PSPACE-com\-ple\-te \cite[Lemma~3.2.3]{kozen77lower}: the input word is interpreted as a sequence of configurations of a \PSPACE Turing machine where each configuration is of length $s(n)$:
    \[
      \gamma_0^{(0)} \gamma_1^{(0)} \gamma_2^{(0)} \dots \gamma_{s(n) - 1}^{(0)} \, \# \, \gamma_0^{(1)} \gamma_1^{(1)} \gamma_2^{(1)} \dots \gamma_{s(n) - 1}^{(1)} \, \# \, \dots
    \]
    In Kozen's proof, there is an acceptor for each position $i$ of the configurations with $0 \leq i < s(n)$ which checks for all $t$ whether the transition from $\gamma_i^{(t)}$ to $\gamma_i^{(t + 1)}$ is valid. In our case, however, the automaton must not depend on the input (or its length $n$) and we have to handle this a bit differently. The first idea is to use a \enquote{check-mark approach}. First, we check all first positions ($\gamma_0^{(0)}, \gamma_0^{(1)}, \dots$) for valid transitions. Then, we put a check-mark on all these first positions, which tells us that we now have to check all second positions ($\gamma_1^{(0)}, \gamma_1^{(1)}, \dots$, i.\,e.\ the first ones without a check-mark). Again, we put a check-mark on all these, continue with checking all third positions and so on (\cref{fig:checkmarking}).
    \begin{figure}\centering%
      \trimbox{0pt 1ex 0pt 0ex}{%
      \begin{tikzpicture}
        \matrix (m) [matrix of math nodes, every node/.style/.append={inner sep=0pt}] {
          \underset{\phantom{\checkmark}}{\gamma_0^{(0)}} & \gamma_1^{(0)} & \gamma_2^{(0)} & \dots & \gamma_{s(n) - 1}^{(0)} & \;\#\;\; & \gamma_0^{(1)} & \gamma_1^{(1)} & \gamma_2^{(1)} & \dots & \gamma_{s(n) - 1}^{(1)} & \# & \dots \\
          \underset{\checkmark}{\gamma_0^{(0)}} & \gamma_1^{(0)} & \gamma_2^{(0)} & \dots & \gamma_{s(n) - 1}^{(0)} & \# & \underset{\checkmark}{\gamma_0^{(1)}} & \gamma_1^{(1)} & \gamma_2^{(1)} & \dots & \gamma_{s(n) - 1}^{(1)} & \# & \dots \\
          \underset{\checkmark}{\gamma_0^{(0)}} & \underset{\checkmark}{\gamma_1^{(0)}} & \gamma_2^{(0)} & \dots & \gamma_{s(n) - 1}^{(0)} & \# & \underset{\checkmark}{\gamma_0^{(1)}} & \underset{\checkmark}{\gamma_1^{(1)}} & \gamma_2^{(1)} & \dots & \gamma_{s(n) - 1}^{(1)} & \# & \dots \\
        };
        \path[->] (m-1-1.west) edge[bend right] node[left] {check} (m-2-1.west)
                  (m-2-1.west) edge[bend right] node[left] {check} (m-3-1.west);
      \end{tikzpicture}}%
      \caption{Illustration of the check-mark approach.}\label{fig:checkmarking}%
    \end{figure}
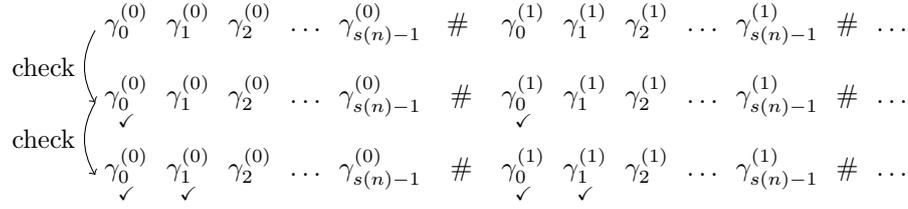%
    
    The problem with this approach is that the check-marking leads to an intrinsically non-invertible automaton (see \cref{fig:addCheckmarkAutomaton}).
    \begin{figure}%
      \centering%
      \trimbox{0pt 1ex 0pt 0ex}{%
      \begin{tikzpicture}[auto, shorten >=1pt, >=latex]
        \node[state] (check) {};
        \node[state, right=of check] (wait) {};
        
        \draw[->] (check) edge[loop left] node {${\gamma \atop \checkmark} / {\gamma \atop \checkmark}$} (check)
                          edge[bend left] node {$\gamma / {\gamma \atop \checkmark}$} (wait)
                  (wait) edge[loop right] node[align=center] {$\gamma / \gamma$\\${\gamma \atop \checkmark} / {\gamma \atop \checkmark}$} (wait)
                         edge[bend left] node {$\# / \#$} (check)
        ;
      \end{tikzpicture}}%
      \caption{Adding a check-mark yields a non-invertible automaton}\label{fig:addCheckmarkAutomaton}%
    \end{figure}
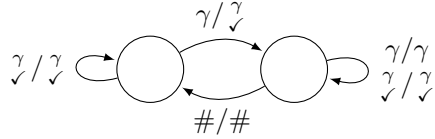
    To circumvent this, we generalize the check-mark approach: in front of each symbol $\gamma_i^{(t)}$ of a configuration, we add a $0^k$ block (of sufficient length $k$). In the spirit of \cref{ex:addingMachine}, we interpret this block as representing a binary number. We consider the symbol following the block as \enquote{unchecked} if the number is zero; for all other numbers, it is considered as \enquote{checked}. Now, checking the next symbol boils down to incrementing each block until we have encountered a block whose value was previously zero (and this can be detected while doing the increment). This idea is depicted in \cref{fig:generalizedCheckmarking}. It would also be possible to have the check-mark block after each symbol instead of before (which might be more intuitive) but it turns out that our ordering has some technical advantages.
    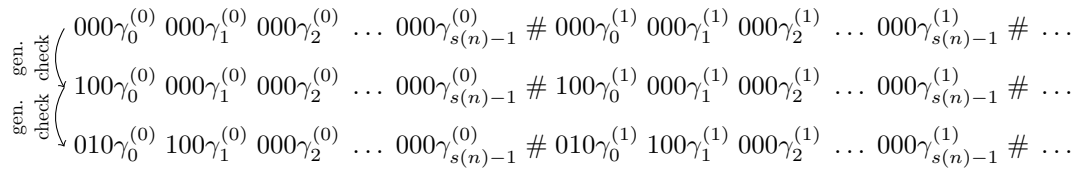
\begin{figure}%
      \centering%
      \resizebox{\linewidth}{!}{%
      \begin{tikzpicture}
        \matrix (m) [matrix of math nodes, column sep=0pt, every node/.style/.append={inner sep=0pt}, row sep=0.5cm] {
          000 \gamma_0^{(0)} \; 000 \gamma_1^{(0)} \; 000 \gamma_2^{(0)} \; \dots \; 000 \gamma_{s(n) - 1}^{(0)} \; \# \; 000 \gamma_0^{(1)} \; 000 \gamma_1^{(1)} \; 000 \gamma_2^{(1)} \; \dots \; 000 \gamma_{s(n) - 1}^{(1)} \; \# \; \dots \\
          100 \gamma_0^{(0)} \; 000 \gamma_1^{(0)} \; 000 \gamma_2^{(0)} \; \dots \; 000 \gamma_{s(n) - 1}^{(0)} \; \# \; 100 \gamma_0^{(1)} \; 000 \gamma_1^{(1)} \; 000 \gamma_2^{(1)} \; \dots \; 000 \gamma_{s(n) - 1}^{(1)} \; \# \; \dots \\
          010 \gamma_0^{(0)} \; 100 \gamma_1^{(0)} \; 000 \gamma_2^{(0)} \; \dots \; 000 \gamma_{s(n) - 1}^{(0)} \; \# \; 010 \gamma_0^{(1)} \; 100 \gamma_1^{(1)} \; 000 \gamma_2^{(1)} \; \dots \; 000 \gamma_{s(n) - 1}^{(1)} \; \# \; \dots \\
        };
        \path[->] (m-1-1.west) edge[bend right] node[right, align=center] {gen.~check-mark} (m-2-1.west)
        (m-2-1.west) edge[bend right] node[right, align=center] {gen.~check-mark} (m-3-1.west);
      \end{tikzpicture}}%
      \caption{The idea of our generalized check-marking approach.}\label{fig:generalizedCheckmarking}
    \end{figure}
    
    \paragraph{Construction of $\mathcal{T}'$.}
    Using the check-mark approach, we can construct the automaton $\mathcal{T}'$, which implements the TM mode over the alphabet $\Sigma' = \{ 0, 1, \#, \$ \} \cup \Gamma$. Our automaton will have two dedicated states $r$ and $\id$. The semantics of them is that $r$ is an \enquote{okay} state while $\id$ is a \enquote{fail} state. Both states will only be entered after reading the first $\$$ and, for now, we let them both operate as the identity but later, when we describe the commutator mode of the automaton, we will replace $r$ with a non-identity state.\footnote{In fact, we will need multiple copies of $\mathcal{T}'$ where we replace $r$ by a different state in each copy.}
    
    The idea is that the input of the automaton is of the form $u \$ v$ where $u$ is in $(\{ 0, 1, \# \} \cup \Gamma)^*$ and describes a computation of $M$ (with digit blocks for the generalized check-marking approach). We will define multiple state sequences where each one checks a different aspect of whether the computation is valid and accepting for some given input word $w \in \Lambda^*$ for the Turing machine. If such a check passes, we basically end in the state $r$ after reading the $\$$ and, if it fails, we end in a state sequence only containing $\id$s. This way, all checking state sequences end in $r$ if and only if $M$ accepts $w$. This information will later be extracted in the commutator mode described below.
    \begin{proposition}\labelx{prop:TMmode}
      There is a \GAut $\mathcal{T}' = (Q', \Sigma', \delta')$ with an identity state $\id$, a dedicated state $r \in Q'$ and alphabet $\Sigma' = \Sigma_1 \cup \{ \$ \}$ for $\Sigma_1 = \{ 0, 1, \# \} \cup \Gamma$ such that, on input of $w \in \Lambda^*$, one can compute state sequences $\bm{p}_{i}$ for $0 \leq i < D$ (for some $D \geq 1$) in logarithmic space so that the following holds:
      
      For every $u \in \Sigma_1^*$ and all $0 \leq i < D$, we have $\bm{p}_i \circ u \$ = u \$$ and
      \begin{enumerate}
          \item if $M$ accepts $w$, then there is some $u \in \Sigma_1^*$ such that, for all $0 \leq i < D$, we have $\bm{p}_{i} \cdot u \$ \in \id^* r \id^*$ and
          \item if $M$ does not accept $w$, then, \emph{for all} $u \in \Sigma_1^*$, there is some $0 \leq i < D$ such that we have $\bm{p}_{i} \cdot u \$ \in \id^*$.
      \end{enumerate}
    \end{proposition}
    \begin{proof}
      The automaton $\mathcal{T}'$ is the union of several simpler automata. We will use $0$ and $1$ for the generalized check-mark approach, $\#$ is used to separate individual configurations and $\$$ acts as an \enquote{end-of-computation} symbol switching the automaton to one of the special states $r$ or $\id$. We call the first part of an input word form $\Sigma_1^*$ the \emph{TM part} (because we consider it to encode a computation of $M$) and everything after the first $\$$ the \emph{commutator part} (because we will only use it later).
      
      The first part of the automaton contains the two special states $r$ and $\id$, which we both let act as the identity\footnote{\dots so they cannot be distinguished algebraically at the moment. We will later replace $r$ with another state to distinguish them.} but it helps to think of $\id$ as a \enquote{fail} state and $r$ as an \enquote{okay} state (like the final states in the proof of \cref{thm:uniformPSPACE}):
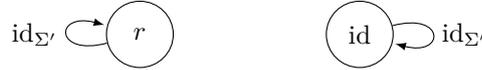
      \begin{center}
        \begin{tikzpicture}[auto, shorten >=1pt, >=latex, node distance=2cm, baseline=(id.base)]
          \node[state] (r) {$r$};
          \node[state, right=of r] (id) {$\id$};
          
          \draw[->] (r) edge[loop left] node {$\id_{\Sigma'}$} (r)
                    (id) edge[loop right] node {$\id_{\Sigma'}$} (r)
          ;
        \end{tikzpicture}
      \end{center}
      Here, we have introduced a convention: we use arrows labeled by $\id_A$ for some $A \subseteq \Sigma'$ to indicate that we have an $a / a$-transition for all $a \in A$.
      
      Next, we add a part to our automaton to check that the TM part of the input is of the form\footnote{We do not check that the digit blocks for the check-marking are non-empty here. This is handled implicitly by other states below (namely by $\checksub{r}$ but independently also by $c$).} $(0^* \Gamma)^+ \left( \# (0^* \Gamma)^+ \right)^*$:
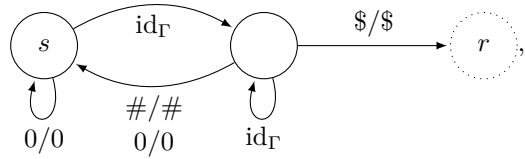
      \begin{center}
        \begin{tikzpicture}[auto, shorten >=1pt, >=latex, node distance=2cm, baseline=(sigma.base)]
          \node[state] (r) {$z$};
          \node[state, right=of r] (s) {};
          \node[state, dotted, right=of s] (sigma) {$r$};
          
          \draw[->] (r) edge[loop below] node {$0/0$} (r)
                        edge[bend left] node[swap] {$\id_\Gamma$} (s)
                    (s) edge[loop below] node {$\id_\Gamma$} (s)
                        edge[bend left] node[align=center] {$\# / \#$\\$0 / 0$} (r)
                        edge node {$\$ / \$$} (sigma)
          ;
        \end{tikzpicture},
      \end{center}
      where we have introduced another convention: whenever a transition is missing for some $a \in \Sigma'$, there is an implicit $a / a$-transition to the state $\id$ (as defined above). Additionally, dotted states refer to the corresponding states defined above.
      
      Note that we do not check that the factors in $(0^* \Gamma)^+$ correspond to well-formed configurations for the Turing machine here. This will be done implicitly by checking that the input word belongs to a valid computation of the Turing machine, which we describe below.
      
      We also need a part which checks whether the TM part of the input word contains a final state (if this is not the case, we want to \enquote{reject} the word):
      
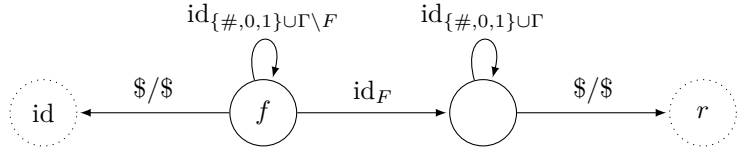
\begin{center}
        \begin{tikzpicture}[auto, shorten >=1pt, >=latex, node distance=2cm, baseline=(sigma.base)]
          \node[state, dotted] (1) {$\id$};
          \node[state, right=of 1] (f) {$f$};
          \node[state, right=of f] (g) {};
          \node[state, dotted, right=of g] (sigma) {$r$};
          
          \draw[->] (f) edge node[swap] {$\$ / \$$} (1)
                        edge[loop above] node {$\id_{\{ \#, 0, 1 \} \cup \Gamma \setminus F}$} (f)
                        edge node {$\id_F$} (g)
                    (g) edge[loop above] node {$\id_{\{ \#, 0, 1 \} \cup \Gamma}$} (g)
                        edge node {$\$ / \$$} (sigma)
          ;
        \end{tikzpicture}
      \end{center}

      Finally, we come to the more complicated parts of $\mathcal{T}'$. The first one is for the generalized check-marking as described above and is depicted in \cref{fig:checkmarkingAutomaton}, which we actually need twice: once for $g = r$ and once for $g = \id$. Notice, however, that, during the TM mode (i.\,e.\ before the first $\$$), both versions of $\checksub{g}$ behave exactly the same way; the only difference is after switching to the commutator mode: while $\checksub{\id}$ always goes to $\id$, $\checksub{r}$ goes to $r$ (if the check-marking was successful and to $\id$, otherwise).
      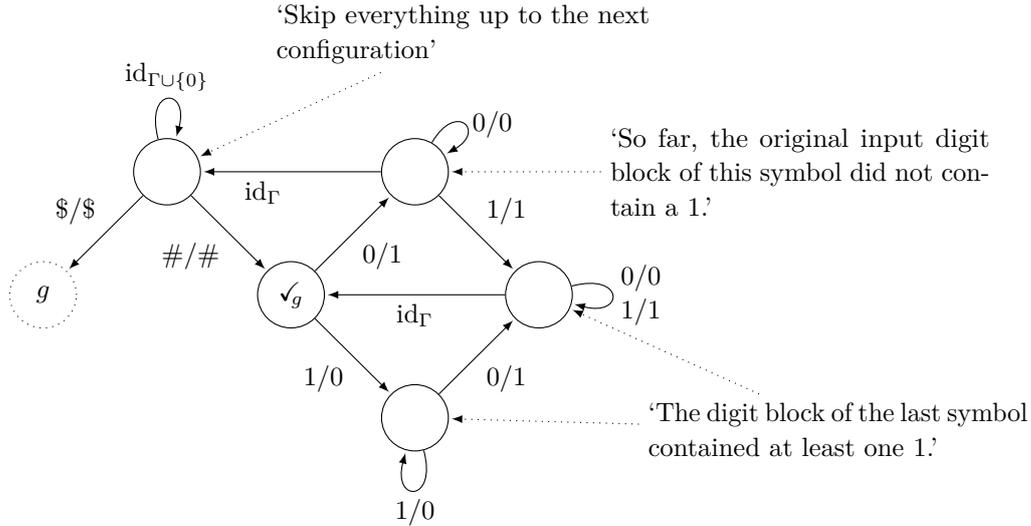
\begin{figure}%
        \centering
        \begin{tikzpicture}[auto, shorten >=1pt, >=latex]
          \node[state] (checkmark) {$\checksub{g}$};
          \node[state, above right=of checkmark] (2) {};
          \node[state, below right=of 2] (3) {};
          \node[state, below right=of checkmark] (4) {};
          \node[state, above left=of checkmark] (skip) {};
          \node[state, dotted, below left=of skip] (end) {$g$};
          
          \node[right=1.5cm of 2, align=left] (2annotation) {\begin{varwidth}{4cm}\enquote{So far, the original input digit block of this symbol did not contain a $1$.}\end{varwidth}};
          \node[below right=1cm and 0.25cm of 3] (34annotation) {\begin{varwidth}{4.5cm}\enquote{The digit block of the last symbol contained at least one $1$.}\end{varwidth}};
          \node[above right=of skip] (skipannotation) {\enquote{Skip everything up to the next configuration}};

          \path[->] (checkmark) edge node[below right] {$0/1$} (2)
                                edge node[below left] {$1/0$} (4)
                    (2) edge[loop right, out=60, looseness=7, in=30] node {$0/0$} (2)
                    (2) edge node {$1/1$} (3)
                    (2) edge node[below left] {$\id_\Gamma$} (skip)
                    (3) edge[loop right] node[align=left] {$0/0$\\$1/1$} (3)
                    (3) edge node[below] {$\id_\Gamma$} (1)
                    (4) edge[loop below] node {$1/0$} (4)
                    (4) edge node[below right] {$0/1$} (3)
                    (skip) edge[loop above] node {$\id_{\Gamma \cup \{ 0 \}}$} (skip)
                    (skip) edge node[swap] {$\#/\#$} (checkmark)
                    (skip) edge node[swap] {$\$/\$$} (end)
                    ;
          \path[->] (2annotation) edge[dotted] (2)
                    (34annotation) edge[dotted] (3)
                    (34annotation) edge[dotted] (4)
                    (skipannotation) edge[dotted] (skip);
        \end{tikzpicture}%
        \vspace*{-\baselineskip}
        \caption{The automaton part used for generalized check-marking}\label{fig:checkmarkingAutomaton}%
      \end{figure}%
      
      Additionally, we also need an automaton part verifying that every configuration symbol has been check-marked (in the generalized sense):
      
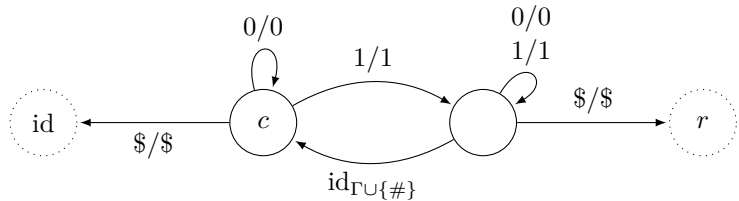
\begin{center}
        \begin{tikzpicture}[auto, shorten >=1pt, >=latex, node distance=2cm, baseline=(sigma.base)]
          \node[state] (c) {$c$};
          \node[state, right=of c] (ok) {};
          \node[state, right=of ok] (g) {};
          \node[state, dotted, right=of g] (r) {$r$};
          
          \draw[->] (c) edge[loop above] node {$0/0$} (c)
                        edge node {$1/1$} (ok)
                    (ok) edge[loop above, out=150, looseness=7, in=120] node[align=center]  {$0/0$\\$1/1$} (ok)
                         edge[bend left] node {$\id_{\Gamma}$} (g)
                    (g) edge node {$1/1$} (ok)
                        edge[bend left] node[below right, pos=0.2, align=center] {$\#/\#$\\$0/0$} (c)
                        edge node {$\$/\$$} (r)
          ;
        \end{tikzpicture}
      \end{center}
      
      The last part is for checking the validity of the transitions at all first so-far unchecked positions. While it is not really difficult, this part is a bit technical. Intuitively, for checking the transition from time step $t - 1$ to time step $t$ at position $i$, we need to compute $\gamma_i^{(t)} = \tau(\gamma_{i - 1}^{(t - 1)}, \gamma_i^{(t - 1)}, \gamma_{i + 1}^{(t - 1)})$ from the configuration symbol at positions $i - 1$, $i$ and $i + 1$ for time step $t - 1$. We store $\gamma_i^{(t)}$ in the state (to compare it to the actual value). Additionally, we need to store the last two symbols of configuration $t$ we have encountered so far (for computing what we expect in the next time step later on) and whether we have seen a $1$ or only $0$s in the check-mark digit block.
      For all this, we use the states
      
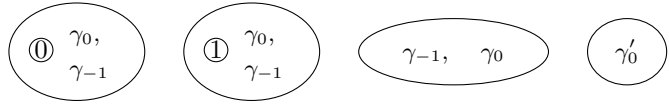
\begin{center}
        \begin{tikzpicture}[auto, shorten >=1pt, >=latex]
          \node[state, ellipse, align=center, inner sep=0pt] (0) {\circledZero \begin{tabular}{l}
             $\gamma_{0}$, \\
             $\gamma_{-1}$
            \end{tabular}};

          \node[state, ellipse, align=center, right=0.5cm of 0, inner sep=0pt] (1) {\circledOne \begin{tabular}{l}
            $\gamma_{0}$, \\
            $\gamma_{-1}$
            \end{tabular}};

          \node[state, ellipse, align=center, right=0.5cm of 1, inner sep=0pt] (b) {\begin{tabular}{ll}
            $\gamma_{-1}$,  & $\gamma_0$
            \end{tabular}};
          
          \node[state, ellipse, align=center, right=0.5cm of b, inner sep=0pt] (skip) {\begin{tabular}{l}
            $\gamma'_{0}$
            \end{tabular}};
        \end{tikzpicture}
      \end{center}
      with $\gamma_{-1}, \gamma_0, \gamma_0' \in \Gamma$. The idea is the following. In the $\circledZero$ and $\circledOne$ states, we store the value we expect for the first unchecked symbol ($\gamma_0$) and the last symbol we have seen in the current configuration ($\gamma_{-1}$). We are in the $\circledZero$ state if we have not seen any $1$ in the digit block yet and in the $\circledOne$ if we did. The two states on the right are used to skip the rest of the current configuration and to compute the symbol we expect for the first unchecked position in the next configuration ($\gamma_0'$).

      We use these states in the transitions schematically depicted in \cref{fig:checkingTransitions}. Here, the dashed transitions exist for all $\gamma_{-1}'$ and $\gamma_1$ in $\Gamma$ but go to different states, respectively, and the dotted states correspond to the respective non-dotted states with different values for $\gamma_{0}$ and $\gamma_{-1}$ (with the exception of $r$, which corresponds to the state defined above). We also define $q_{\gamma'}$ as the state on the bottom right (for $\gamma' \in \Gamma$).
      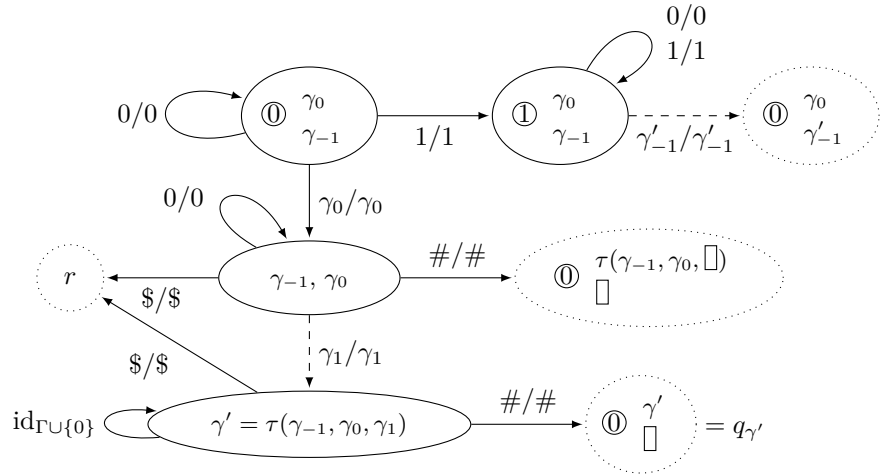
\begin{figure}[h]
        \begin{center}
          \begin{tikzpicture}[auto, shorten >=1pt, >=latex, node distance=1cm and 1.5cm]
            \node[state, ellipse, align=center, inner sep=0pt] (0) {\circledZero \begin{tabular}{l}
              $\gamma_{0}$ \\
              $\gamma_{-1}$
              \end{tabular}};
            
            \node[state, ellipse, align=center, below=of 0] (b) {\begin{tabular}{ll}
              $\gamma_{-1}$, $\gamma_{0}$
              \end{tabular}};
            
            \node[state, ellipse, align=center, below=of b, inner sep=0pt] (skip) {\begin{tabular}{l}
              $\gamma' = \tau(\gamma_{-1}, \gamma_0, \gamma_1)$
              \end{tabular}};
            
            \node[state, ellipse, align=center, right=of 0, inner sep=0pt] (1) {\circledOne \begin{tabular}{l}
              $\gamma_{0}$ \\
              $\gamma_{-1}$
              \end{tabular}};

            \node[state, ellipse, align=center, dotted, right=of 1, inner sep=0pt] (0') {\circledZero  \begin{tabular}{l}
              $\gamma_{0}$ \\
              $\gamma_{-1}'$
              \end{tabular}};
            
            \node[state, ellipse, align=center, dotted, right=of b, inner sep=0pt] (0nb) {\circledZero \begin{tabular}{l}
              $\tau(\gamma_{-1}, \gamma_{0}, \blank)$ \\
              $\blank$
              \end{tabular}};
            
            \node[state, ellipse, align=center, dotted, right=of skip, inner sep=0pt] (0n) {\circledZero \begin{tabular}{l}
              $\gamma'$ \\
              $\blank$
              \end{tabular}};
            \node[anchor=west, inner sep=0pt] at (0n.base -| 0n.east) {${} = q_{\gamma'}$};

            \node[state, dotted, anchor=center, left=of b] (sigma) {$r$};
  
            \path[->] (0) edge[loop left] node {$0/0$} (0)
                          edge node[swap] {$1/1$} (1)
                          edge node {$\gamma_0/\gamma_0$} (b)
                      (1) edge[loop, out=60, looseness=7, in=30] node[right, xshift=0.5em, align=center] {$0/0$\\$1/1$} (1)
                          edge[dashed] node[swap] {$\gamma_{-1}'/\gamma_{-1}'$} (0')
                      (b) edge[loop, out=150, in=120, looseness=7] node[left, xshift=-0.5em] {$0/0$} (b)
                          edge node {$\#/\#$} (0nb)
                          edge node {$\$/\$$} (sigma)
                          edge[dashed] node {$\gamma_1/\gamma_1$} (skip)
                      (skip) edge[loop left, out=185, looseness=7, in=175] node {$\id_{\Gamma \cup \{ 0 \}}$} (skip)
                             edge node {$\$/\$$} (sigma)
                             edge node {$\#/\#$} (0n)
            ;
          \end{tikzpicture}
        \end{center}
        \caption{Schematic representation of the transitions used for checking Turing machine transitions and definition of $q_{\gamma'}$; the dashed transitions exist for all $\gamma_{-1}'$ and $\gamma_1$ in $\Gamma$ but go to different states, respectively}\label{fig:checkingTransitions}
      \end{figure}
      
      The automaton parts depicted in \cref{fig:checkmarkingAutomaton} and \cref{fig:checkingTransitions} are best understood with an example. Consider the input word
      \[
        100 \gamma_0^{(0)} \; 000 \gamma_1^{(0)} \; 000 \gamma_2^{(0)} \; \# \; 100 \gamma_0^{(1)} \; 000 \gamma_1^{(1)} \; 000 \gamma_2^{(1)} \; \$
      \]
      where we consider the $\gamma_i^{(t)}$ to form a valid computation. If we start in state $q_{\gamma_1^{(0)}}$ and read the above word, we immediately take the $1/1$-transition and go into the corresponding $\circledOne$ state where we skip the rest of the digit block. Using the dashed transition, the next symbol $\gamma_0^{(0)}$ takes us back into a $\circledZero$-state where the upper entry is still $\gamma_1^{(0)}$ but the lower entry is now $\gamma_0^{(0)}$ (i.\,e.\ the last configuration symbol we just read). We loop at this state while reading the next three $0$s and, since the next symbol $\gamma_1^{(0)}$ matches with the one stored in the state, we get into the state with entries $\gamma_0^{(0)}, \gamma_1^{(0)}$ where we skip the next three $0$s again. Reading $\gamma_2^{(0)}$ now gets us into the state with entry $\gamma_1^{(1)}$ since we have $\tau(\gamma_0^{(0)}, \gamma_1^{(0)}, \gamma_2^{(0)}) = \gamma_1^{(1)}$ by assumption that the $\gamma_i^{(t)}$ form a valid computation. Here, we read $\#/\#$ and the process repeats for the second configuration, this time starting in $q_{\smash{\gamma_1^{(1)}}}$. When reading the final $\$$, we are in the state with entry $\tau(\gamma_0^{(1)}, \gamma_1^{(1)}, \gamma_2^{(1)})$ and finally go to $r$. Notice that during the whole process, we have not changed the input word at all!
      
      If we now start reading the input word again in state $\checksub{r}$ (see \cref{fig:checkmarkingAutomaton} and also refer to \cref{fig:generalizedCheckmarking}), we turn the first $1$ into a $0$, go to the state at the bottom, turn the next $0$ into a $1$ and go to the state on the right, where we ignore the next $0$. When reading $\gamma_0^{(0)}$, we go back to $\checksub{r}$. Next, we take the upper exit and turn the next $0$ into a $1$. The remaining $0$s are ignored and we remain in the state at the top right until we read $\gamma_1^{(0)}$ and go to the state at the top left. Here, we ignore everything up to $\#$, which gets us back into $\checksub{r}$. The second part works in the same way with the difference that we go to $r$ at the end since we encounter the $\$$ instead of $\#$. The output word, thus, is
      \[
        010 \gamma_0^{(0)} \; 100 \gamma_1^{(0)} \; 000 \gamma_2^{(0)} \; \# \; 010 \gamma_0^{(1)} \; 100 \gamma_1^{(1)} \; 000 \gamma_2^{(1)} \; \$
      \]
      and we have check-marked the next position in both configurations.
      
      This concludes the definition of the automaton and the reader may verify that $\mathcal{T}'$ is indeed a \GAut since all individual parts are \GAuta. Furthermore, apart from the check-marking,
      no state has a non-identity transitions. Also note that $\$$ is not modified by any state.
      
      \paragraph{Definition of the State Sequences.}
      It remains to define the state sequences $\bm{p}_{i}$ that satisfy the conditions from the proposition. First, all state sequences $\bm{p}_i$ need to act trivially on words from $\Sigma_1^*$ and, after reading the first $\$$, we will either be in a state sequence from $\id^* r \id^*$ (this is the \enquote{successful} case) or from $\id^*$ (this is the \enquote{fail} case). Second, we need to satisfy the two implications. The idea is that we have some $u \in \Sigma_1^*$ which describes a computation of the Turing machine $M$ on input $w$. Each $\bm{p}_{i}$ verifies a certain aspect of the computation. If the Turing machine accepts (first implication), there is some $u$ encoding the valid and accepting computation and all verifications will pass (i.\,e.\ we end up in a state sequence from $\id^* r \id^*$). If the Turing machine does not accept the input $w$ (second implication), then no $u \in \Sigma_1^*$ can describe a valid and accepting computation and at least one verification will fail for all such $u$ (i.\,e.\ we end up in a state sequence from $\id^*$).
      
      We simply use the state $\bm{p}_{0} = z$ to verify that $u$ is from
      $(0^* \Gamma)^+ \left( \# (0^* \Gamma)^+ \right)^*$%Keep this in sync with the regex given above
      . Thus, we only need to consider the case that $u$ is of the form 
      \begin{equation}\label{eqn:formOfU}
        0^{\ell_0^{(0)}} \gamma_0^{(0)} \; 0^{\ell_1^{(0)}} \gamma_1^{(0)} \; \dots \; 0^{\ell_{L_0 - 1}^{(0)}} \gamma_{L_0 - 1}^{(0)} \; \# \dots \# \; 0^{\ell_0^{(T)}} \gamma_0^{(T)} \; 0^{\ell_1^{(T)}} \gamma_1^{(T)} \; \dots \; 0^{\ell_{L_T - 1}^{(T)}} \gamma_{L_T - 1}^{(T)}\tag{$\dagger$}
      \end{equation}
      with $\gamma_i^{(t)} \in \Gamma$ any further. Also observe that $\bm{p}_{0}$ acts trivially on all words from $\Sigma_1^* \$$ by construction.
      
      To test that $u$ encodes a valid and accepting computation, we need to verify that, for every $0 \leq i < s(n)$, we can check-mark the first $i$ positions. For this, we let
      \[
        \bm{p}_{1 + i} = \bm{c}_{i} = \checksub{\id}^{-(i + 1)} \checksub{r} \checksub{\id}^{i}
      \]
      as we have the cross diagram
      \begin{center}
        \resizebox{\linewidth}{!}{%
        \begin{tikzpicture}
          \matrix[matrix of math nodes,
            column sep=0pt, nodes={inner xsep=0pt, outer xsep=0pt},
            column 1/.style={nodes={outer xsep=0.5ex}},
            column 15/.style={nodes={outer xsep=0.5ex}},
            text height=1.75ex, text depth=0.25ex,
            ampersand replacement=\&] (m) {
              \& \revbin(0) \& \gamma_0^{(t)} \; \& \dots \; \& \revbin(0) \& \gamma_{i - 1}^{(t)} \; \& \revbin(0) \& \gamma_{i}^{(t)} \; \& \revbin(0) \& \gamma_{i + 1}^{(t)} \; \& \dots \; \& \revbin(0) \& \gamma_{L_t - 1}^{(t)} \; \& \# \textcolor{gray}{/\$} \& \\
            \checksub{\id}^{i} \&\&\&\&\&\&\&\&\&\&\&\&\&\& \checksub{\id}^{i} \textcolor{gray}{/\id^{i}} \\
              \& \revbin(i) \& \gamma_0^{(t)} \; \& \dots \; \& \revbin(1) \& \gamma_{i - 1}^{(t)} \; \& \revbin(0) \& \gamma_{i}^{(t)} \; \& \revbin(0) \& \gamma_{i + 1}^{(t)} \; \& \dots \; \& \revbin(0) \& \gamma_{L_t - 1}^{(t)} \; \& \# \textcolor{gray}{/\$} \& \\
            \checksub{r} \&\&\&\&\&\&\&\&\&\&\&\&\&\& \checksub{r} \textcolor{gray}{/r} \\
              \& \revbin(i + 1) \& \gamma_0^{(t)} \; \& \dots \; \& \revbin(2) \& \gamma_{i - 1}^{(t)} \; \& \revbin(1) \& \gamma_{i}^{(t)} \; \& \revbin(0) \& \gamma_{i + 1}^{(t)} \; \& \dots \; \& \revbin(0) \& \gamma_{L_t - 1}^{(t)} \; \& \# \textcolor{gray}{/\$} \& \\
            \checksub{\id}^{- 1} \&\&\&\&\&\&\&\&\&\&\&\&\&\& \checksub{\id}^{- 1} \textcolor{gray}{/\id^{- 1}} \\
              \& \revbin(i) \& \gamma_0^{(t)} \; \& \dots \; \& \revbin(1) \& \gamma_{i - 1}^{(t)} \; \& \revbin(0) \& \gamma_{i}^{(t)} \; \& \revbin(0) \& \gamma_{i + 1}^{(t)} \; \& \dots \; \& \revbin(0) \& \gamma_{L_t - 1}^{(t)} \; \& \# \textcolor{gray}{/\$} \& \\
            \checksub{\id}^{-i} \&\&\&\&\&\&\&\&\&\&\&\&\&\& \checksub{\id}^{-i} \textcolor{gray}{/\id^{-i}} \\
              \& \revbin(0) \& \gamma_0^{(t)} \; \& \dots \; \& \revbin(0) \& \gamma_{i - 1}^{(t)} \; \& \revbin(0) \& \gamma_{i}^{(t)} \; \& \revbin(0) \& \gamma_{i + 1}^{(t)} \; \& \dots \; \& \revbin(0) \& \gamma_{L_t - 1}^{(t)} \; \& \# \textcolor{gray}{/\$} \& \\
          };
          
          \foreach \i in {2,4,6,8} {
            \draw[->] (m-\i-1) -> (m-\i-15);
            
            \draw[->] let
              \n1 = {int(\i-1)},
              \n2 = {int(1+\i)}
            in
              (m-\n1-8) -> (m-\n2-8);
          };
        \end{tikzpicture}}
        % TODO:
%        \todo[inline]{Why doesn't align=right work in nodes?}
      \end{center}
      where $\revbin(z)$ denotes the reverse/least significant bit first binary representation of $z$ (of sufficient length). In particular, $\bm{c}_{i}$ acts trivially on all words $u \in \Sigma_1^*$ since $\checksub{r}$ acts in the same way as $\checksub{\id}$ on such words (i.\,e.\ any change made is reverted later). Here, it is useful to observe that, if the $0$ block for $\gamma_j^{(t)}$ with $j \leq i$ is not long enough to count to its required value (including the case that it is empty), then we will always end up in $\id$ after reading a $\$$. The same happens if $L_t < i + 1$ (i.\,e.\ if one of the configurations is \enquote{too short}). So this guarantees, $L_t \geq s(n)$ for all $t$.
      
      On the other hand, we use
      \[
        \bm{p}_{1 + s(n)} = \bm{c}' = \checksub{\id}^{-s(n)} c\, \checksub{\id}^{s(n)}
      \]
      to ensure that, after check-marking the first $s(n)$ positions in every configurations, all symbols have been check-marked (i.\,e.\ that no configuration is \enquote{too long}), which guarantees $L_t = s(n)$ for all $t$. Again, $\bm{c}'$ does not change words from $\Sigma_1^*$.
      
      Now that we have ensured that the word is of the correct form and we can count high enough for our check-marking, we need to actually verify that the $\gamma_i^{(t)}$ constitute a valid computation of the Turing machine with the initial configuration $\gamma_0' \dots \gamma_{s(n) - 1}' = p_0 w \blank^{s(n) - n - 1}$ for the input word $w$. To do this, we define
      \[
        \bm{p}_{2 + s(n) + i} = \bm{q}_{i} = \checksub{\id}^{-i} q_{\gamma_i'} \checksub{\id}^{i}
      \]
      for every $0 \leq i < s(n)$ as we have the cross diagram
      \begin{center}
        \resizebox{\linewidth}{!}{%
        \begin{tikzpicture}
          \matrix[matrix of math nodes, 
            column sep=0pt, nodes={inner xsep=0pt, outer xsep=0pt}, 
            column 1/.style={nodes={outer xsep=0.5ex}},
            column 15/.style={nodes={outer xsep=0.5ex}},
            text height=1.75ex, text depth=0.25ex,
            ampersand replacement=\&] (m) {
              \& \revbin(0) \& \gamma_0^{(t)} \; \& \dots \; \& \revbin(0) \& \gamma_{i - 1}^{(t)} \; \& \revbin(0) \& \gamma_{i}^{(t)} \; \& \revbin(0) \& \gamma_{i + 1}^{(t)} \; \& \dots \; \& \revbin(0) \& \gamma_{L_t - 1}^{(t)} \; \& \# \textcolor{gray}{/\$} \& \\
            \checksub{\id}^{i} \&\&\&\&\&\&\&\&\&\&\&\&\&\& \checksub{\id}^{i} \textcolor{gray}{/\id^{i}} \\
              \& \revbin(i) \& \gamma_0^{(t)} \; \& \dots \; \& \revbin(1) \& \gamma_{i - 1}^{(t)} \; \& \revbin(0) \& \gamma_{i}^{(t)} \; \& \revbin(0) \& \gamma_{i + 1}^{(t)} \; \& \dots \; \& \revbin(0) \& \gamma_{L_t - 1}^{(t)} \; \& \# \textcolor{gray}{/\$} \& \\
            q_{\gamma_i'} \&\&\&\&\&\&\&\&\&\&\&\&\&\& q_{\tau(\gamma_{i - 1}^{(t)}, \gamma_{i}^{(t)}, \gamma_{i + 1}^{(t)})} \textcolor{gray}{/r} \\
              \& \revbin(i) \& \gamma_0^{(t)} \; \& \dots \; \& \revbin(1) \& \gamma_{i - 1}^{(t)} \; \& \revbin(0) \& \gamma_{i}^{(t)} \; \& \revbin(0) \& \gamma_{i + 1}^{(t)} \; \& \dots \; \& \revbin(0) \& \gamma_{L_t - 1}^{(t)} \; \& \# \textcolor{gray}{/\$} \& \\
            \checksub{\id}^{-i} \&\&\&\&\&\&\&\&\&\&\&\&\&\& \checksub{\id}^{-i} \textcolor{gray}{/\id^{-i}} \\
              \& \revbin(0) \& \gamma_0^{(t)} \; \& \dots \; \& \revbin(0) \& \gamma_{i - 1}^{(t)} \; \& \revbin(0) \& \gamma_{i}^{(t)} \; \& \revbin(0) \& \gamma_{i + 1}^{(t)} \; \& \dots \; \& \revbin(0) \& \gamma_{L_t - 1}^{(t)} \; \& \# \textcolor{gray}{/\$} \& \\
          };
          
          \foreach \i in {2,4,6} {
            \draw[->] (m-\i-1) -> (m-\i-15);
            
            \draw[->] let
              \n1 = {int(\i-1)},
              \n2 = {int(1+\i)}
            in
              (m-\n1-8) -> (m-\n2-8);
          };
        \end{tikzpicture}}%
%        \todo[inline]{Why doesn't align=right work in nodes?}
      \end{center}
      if $\gamma_i^{(t)}$ is the expected $\gamma_i'$. Otherwise (if $\gamma_i^{(t)} \neq \gamma_i'$), we always end in the state $\id$ after reading the first $\$$. Finally, to ensure that the computation is not only valid but also accepting, we use the state $\bm{p}_{2 + 2s(n)} = f$. Both, $\bm{q}_{i}$ and $f$ do not change words from $\Sigma_1^*$.
      
      Finally, we observe that we can compute all $\bm{p}_{i}$ with $0 \leq i < D = 3 + 2s(n)$ in logarithmic space on input $w \in \Lambda^*$.
      
      \paragraph{The Two Implications.}
      For the first implication, we assume that the Turing machine $M$ accepts on the initial configuration $\blank p_0 w \blank^{s(n) - n - 1} \blank$. Let
      \[
        \gamma_0^{(0)} \dots \gamma_{s(n) - 1}^{(0)} \vdash \gamma_0^{(1)} \dots \gamma_{s(n) - 1}^{(1)} \vdash \dots \vdash \gamma_0^{(T)} \dots \gamma_{s(n) - 1}^{(T)}
      \]
      be the corresponding computation with $\gamma_0^{(0)} = p_0$, $\gamma_1^{(0)} \dots \gamma_{n}^{(0)} = w$ and $\gamma_i^{(T)} \in F$ for some $0 \leq i < s(n)$. We choose $\ell = \lceil \log(s(n)) \rceil + 1$ and define
      \[
        u = 0^\ell \gamma_0^{(0)} \dots 0^\ell \gamma_{s(n) - 1}^{(0)} \# 0^\ell \gamma_0^{(1)} \dots 0^\ell \gamma_{s(n) - 1}^{(1)} \# \dots \# 0^\ell \gamma_0^{(T)} \dots 0^\ell \gamma_{s(n) - 1}^{(T)} \in \Sigma_1^* \text{.}
      \]
      The reader may verify that we have the cross diagram depicted in \cref{fig:acceptingCrossDiagram} for this choice of $u$ (we only have to combine the cross diagrams given above for the individual $\bm{p}_{i}$). This shows the first implication.
      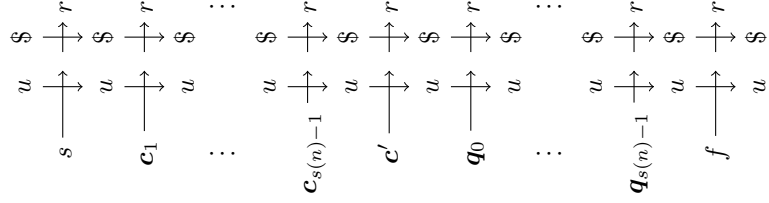
\begin{figure}[b]
        \centering%
        \begin{sideways}%
          \begin{tikzpicture}
            \matrix[matrix of math nodes, ampersand replacement=\&,
                    text height=1.75ex, text depth=0.25ex] (m) {
                            \& u \&    \& \$ \&        \\
                          z \&   \& {} \&    \& r \\
                            \& u \&    \& \$ \&        \\
                 \bm{c}_{0} \&   \& {} \&    \& r \\
                            \& u \&    \& \$ \&        \\
                     \vdots \&   \&    \&    \& \vdots \\
                            \& u \&    \& \$ \&        \\
          \bm{c}_{s(n) - 1} \&   \& {} \&    \& r \\
                            \& u \&    \& \$ \&        \\
                    \bm{c}' \&   \& {} \&    \& r \\
                            \& u \&    \& \$ \&        \\
                 \bm{q}_{0} \&   \& {} \&    \& r \\
                            \& u \&    \& \$ \&        \\
                     \vdots \&   \&    \&    \& \vdots \\
                            \& u \&    \& \$ \&        \\
           \bm{q}_{s(n)- 1} \&   \& {} \&    \& r \\
                            \& u \&    \& \$ \&        \\
                          f \&   \& {} \&    \& r \\
                            \& u \&    \& \$ \&        \\
            };
            
            \foreach \j in {1,3,7,9,11,15,17} {
              \foreach \i in {1,3} {
                \draw[->] let
                  \n1 = {int(2+\i)},
                  \n2 = {int(1+\j)}
                in
                  (m-\n2-\i) -> (m-\n2-\n1);
                \draw[->] let
                  \n1 = {int(1+\i)},
                  \n2 = {int(2+\j)}
                in
                  (m-\j-\n1) -> (m-\n2-\n1);
              };
            };
          \end{tikzpicture}%
        \end{sideways}%
        \caption{Cross diagram for the $\bm{p}_{i, r}$}\label{fig:acceptingCrossDiagram}
      \end{figure}
      
      For the second implication, assume that no valid computation of $M$ on the initial configuration $\blank p_0 w \blank^{s(n) - n - 1} \blank$ contains an accepting state from $F$ and consider an arbitrary word $u \in \Sigma_1^*$. If $u$ is not of the form
      $(0^* \Gamma)^+ \left( \# (0^* \Gamma)^+ \right)^*$%Keep this in sync with the regex given above
      , we have the cross diagram
      \begin{center}
        \begin{tikzpicture}[baseline=(m-2-1.base)]
          \matrix[matrix of math nodes, ampersand replacement=\&,
                  text height=1.75ex, text depth=0.25ex] (m) {
                \& u \&    \& \$ \&     \\
            \bm{p}_{0} = z \&   \& {} \&    \& \id \\
                \& u \&    \& \$ \&     \\
          };
          \foreach \j in {1} {
            \foreach \i in {1,3} {
              \draw[->] let
                \n1 = {int(2+\i)},
                \n2 = {int(1+\j)}
              in
                (m-\n2-\i) -> (m-\n2-\n1);
              \draw[->] let
                \n1 = {int(1+\i)},
                \n2 = {int(2+\j)}
              in
                (m-\j-\n1) -> (m-\n2-\n1);
            };
          };
        \end{tikzpicture}
      \end{center}
      and, thus, satisfied the implication (with $i = 0$).
      
      Therefore, we assume $u$ to be of the form mentioned in \cref{eqn:formOfU} and use a similar argumentation for the remaining cases. If $u$ does not contain a state from $F$, then we end up in the state $\id$ after reading $\$$ for $f$. As $w$ is not accepted by the machine, this includes in particular all valid computations on the initial configuration $\blank p_0 w \blank^{s(n) - n - 1} \blank$. If one of the $0$ blocks in $u$ is too short to count to a value required for the check-marking (i.\,e.\ one $\ell_i^{(t)}$ is too small), then the corresponding $\bm{c}_{i}$ will go to a state sequence from $\id^*$. This is also true if one configuration is too short (i.\,e.\ $L_t < s(n)$ for some $t$). If one configuration is too long (i.\,e.\ $L_t > s(n)$), then this will be detected by $\bm{c}'$ as not all positions will be check-marked after check-marking all first $s(n)$ positions in every configuration. Finally, $\bm{q}_{i}$ yields a state sequence form $\id^*$ if $\gamma_i^{(0)}$ is not the correct symbol from the initial configuration or if we have $\gamma_i^{(t + 1)} \neq \tau(\gamma_{i - 1}^{(t)}, \gamma_i^{(t)}, \gamma_{i + 1}^{(t)})$ for some $t$ (where we let $\gamma_{-1}^{(t)} = \blank = \gamma_{s(n)}^{(t)}$).\qed
    \end{proof}
    \begin{remark}\labelx{rmk:sizeOfTPrime}
      The constructed automaton $\mathcal{T}'$ has $3 |\Gamma|^2 + |\Gamma| + 19$ states (including the special states $r$ and $\id$ and the additional $5$ states belonging to $\checksub{\id}$ in \cref{fig:checkmarkingAutomaton}) where $|\Gamma|$ is the sum of the number of states and the number of tape symbols for a Turing machine for a $\PSPACE$-complete problem.
    \end{remark}
    
    \paragraph{Encoding over Two Letters.}
    Eventually, the automaton $\mathcal{T}$ should operate over a binary alphabet $\Sigma$. We will achieve this by using an automaton group with a binary alphabet where we still can implement a $D$-ary logical conjunction using nested commutators. However, for now, we will keep things a bit more general and also consider larger alphabets. This will allow us to generally describe our approach for various groups.\footnote{including, in particular, $A_5$ from \cref{ex:conjunctionInA5}, which requires an alphabet of size of five.}
    
    Assume that $\Sigma$ contains at least two distinct symbols $\boxedZero$ and $\boxedOne$. We will use these two symbols to encode the computations of the Turing machine $M$. We let $\tilde{\Sigma} = \Sigma \setminus \boxedAlph$ and, without loss of generality, assume that $|\Gamma| = 2^L$ is a power of two and $\Gamma = \{ \gamma_0, \dots, \gamma_{|\Gamma| - 1} \}$. We interpret the symbols $\{ 0, 1, \#, \$ \} \cup \Gamma$ in the alphabet of $\mathcal{T}'$ (from \cref{prop:TMmode}) as the words
    \begin{align*}
      0 &= \boxedOne \boxedZero \boxedZero, & \# &= \boxedOne \boxedOne \boxedZero, \\
      1 &= \boxedOne \boxedZero \boxedOne, & \$ &= \boxedOne \boxedOne \boxedOne \quad\text{ and} & \gamma_i &= \boxedZero \inbox{\operatorname{bin}_L i}
    \end{align*}
    over $\boxedAlph$ where $\inbox{\operatorname{bin}_L i}$ is the binary representation of $i$ with length $L$, $\boxedZero$ as the zero digit and $\boxedOne$ as the one digit. Furthermore, we let $Y = \{ 0, 1, \# \} \cup \Gamma \subseteq \boxedAlph^*$ and $X = Y \cup \{ \$ \} \subseteq \boxedAlph^*$.
    \begin{remark}
      Note that every word in $X^* \subseteq \boxedAlph^*$ can be uniquely decomposed into a product of words from $X$, which means that $X$ is a \emph{code} (in the definition of \cite{berstel2010codes}). Since no word in $X$ is a prefix of another one, we obtain that $X$ is even a \emph{prefix code}.
    \end{remark}
    
    Not every word in $\Sigma^*$ is in $X^*$; not even every word in $\boxedAlph^*$. However, we have that every word in $\boxedAlph^*$ is a prefix of a word in $X^*$. This can be proven by using the following fact.
    \begin{fact}\labelx{fct:encodingIsSurjective}
      For our special choice of $0, 1, \#, \$$ and $\Gamma$, we have 
      \[
        (\PPre X) \boxedAlph \subseteq (\PPre X) \cup X \text{.}
      \]
    \end{fact}
    
    \pagebreak
    If a word is not a prefix of $Y^* \$ \Sigma^*$, it must contain a symbol from $\tilde{\Sigma}$:
    \begin{lemma}\labelx{lem:malformedWords}
      If $u \in \Sigma^*$ is not a prefix of a word in $Y^* \$ \Sigma^*$, then $u$ is in $Y^* (\PPre Y) \tilde{\Sigma} \Sigma^*$ (where $\tilde{\Sigma} = \Sigma \setminus \boxedAlph$).
    \end{lemma}
    \begin{proof}
      We factorize $u = u_1 u_2 a u_3$ with $u_1$ maximal in $Y^*$, $u_2$ maximal in $\PPre Y$ (both possibly empty), $a \in \Sigma$ and $u_3 \in \Sigma^*$. We show $a \not\in \boxedAlph$ by contradiction. So, assume $a \in \boxedAlph$. We have $u_2 a \in (\PPre Y) \boxedAlph \subseteq (\PPre X) \boxedAlph \subseteq (\PPre X) \cup X$ (where the last inclusion follows by \cref{fct:encodingIsSurjective}). Since we have $\PPre \$ = \{ \varepsilon, \boxedOne, \boxedOne\boxedOne \} = \PPre \#$, we obtain $\PPre X = \PPre Y$ and, combining this with the last statement, $u_2 a \in (\PPre Y) \cup X$. We cannot have $u_2 a \in \PPre Y$ since $u_2$ was chosen maximal with this property. Similarly, we cannot have $u_2 a \in Y$ since $u_1$ was chosen maximal. The only remaining case $u_2 a = \$$ is not possible either, however, since, then, $u$ was in $Y^* \$ \Sigma^*$, which violates the hypothesis.\qed
    \end{proof}
    
    We can now describe how we can encode the automaton $\mathcal{T}'$ from \cref{prop:TMmode} over $\Sigma$. For this, it is important that all transitions of $\mathcal{T}'$ are of a special form. Namely, the symbols $\#, \$$ and $\gamma_i \in \Gamma$ are not changed by any transition and $0$ and $1$ are either also not changed or they are toggled (i.\,e.\ we either have the transitions $\trans{p}{0}{0}{p \cdot 0}$ and $\trans{p}{1}{1}{p \cdot 1}$ or the transitions $\trans{p}{0}{1}{p \cdot 0}$ and $\trans{p}{1}{0}{p \cdot 1}$).
    
    The general idea to obtain the encoded automaton $\mathcal{T}_2 = (Q_2, \Sigma, \delta_2)$ from $\mathcal{T}' = (Q', \{ 0, 1, \#,\allowbreak \$ \} \cup \Gamma, \delta')$ is not to read $0, 1, \#, \$$ and the elements of $\Gamma$ as single symbols but to use a tree to read the prefixes of their encodings as words over $\boxedAlph$. When we know which encoding we have read, we move to the corresponding state (and reset the prefix of the current encoding); additionally, we possibly also toggle $0$ and $1$ if the corresponding state in the original automaton did this. Here comes the special encoding we have chosen above into play: we could not simply toggle $\gamma_0$ and $\gamma_{2^{L} - 1}$, for example, because it is not clear how to do this in a prefix-compatible way.
    
    To formalize this general idea, we consider the set $\PPre X = \{ \varepsilon, \boxedOne, \boxedOne \boxedZero,\allowbreak \boxedOne \boxedOne \, \}\allowbreak \cup \boxedZero \boxedAlph^{< L}$ and let $P' = Q' \setminus \{ r, \id \}$ and
    \begin{align*}
      Q_2 ={}& \{ (p', x) \mid p' \in P', x \in \PPre X \} \cup \{ \id, r \} \text{ and} \\
      \delta'' ={}& \{ \trans{\id}{a}{a}{\id} \mid a \in \Sigma \} \cup \{ \trans{r}{a}{a}{r} \mid a \in \Sigma \} \\
      {}\cup{}& \left\{ \trans{(p', x)}{\boxedZero}{\boxedZero}{(p', x\boxedZero)} \mid p' \in P', x\boxedZero \in \PPre X \right\} \\
      {}\cup{}& \left\{ \trans{(p', x)}{\boxedOne}{\boxedOne}{(p', x\boxedOne)} \mid p' \in P', x\boxedOne \in \PPre X \right\} \\
      {}\cup{}& \left\{ \trans{(p', \boxedOne\boxedZero)}{\boxedZero}{\inbox{p' \circ 0}}{(p' \cdot 0, \varepsilon)}, \trans{(p', \boxedOne\boxedZero)}{\boxedOne}{\inbox{p' \circ 1}}{(p' \cdot 1, \varepsilon)} \mid p' \in P' \right\} \\
      {}\cup{}& \left\{ \trans{(p', \boxedOne \boxedOne)}{\boxedZero}{\boxedZero}{(p' \cdot \#, \varepsilon)}, \trans{(p', \boxedOne \boxedOne)}{\boxedOne}{\boxedOne}{(p' \cdot \$, \varepsilon)} \mid p' \in P' \right\} \\
      {}\cup{}& \left\{ \trans{(p', x)}{\boxedZero}{\boxedZero}{(p' \cdot \gamma_i, \varepsilon)} \mid p' \in P', x\boxedZero = \gamma_i \right\} \\
      {}\cup{}& \left\{ \trans{(p', x)}{\boxedOne}{\boxedOne}{(p' \cdot \gamma_i, \varepsilon)} \mid p' \in P', x\boxedOne = \gamma_i \right\}
    \end{align*}
    where we use the convention $(r, \varepsilon) = r$ and $(\id, \varepsilon) = \id$. The first line in the definition of $\delta''$ extends $r$ and $\id$ into identities over $\Sigma$. The second and third line handle the case $x \boxedAlph \subseteq \PPre X$. The fourth and fifth line are for the case $x \boxedAlph \subseteq \{ 0, 1, \#, \$ \} \subseteq X$ and the last two lines are for $x \boxedAlph \subseteq \Gamma \subseteq X$. By \cref{fct:encodingIsSurjective}, this covers all cases and we have an outgoing transition with input $\boxedZero$ and one with input $\boxedOne$ for every $q \in Q_2$ (since $\mathcal{T}'$ must be a complete automaton over $\{ 0, 1, \#, \$ \} \cup \Gamma$). Therefore, to make $\mathcal{T}_2$ complete, we only have to handle the letters in $\tilde{\Sigma} = \Sigma \setminus \boxedAlph$ and, to this end, we let
    \[
      \delta_2 = \delta'' \cup \{ \trans{q}{a}{a}{\id} \mid q \in Q_2, a \in \tilde{\Sigma} \} \text{.}
    \]
    Observe that this turns $\mathcal{T}_2$ into a \GAut over $\Sigma$.
    
    \begin{example}
      Using the just described encoding method, a schematic part
      
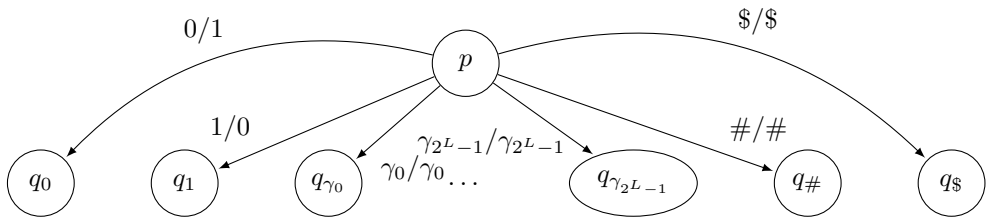
\begin{center}
        \begin{tikzpicture}[auto, shorten >=1pt, >=latex, baseline=(p.base)]
          
          \node[state] (qg0) {$q_{\gamma_0}$};
          \node[right=0.5cm of qg0] (dots) {$\dots$};
          \node[state, ellipse, right=0.5cm of dots] (qgn) {$q_{\gamma_{2^L - 1}}$};
          
          \node[state, right=of qgn] (q0) {$q_0$};
          \node[state, right=of q0] (q1) {$q_1$};
          \node[state, right=of q1] (qs) {$q_{\#}$};
          \node[state, right=of qs] (qd) {$q_{\$}$};
          
          \coordinate (pc) at ($(qg0.west)!0.5!(qd.east)$);
          \node[state, above=of pc] (p) {$p$};
          
          \path[->] (p) edge node[swap, pos=0.8] {$0/1$} (q0)
                        edge node[swap, pos=0.8] {$1/0$} (q1)
                        edge[bend right] node[swap] {$\gamma_0/\gamma_0$} (qg0)
                        edge[bend right] node[swap, pos=0.9] {$\gamma_{2^L-1}/\gamma_{2^L-1}$} (qgn)
                        edge node[pos=0.8] {$\#/\#$} (qs)
                        edge[bend left] node {$\$/\$$} (qd)
          ;
        \end{tikzpicture}
      \end{center}
      of the automaton $\mathcal{T}'$ yields the part
      \begin{center}
        \resizebox{\linewidth}{!}
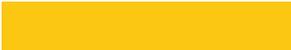
{
        \begin{tikzpicture}[auto, shorten >=1pt, >=latex]
          \node[state, ellipse] (qg0) {$(q_{\gamma_0}, \varepsilon)$};
          \node[state, ellipse, right=0.25cm of qg0] (qg1) {$(q_{\gamma_1}, \varepsilon)$};
          \node[right=0cm of qg1] (dots) {$\dots$};
          \node[state, ellipse, right=0cm of dots] (qgn) {$(q_{\gamma_{2^L - 1}}, \varepsilon)$};
          \node[state, ellipse, right=0.25cm of qgn] (q0) {$(q_0, \varepsilon)$};
          \node[state, ellipse, right=0.25cm of q0] (q1) {$(q_1, \varepsilon)$};
          \node[state, ellipse, right=0.25cm of q1] (qs) {$(q_{\#}, \varepsilon)$};
          \node[state, ellipse, right=0.25cm of qs] (qd) {$(q_{\$}, \varepsilon)$};
          
          \node[state, ellipse] at ($(qg0.center)!0.5!(qg1.center)+(0pt,1.75cm)$) (p0L) {$(p, \boxedZero\boxedZero^{L - 1})$};
          \node[state, ellipse] at ($(dots.center)!0.5!(qgn.center)+(0pt,1.75cm)$) (p1L) {$(p, \boxedZero\boxedOne^{L - 1})$};
          \path[->] (p0L) edge node[swap, pos=0.75] {$\boxedZero/\boxedZero$} (qg0)
                          edge node[pos=0.75] {$\boxedOne/\boxedOne$} (qg1)
                    (p1L) edge (dots)
                          edge node[pos=0.75] {$\boxedOne/\boxedOne$} (qgn)
          ;
          
          \node[above=0pt of p0L] (lvdots) {$\vdots$};
          \node[above=0pt of p1L] (rvdots) {$\vdots$};
          \node[state, ellipse, anchor=south] at ($(lvdots.north)!0.5!(rvdots.north)+(0pt,0pt)$) (p0) {$(p, \boxedZero)$};
          \path[->] (p0) edge node[swap] {$\boxedZero/\boxedZero$} (lvdots)
                         edge node {$\boxedOne/\boxedOne$} (rvdots)
          ;
          
          \node[state, ellipse] at ($(q0.center)!0.5!(q1.center)+(0pt,1.75cm)$) (p10) {$(p, \boxedOne\boxedZero)$};
          \node[state, ellipse] at ($(qs.center)!0.5!(qd.center)+(0pt,1.75cm)$) (p11) {$(p, \boxedOne\boxedOne)$};
          \path[->] (p10) edge node[swap, pos=0.75] {$\boxedZero/\boxedOne$} (q0)
                          edge node[pos=0.75] {$\boxedOne/\boxedZero$} (q1)
                    (p11) edge node[swap, pos=0.75] {$\boxedZero/\boxedZero$} (qs)
                          edge node[pos=0.75] {$\boxedOne/\boxedOne$} (qd)
          ;
          
          \coordinate (c) at ($(p10.center)!0.5!(p11.center)$);
          \node[state, ellipse] at (c|-p0) (p1) {$(p, \boxedOne)$};
          \path[->] (p1) edge node[swap] {$\boxedZero/\boxedZero$} (p10)
                         edge node {$\boxedOne/\boxedOne$} (p11)
          ;
          
          \node[state, ellipse] at ($(p0.center)!0.5!(p1.center)+(0pt,0.75cm)$) (p) {$(p, \varepsilon)$};
          \path[->] (p) edge node[swap] {$\boxedZero/\boxedZero$} (p0)
                        edge node {$\boxedOne/\boxedOne$} (p1)
          ;
        \end{tikzpicture}
        }
      \end{center}
      of the automaton $\mathcal{T}_2$ (where the possibly missing transitions all go to $\id$).
    \end{example}
    
    \begin{remark}\labelx{rmk:crossDiagarmsInEncoding}
      An important point to note on this encoding is that we have all cross diagrams from $\mathcal{T}'$ also in $\mathcal{T}_2$ (if we identify $p' \in P'$ with $(p', \varepsilon)$). In particular, we still have the statements from \cref{prop:TMmode} about the words $u \in \Sigma_1^*$ also for $\mathcal{T}_2$ and the words from $Y^*$.
    \end{remark}

    Additionally, the encoding ensures that we always end in $\id$ after reading a word from $Y^* (\PPre Y) \tilde{\Sigma}$. Since no $\bm{p}_i$ from \cref{prop:TMmode} changes a word from $u \in \Sigma_1^*$ in $\mathcal{T}'$, this shows the following fact.
    \begin{fact}\labelx{fct:malformedWords}
      \makebox[0pt][l]{We have}\hspace*{\fill}
%      \begin{center}
        \begin{tikzpicture}[baseline=(m-1-2.base)]
          \matrix[matrix of math nodes, ampersand replacement=\&,
                  text height=1.75ex, text depth=0.25ex] (m) {
                       \& u \&    \& a \&     \\
            \bm{p}_{i} \&   \& {} \&   \& \id^{|\bm{p}_i|} \\
                       \& u \&    \& a \&     \\
          };
          \foreach \j in {1} {
            \foreach \i in {1,3} {
              \draw[->] let
                \n1 = {int(2+\i)},
                \n2 = {int(1+\j)}
              in
                (m-\n2-\i) -> (m-\n2-\n1);
              \draw[->] let
                \n1 = {int(1+\i)},
                \n2 = {int(2+\j)}
              in
                (m-\j-\n1) -> (m-\n2-\n1);
            };
          };
        \end{tikzpicture}
%      \end{center}\csname @beginparpenalty\endcsname10000
      \hspace*{\fill}\\
      in $\mathcal{T}_2$ for all $u \in Y^* \PPre Y$, $a \in \tilde{\Sigma}$ and $0 \leq i < D$.
    \end{fact}
    
    \begin{remark}\labelx{rmk:sizeOfT2}
      The encoded automaton $\mathcal{T}_2$ has $(|\mathcal{T}'| - 2)(|\Gamma| + 3) + 2$ many states where $|\Gamma|$ is the sum of the number of states and the number of tape symbols for a Turing machine for a $\PSPACE$-complete problem (which we assume to be a power of two) and $|\mathcal{T}'|$ is the size of $\mathcal{T}'$ from \cref{rmk:sizeOfTPrime}. This yields $3 |\Gamma|^3 + 10 |\Gamma|^2 + 20 |\Gamma| + 53$ many states in total, where $2$ states are $\id$ and $r$ and $5 \cdot (|\Gamma| + 3) = 5 |\Gamma| + 15$ additional states belong to (the encoding of) $\checksub{\id}$.
    \end{remark}

    \paragraph{The Commutator Mode.}
    For the commutator mode, we fix an arbitrary \GAut $\mathcal{R} = (R, \Sigma, \rho)$ which admits \DLINSPACE-computable functions $\alpha, \beta: \mathbb{N} \to R^{\pm *}$ (where the input is given in binary) and a \LOGSPACE-computable function $b$
    \function
      {a number $D$ in unary with $D = 2^d$ for some $d$\newline
       (i.\,e.\ the string $1^D$) and\newline
       a number $0 \leq i < D$ in binary}
      {$b(D, i) \in R^{\pm *}$}\noindent
    with
    \[
      B_{\beta, \alpha}[ b(D, D - 1), \dots, b(D, 0) ] \neq \idGrp \text{ in } \mathscr{G}(\mathcal{R})
    \]
    for all (positive) powers $D$ of two. Note that the choice of $\mathcal{R}$ implies $|\Sigma| \geq 2$.
    
    Similar to $A_5$ in \cref{ex:conjunctionInA5}, we use this group to simulate a logical conjunction. For state sequences $\bm{p}_0, \dots, \bm{p}\sub{D - 1} \in R^{\pm *}$ with $\bm{p}_i = \idGrp$ or $\bm{p}_i = b(D, i)$ in $\mathscr{G}(\mathcal{R})$ for all $0 \leq i < D$, we have $B[ \bm{p}\sub{D - 1}, \dots, \bm{p}_0 ] \neq \idGrp$ in $\mathscr{G}(\mathcal{R})$ if and only if we have $\bm{p}_i \neq \idGrp$ in $\mathscr{G}(\mathcal{R})$ for all $0 \leq i < D$.
    
    \begin{example}\labelx{ex:commutatorInA5}
      One possible choice for $\mathcal{R}$ is to continue using the group $A_5$ from \cref{ex:conjunctionInA5} where we have shown $B_{\beta, \alpha}[\sigma, \dots, \sigma] = \sigma \neq \idGrp$ for the special elements $\alpha, \beta, \sigma \in A_5$. The group $A_5$ is an automaton group as it is generated, e.\,g., by the \GAut with states $A_5$, alphabet $\{ 1, 2, \dots, 5 \}$ and transitions
      \[
        \left\{ \trans{\pi}{i}{\pi(i)}{\id} \mid \pi \in A_5 \right\}
      \]
      and we can simply let $b(D, i) = \sigma$ for all $0 \leq i < D$, which obviously make $b$ \LOGSPACE-computable. Of course, the two constant functions $\alpha$ and $\beta$ are \DLINSPACE-computable. Thus, we can choose $A_5$ for $\mathscr{G}(\mathcal{R})$.
    \end{example}
    \begin{example}\labelx{ex:freeGroup}
      Another possibility for $\mathcal{R}$ is the first Aleshin automaton
      \begin{center}\enlargethispage{1.5\baselineskip}
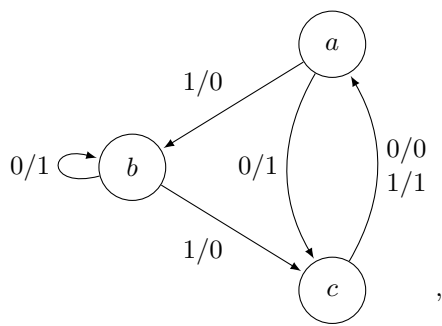
        \begin{tikzpicture}[auto, shorten >=1pt, >=latex, baseline=(c.base)]
          \node[state] (b) {$b$};
          \node[state, above right=0.5cm and 2cm of b] (a) {$a$};
          \node[state, below right=0.5cm and 2cm of b] (c) {$c$};
          
          \draw[->] (b) edge[loop left] node {$0/1$} (b)
                        edge node[swap] {$1/0$} (c)
                    (a) edge[bend right] node[swap] {$0/1$} (c)
                        edge node[swap] {$1/0$} (b)
                    (c) edge[bend right] node[swap, align=center] {$0/0$\\$1/1$} (a)
          ;
        \end{tikzpicture},
      \end{center}\csname @beginparpenalty\endcsname10000
      which generates the free group $F_3$ in the generators $a, b$ and $c$ with a binary alphabet \cite{Aleshin83,VorobetsV07}.\footnote{For the idea to use the free group for a logical conjunction, see also \cite{Robinson93phd}.}
      
      We claim that we have $B_{\beta, \varepsilon}[b(D, D - 1), \dots, b(D, 0)] \neq \idGrp$ in $F_3$ if we choose
      \[
        \beta(d) = 
          \begin{cases}
            c & \text{for $d$ even,} \\
            b & \text{for $d$ odd}
          \end{cases}
      \]
      and $b(D, i) = b^{-1} a$ for all $0 \leq i < D$. To show the claim, we write $B_3(D)$ for
      \[
        B_3(D) = B_{\beta, \varepsilon}[\smash{\underbrace{b^{-1} a, \dots, b^{-1} a}_{D \text{ times}}}]
      \]
      where $D = 2^d$ and show
      \[
        B_3(2^d) \in \begin{cases}
          b^{-1} \, \{ a, b, c \}^{\pm *} \, a & \text{for $d$ even,}\\
          c^{-1} \, \{ a, b, c \}^{\pm *} \, a & \text{for $d$ odd}
        \end{cases}
      \]
      and that $B_3(2^d)$ is freely reduced in both cases by induction (on $d$). From this, we obtain that all $B_3(2^d)$ are freely reduced but non-empty and, therefore, not the identity in $F_3$.
      
      For $d = 0$ (or, equivalently, $D = 1$), we have
      \[
        B_3(2^0) = B_3(1) = B_{\beta, \varepsilon} [ b^{-1} a ] = b^{-1} a\text{,}
      \]
      which is in $b^{-1} \, \{ a, b, c \}^{\pm *} \, a$ and freely reduced. For the inductive step from $d$ to $d + 1$ (or, equivalently, from $D$ to $2D$), we have:
      \begin{alignat*}{2}
        B_3(2^{d + 1}) &= B_3(2D) = B_{\beta, \varepsilon}[\underbrace{b^{-1} a, \dots, b^{-1} a}_{2D \text{ times}}] &\quad& \text{(definition of $B_3$)}\\
        &= \Big[ B_{\beta, \varepsilon}[\underbrace{b^{-1} a, \dots, b^{-1} a}_{D \text{ times}}]^{\beta(d)}, \\
        &\phantom{{}={} \Big[}
          B_{\beta, \varepsilon}[\underbrace{b^{-1} a, \dots, b^{-1} a}_{D \text{ times}}]^{\varepsilon} \Big] &&\text{(definition of $B_{\beta, \varepsilon}$)}\\
        &= \big[ B_3(2^d)^{\beta(d)}, B_3(2^d) \big] && \text{(definition of $B_3$)}\\
        &= \begin{array}[t]{r@{\;}c@{\;}l@{\;}c}
          \beta(d)^{-1} & B_3(2^d)^{-1} & \beta(d) & B_3(2^d)^{-1} \\
          \beta(d)^{-1} & B_3(2^d)      & \beta(d) & B_3(2^d)
        \end{array}
      \end{alignat*}
      We distinguish two cases. If $d$ is even (and, equivalently, $d + 1$ is odd), we have $\beta(d) = c$ (by definition) and $B_3(2^d) = b^{-1} w a$ freely reduced for some $w$ (by induction). Thus, we obtain
      \[
        B_3(2^{d + 1}) = c^{-1} (a^{-1} w^{-1} b) c \; (a^{-1} w^{-1} b) \; c^{-1} (b^{-1} w a) c \; (b^{-1} w a) \text{,}
      \]
      which is freely reduced and in $c^{-1} \, \{ a, b, c \}^{\pm *} \, a$ (as desired since $d + 1$ is odd). The other case is that $d$ is odd (and $d + 1$ is even). Here, we have $\beta(d) = b$ and $B_3(2^d) = c^{-1} w a$ freely reduced for some $w$. This yields
      \[
        B_3(2^{d + 1}) = b^{-1} (a^{-1} w^{-1} c) b \; (a^{-1} w^{-1} c) \; b^{-1} (c^{-1} w a) b \; (c^{-1} w a) \text{,}
      \]
      which is again freely reduced and in $b^{-1} \, \{ a, b, c \}^{\pm *} \, a$ (as desired since $d + 1$ is even).
      
      On a side note, we point out that this argument only requires that all $b(D, i)$ are from $b^{-1} \{ a, b, c \}^{\pm *} a$, not that they are all equal or even all equal to $b^{-1} a$ (as we have chosen them here). This will become more important later on in \cref{sec:compressedWP} because it allows us to also use $b(D, i)$ which are nested commutators of the from $B_{\beta, \varepsilon}$ themselves.
    \end{example}
    
    \begin{remark}
      Instead of using $B_{\beta, \alpha}$ for \DLINSPACE-computable functions $\alpha$ and $\beta$, we could restrict ourselves to $B_{\varepsilon, \varepsilon}$. The idea here is that we can move the conjugation to the leaves of the tree representing the nested commutators.
      
      A schematic representation of such a tree can be found in \cref{fig:commutatorTree} for $D = 2^3$,
      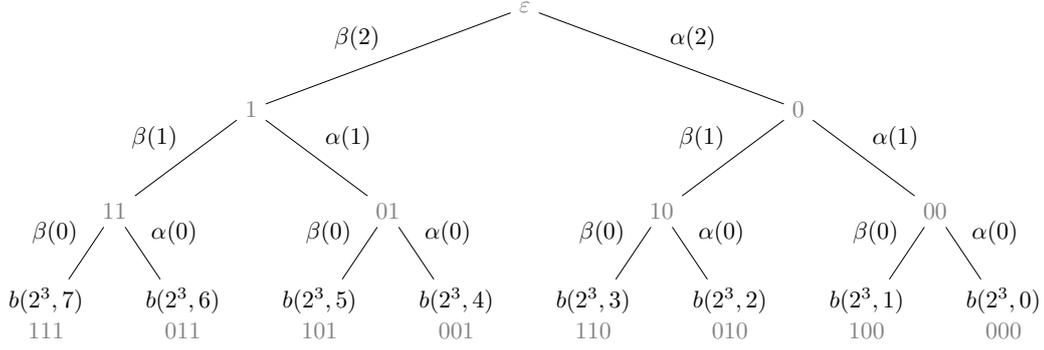
\begin{figure}\centering
        \resizebox{\linewidth}{!}{
          \begin{tikzpicture}[level distance=1.5cm, auto,
              level 1/.style={sibling distance=8cm},
              level 2/.style={sibling distance=4cm},
              level 3/.style={sibling distance=2cm},
              level 4/.style={sibling distance=1cm}]
            \node {\textcolor{gray}{$\varepsilon$}}
              child {node {\textcolor{gray}{1}}
                child {node {\textcolor{gray}{11}}
                  child {
                    node[align=center] {$b(2^3, 7)$\\\textcolor{gray}{111}}
                    edge from parent node[swap] {$\beta(0)$}
                  }
                  child {
                    node[align=center] {$b(2^3, 6)$\\\textcolor{gray}{011}}
                    edge from parent node {$\alpha(0)$}
                  }
                  edge from parent node[swap] {$\beta(1)$}
                }
                child {node {\textcolor{gray}{01}}
                  child {
                    node[align=center] {$b(2^3, 5)$\\\textcolor{gray}{101}}
                    edge from parent node[swap] {$\beta(0)$}
                  }
                  child {
                    node[align=center] {$b(2^3, 4)$\\\textcolor{gray}{001}}
                    edge from parent node {$\alpha(0)$}
                  }
                  edge from parent node {$\alpha(1)$}
                }
                edge from parent node[swap] {$\beta(2)$}
              }
             child {node {\textcolor{gray}{0}}
               child {node {\textcolor{gray}{10}}
                 child {
                   node[align=center] {$b(2^3, 3)$\\\textcolor{gray}{110}}
                   edge from parent node[swap] {$\beta(0)$}
                  }
                  child {
                    node[align=center] {$b(2^3, 2)$\\\textcolor{gray}{010}}
                    edge from parent node {$\alpha(0)$}
                  }
                  edge from parent node[swap] {$\beta(1)$}
                }
                child {node {\textcolor{gray}{00}}
                  child {
                    node[align=center] {$b(2^3, 1)$\\\textcolor{gray}{100}}
                    edge from parent node[swap] {$\beta(0)$}
                  }
                  child {
                    node[align=center] {$b(2^3, 0)$\\\textcolor{gray}{000}}
                    edge from parent node {$\alpha(0)$}
                  }
                  edge from parent node {$\alpha(1)$}
                }
                edge from parent node {$\alpha(2)$}
              };
          \end{tikzpicture}
        }%
        \caption{Schematic representation of $B_{\beta, \alpha}[ b(2^3, 7), \dots, b(2^3, 0) ]$. The labels at the edges indicate a conjugation.}\label{fig:commutatorTree}
      \end{figure}
      where the labels at the edges indicate a conjugation. From the picture, it becomes apparent that the conjugating element for $b(D, i)$ with $D = 2^d$ is coupled to the reverse/least significant bit first binary representation of $i$ with length $d$.
      
      Formally, we can define $w(D, i) = f(\revbin_d(i))$ for $D = 2^d$ for the function $f$ that works as follows: the first letter of $\revbin_d(i)$ is replaced with $\alpha(0)$ if it is a zero digit or with $\beta(0)$ if it is a one digit, the second letter is replaced with $\alpha(1)$ (if it is a zero digit) or with $\beta(1)$ (if it is a one digit), and so on. The counter for the positions in $\revbin_d (i)$, which is the argument to $\alpha$ and $\beta$, needs to count up to $d$ and, thus, requires space $\log \log D$. Therefore, $w(D, i)$ can certainly be computed in \LOGSPACE (in particular, if $D$ is given in unary).
      
      For the group $A_5$ (from \cref{ex:commutatorInA5}), we can replace all zero digits by $\alpha$ and all one digits by $\beta$ and, for $F_3$ (from \cref{ex:freeGroup}), the function $f$ is implemented by the automaton\footnote{Technically, the automaton is not synchronous because we use the empty word in some outputs. However, we can emit a special symbol $e$ instead and then map all $e$s to $\varepsilon$ using a homomorphism.}
      \begin{center}
        \begin{tikzpicture}[auto, shorten >=1pt, >=latex, baseline=(p.base)]
          \node[state] (p) {$f$};
          \node[state, right=2cm of p] (q) {};
          
          \draw[->] (p) edge[bend left] node[align=center] {$1/c$, $0/\varepsilon$} (q)
                    (q) edge[bend left] node[align=center] {$1/b$, $0/\varepsilon$} (p)
          ;
        \end{tikzpicture}.
      \end{center}
      
      Now, we define $b'(D, i) = b(D, i)^{w(D, i)}$ (which clearly is still \LOGSPACE-computable) and show
      \[
        B_{\varepsilon, \varepsilon} \big[ b'(D, D - 1), \dots, b'(D, 0) \big] = 
          B_{\beta, \alpha} \big[ b(D, D - 1), \dots, b(D, 0) \big]
      \]
      for all $D = 2^d$. In fact, we have
      \[
        B_{\varepsilon, \varepsilon} \big[ \bm{p}\sub{D - 1}^{w(D, D - 1)}, \dots, \bm{p}_0^{w(D, 0)} \big] = 
          B_{\beta, \alpha} \big[\bm{p}\sub{D - 1}, \dots, \bm{p}_0 \big]
      \]
      for all $\bm{p}_0, \dots, \bm{p}\sub{D - 1} \in R^{\pm *}$, which we show by induction on $d$. For $d = 0$ (or, equivalently, $D = 1$), we have $B_{\varepsilon, \varepsilon} [ \bm{p}_0^{w(1, 0)} ] = \bm{p}_0^\varepsilon = \bm{p}_0 = B_{\beta, \alpha} [ \bm{p}_0 ]$. For the inductive step from $d$ to $d + 1$ (or, equivalently, form $D$ to $2D$), we observe
%      \[
%        w(2D, i) = f\left( \revbin_{d + 1}(i) \right) =
%        \begin{cases}
%          f\left( \revbin_{d}(i) \, 0 \right) = w(D, i) \alpha(d) & \text{for $0 \leq i < D$}\\
%          f\left( \revbin_{d}(i - D) \, 1 \right) = w(D, i - D) \beta(d) & \text{for $D \leq i < 2D$.}
%        \end{cases}
%      \]
%      
      \begin{alignat*}{2}
        w(2D, i) &= f\left( \revbin_{d + 1}(i) \right) = f\left( \revbin_{d}(i) \, 0 \right) = w(D, i) \alpha(d)\\
      \intertext{for all $0 \leq i < D$ and}
        w(2D, i) &= f\left( \revbin_{d + 1}(i) \right) = f\left( \revbin_{d}(i - D) \, 1 \right) = w(D, i - D) \beta(d)
      \end{alignat*}
      for all $D \leq i < 2D$.
      Thus, we have in $\mathscr{G}(\mathcal{R})$:
      \begin{multline*}
        B_{\varepsilon, \varepsilon} \big[ \bm{p}_{2D - 1}^{w(2D, 2D - 1)}, \dots, \bm{p}_0^{w(2D, 0)} \big] \\
          \begin{alignedat}[t]{2}
            &= \Big[
              B_{\varepsilon, \varepsilon} \big[ \bm{p}_{2D - 1}^{w(2D, 2D - 1)}, \dots, \bm{p}\sub{D}^{w(2D, D)} \big], \\
            &\phantom{{}={}\Big[}
              B_{\varepsilon, \varepsilon} \big[ \bm{p}\sub{D - 1}^{w(2D, D - 1)}, \dots, \bm{p}_0^{w(2D, 0)} \big] \Big]
              &\quad& \text{(by definition)} \\
            &= \Big[
              B_{\varepsilon, \varepsilon} \big[ \bm{p}_{2D - 1}^{w(D, D - 1) \beta(d)}, \dots, \bm{p}\sub{D}^{w(D, D) \beta(d)} \big], \\
            &\phantom{{}= \Big[}
              B_{\varepsilon, \varepsilon} \big[ \bm{p}\sub{D - 1}^{w(D, D - 1) \alpha(d)}, \dots, \bm{p}_0^{w(D, 0) \alpha(d)} \big] \Big]
              && \text{(by the above)} \\
            &= \Big[
              B_{\varepsilon, \varepsilon} \big[ \bm{p}_{2D - 1}^{w(D, D - 1)}, \dots, \bm{p}\sub{D}^{w(D, D)} \big]^{\beta(d)}, \\
            &\phantom{{}={}\Big[}
              B_{\varepsilon, \varepsilon} \big[ \bm{p}\sub{D - 1}^{w(D, D - 1)}, \dots, \bm{p}_0^{w(D, 0)} \big]^{\alpha(d)} \Big]
              && \text{(by \cref{fct:conjugatedCommutator})} \\
            &= \Big[
              B_{\beta, \alpha} \big[ \bm{p}_{2D - 1}, \dots, \bm{p}\sub{D} \big]^{\beta(d)}, \,
              B_{\beta, \alpha} \big[ \bm{p}\sub{D - 1}, \dots, \bm{p}_0 \big]^{\alpha(d)} \Big]
              && \text{(by induction)} \\
            &= B_{\beta, \alpha} \big[ \bm{p}_{2D - 1}, \dots, \bm{p}_0 \big]
              && \text{(by definition)}
          \end{alignedat}
      \end{multline*}
    \end{remark}
    \begin{example}\labelx{ex:SENSandGrigorchuk}
      In fact, we can choose any automaton group that satisfies the uniformly SENS property of \cite{BartholdiFLW20} for $\mathscr{G}(\mathcal{R})$. These include, for example, the Grigorchuk group generated by the automaton
      
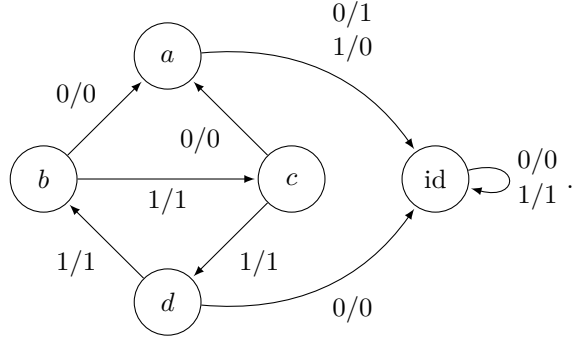
\begin{center}
        \begin{tikzpicture}[auto, shorten >=1pt, >=latex, baseline=(b.base)]
          \node[state] (b) {$b$};
          \node[state, above right=of b] (a) {$a$};
          \node[state, below right=of b] (d) {$d$};
          \node[state, below right=of a] (c) {$c$};
          \node[state, right=of c] (id) {$\textnormal{id}$};
          
          \draw[->] (a) edge[bend left] node[align=center, pos=0.7] {$0/1$\\$1/0$} (id)
                    (b) edge node {$0/0$} (a)
                    (b) edge node[swap] {$1/1$} (c)
                    (c) edge node {$0/0$} (a)
                    (c) edge node {$1/1$} (d)
                    (d) edge[bend right] node[swap] {$0/0$} (id)
                    (d) edge node {$1/1$} (b)
                    (id) edge[loop right] node[align=center] {$0/0$\\$1/1$} (id)
          ;
        \end{tikzpicture}.
      \end{center}
      The definition of a uniformly SENS group is very similar but slightly stronger than what we require for our group $\mathscr{G}(\mathcal{R})$ for the commutator mode:

      A group $G$ generated by a finite set $R$ is called \emph{uniformly strongly efficiently non-solvable} \emph{(uniformly SENS)}\footnote{Our definition of uniformly SENS is not exactly the one of \cite{BartholdiFLW20}: we have changed the ordering of some indices here because it is more convenient for our other definitions.} if there is a constant $\mu \in \mathbb{N}$ and words $\bm{r}_{d,v} \in R^{\pm*}$ for all $d \in \mathbb{N}$, $v \in \{ 0,1 \}^{\leq d}$ such that
      \begin{enumerate}[label=(\alph*), ref=(\alph*)]
        \item\label{SENSa} $|\bm{r}_{d,v}| \leq 2^{\mu d}$ for all $v \in \{ 0,1 \}^{d}$,
        \item $\bm{r}_{d,v} = \bigl[ \bm{r}_{d, 1v},\, \bm{r}_{d, 0v} \bigr]$ for all $v \in \{ 0,1 \}^{< d}$ (here we take the commutator of words)\footnote{Compare this to the tree in \cref{fig:commutatorTree}.},
        \item $\bm{r}_{d, \varepsilon} \neq \idGrp$ in $G$ and
        \item\label{SENSu} given $v \in \{0,1\}^d$, a positive integer $i$ encoded in binary with $\mu d$ bits, and $a \in R^\pm$ one can decide in $\DLINTIME$ whether the $i^\textnormal{th}$ letter of $\bm{r}_{d,v}$ is $a$. Here, \DLINTIME is the class of problems decidable in linear time on a random access Turing machine.
      \end{enumerate}
      
      Essentially, a group is uniformly SENS if there are non-trivial balanced iterated commutators of arbitrary depth and these balanced iterated commutators can be computed efficiently.
      
      To define the elements $b(D, i)$ for a uniformly SENS automaton group, we let $b(D, i) = \bm{r}_{d, \revbin_d(i)}$ where $D = 2^d$ and $\revbin_d(i)$ is the reverse/least significant bit first binary representation of $i$ with length $d$. For this choice, we have $B[ b(D, D - 1), \dots, B(D, 0) ] \neq \idGrp$ in the group (and, thus, choose $\alpha = \beta = \varepsilon$, which is clearly \DLINSPACE-computable).
      
      It remains to show that $b$ is \LOGSPACE-computable (where $D$ is given in unary). Observe that the last condition \ref{SENSu} requires that each letter of $\bm{r}_{d, v}$ can be computed in time $\mu d$ on a random access Turing machine. Thus, it can, in particular, be computed in space $\mu d$ on a normal (non-random access) Turing machine. Since, by the first condition \ref{SENSa}, its length is at most $2^{\mu d}$, we only need a counter of size $\mu d$ to compute $\bm{r}_{d, v}$ entirely. Thus, we can compute $b(D, i)$ in space $\mu \log D$, which is logarithmic in the input as $D$ is given in unary (i.\,e.\ as a string $1^D$).
    \end{example}
    
    \begin{proof}[Proof of \cref{thm:nonuniformPSPACE}]
      Since the uniform word problem for automaton groups is in \PSPACE (see \cref{thm:uniformPSPACE}), so is the word problem of any (fixed) automaton group. Therefore, we only have to show the hardness part of the result.
      
      As already stated at the beginning of this section, we reduce the \PSPACE-complete word problem for $M$ to the (complement of the) word problem for a \GAut $\mathcal{T}$ with state set $Q$. For this reduction, it remains to finally construct $\mathcal{T}$ and to map an input $w \in \Lambda^*$ for the Turing machine $M$ to a state sequence $\bm{q} \in Q^{\pm *}$ in such a way that $\bm{q}$ can be computed from $w$ in logarithmic space and that we have $\bm{q} = \idGrp$ in $\mathscr{G}(\mathcal{T})$ if and only if the Turing machine does \textbf{not} accept $w$.
      
      We first define the \GAut $\mathcal{T} = (Q, \Sigma, \delta)$, which is composed of multiple parts. Its first part is the automaton $\mathcal{R} = (R, \Sigma, \rho)$. If we want to show \cref{thm:nonuniformPSPACE} for $|\Sigma| = 2$, we have to choose a suitable $\mathcal{R}$ over the binary alphabet $\Sigma = \boxedAlph$ (see \cref{ex:freeGroup} and \cref{ex:SENSandGrigorchuk} for such choices). However, we will continue the proof without this assumption and also allow other groups (such as $A_5$ from \cref{ex:commutatorInA5}).
      
      \begin{figure}\centering%
        \begin{tikzpicture}[auto, shorten >=1pt, >=latex]
          \node[state] (T21) {};
          \node[state, above right=0.25cm and 2cm of T21, dashed] (T21r) {$r$};
          \node[state, right=6cm of T21] (r1) {$r_1$};
          \draw[->] (T21) edge[dashed] node[sloped] {$a/b$} (T21r)
                          edge node[swap, pos=0.4] {$a/b$} (r1)
                    (T21r) edge[loop right, dashed] node (T21rlabel) {$\id_\Sigma$} (T21r)
          ;
          \draw[decorate, decoration=pencilline, thick, opacity=0.5]
            ($(T21r.south west)-(0.1cm, 0.1cm)$) -- ($(T21r.north east)+(0.1cm, 0.1cm)$)
            ($(T21r.north west)-(0.1cm, -0.1cm)$) -- ($(T21r.south east)+(0.1cm, -0.1cm)$);
          \coordinate (T21southeast) at ($(T21.south -| T21rlabel.east)+(0.25cm, -0.25cm)$);
          \coordinate (T21northwest) at ($(T21.west |- T21r.north)+(-0.25cm, 0.25cm)$);
          \draw (T21northwest) rectangle (T21southeast);
          \node[anchor=north west, inner sep=7pt] (T21label) at (T21northwest) {$\mathcal{T}_2$};

          \node[state, below=2cm of T21] (T2R) {};
          \node[state, above right=0.25cm and 2cm of T2R, dashed] (T2Rr) {$r$};
          \node[state, right=6cm of T2R] (rR) {$r_{|R|}$};
          \draw[->] (T2R) edge[dashed] node[sloped] {$a'/b'$} (T2Rr)
                          edge node[swap, pos=0.4] {$a'/b'$} (rR)
                    (T2Rr) edge[loop right, dashed] node (T2Rrlabel) {$\id_\Sigma$} (T2Rr)
          ;
          \draw[decorate, decoration=pencilline, thick, opacity=0.5]
            ($(T2Rr.south west)-(0.1cm, 0.1cm)$) -- ($(T2Rr.north east)+(0.1cm, 0.1cm)$)
            ($(T2Rr.north west)-(0.1cm, -0.1cm)$) -- ($(T2Rr.south east)+(0.1cm, -0.1cm)$);
          \coordinate (T2Rsoutheast) at ($(T2R.south -| T2Rrlabel.east)+(0.25cm, -0.25cm)$);
          \coordinate (T2Rnorthwest) at ($(T2R.west |- T2Rr.north)+(-0.25cm, 0.25cm)$);
          \draw (T2Rnorthwest) rectangle (T2Rsoutheast);
          \node[anchor=north west, inner sep=7pt] (T2Rlabel) at (T2Rnorthwest) {$\mathcal{T}_2$};

          \coordinate (T21south) at ($(T21.south west)!0.5!(T21southeast)$);
          \node[anchor=center, yshift=3pt] (dotsT2) at ($(T21southeast)!0.5!(T2Rnorthwest)$) {$\vdots$};
          \node[anchor=center, yshift=3pt] (dotsR) at (dotsT2.center -| r1.center) {$\vdots$};
          
          \coordinate (Rnorthwest) at ($(T21label.north west -| r1.west)+(-0.25cm, 0cm)$);
          \node[anchor=north west, inner sep=7pt] (Rlabel) at (Rnorthwest) {$\mathcal{R}$};
          \draw (Rnorthwest) rectangle ($(rR.south -| rR.east)+(0.25cm, -0.25cm)$);
        \end{tikzpicture}
        \caption{The dedicated state $r$ gets replaced by an actual state of $\mathcal{R}$ in each copy $\mathcal{T}_{2, r}$ of $\mathcal{T}_2$.}\label{fig:T2r}
      \end{figure}
      
      Next, we take the automaton $\mathcal{T}'$ from \cref{prop:TMmode} and encode it into the automaton $\mathcal{T}_2$ over $\Sigma$. Then, for every $r \in R$, we take a disjoint copy $\mathcal{T}_{2, r}$ of $\mathcal{T}_2$. Each copy $\mathcal{T}_{2, r}$ contains the place-holder state $r$ (mentioned in \cref{prop:TMmode}) and we replace it with the actual state $r$ from $\mathcal{R}$ which belongs to the respective copy (see \cref{fig:T2r}). Thus, in general, the action of $r$ is not the identity anymore. By $\bm{p}_{i, r}$, we denote the corresponding $\bm{p}_i$ from \cref{prop:TMmode} for $\mathcal{T}_{2, r}$.
      
      Finally, we consider the \GAut $\mathcal{R}_0$
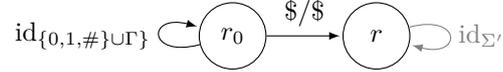
      \begin{center}
        \begin{tikzpicture}[auto, shorten >=1pt, >=latex, baseline=(r.base)]
          \node[state] (r0) {$r_0$};
          \node[state, right=of r0] (r) {$r$};
          
          \draw[->] (r0) edge[loop left] node {$\id_{ \{ 0, 1, \# \} \cup \Gamma }$} (r0)
                         edge node {$\$/\$$} (r)
                    (r) edge[gray, loop right] node {$\id_{\Sigma'}$} (r)
          ;
        \end{tikzpicture}
      \end{center}
      over the alphabet $\Sigma' = \{ 0, 1, \#, \$ \} \cup \Gamma$ (where we use the same convention about edges labeled by $\id_{A}$ as in the proof of \cref{prop:TMmode}). We encode $\mathcal{R}_0$ over $\Sigma$ (similar to the encoding $\mathcal{T}_2$ of $\mathcal{T}'$) by using the automaton
      \begin{center}
        \resizebox{\linewidth}{!}{%
        \begin{tikzpicture}[auto, shorten >=1pt, >=latex, baseline=(r0.base)]
          \node[state] (L) {};
          \node[right=of L] (dots) {$\dots$};
          \node[state, right=of dots] (0) {};
          \node[state, right=of 0] (r0) {$r_0$};
          \node[state, right=of r0] (1) {};
          \node[state, right=of 1] (2) {};
          \node[state, above=of 1] (3) {};
          \node[state, right=of 2, dotted] (r) {$r$};
          
          \path[->] (L) edge[bend left] node[align=center] {$\boxedZero/\boxedZero$\\$\boxedOne/\boxedOne$} (r0)
                    (dots) edge node[align=center] {$\boxedZero/\boxedZero$\\$\boxedOne/\boxedOne$} (L)
                    (0) edge node[align=center] {$\boxedZero/\boxedZero$\\$\boxedOne/\boxedOne$} (dots)
                    (r0) edge node {$\boxedZero/\boxedZero$} (0)
                         edge node {$\boxedOne/\boxedOne$} (1)
                    (1) edge node {$\boxedOne/\boxedOne$} (2)
                        edge node[swap] {$\boxedZero/\boxedZero$} (3)
                    (2) edge node {$\boxedOne/\boxedOne$} (r)
                        edge[bend left=50] node {$\boxedZero/\boxedZero$} (r0)
                    (3) edge[bend right] node[swap, align=center] {$\boxedZero/\boxedZero$\\$\boxedOne/\boxedOne$} (r0)
          ;
          
          \draw[decorate, decoration={brace}] ($(0.east)+(0pt,-1cm)$) -- node[yshift=-0.5ex] {$L = \log |\Gamma|$ many} ($(L.west)+(0pt,-1cm)$);
        \end{tikzpicture}}
      \end{center}\pagebreak[1]
      and, again, take a disjoint copy $\mathcal{R}_{0, r}$ for every $r$ where we replace the place-holder $r$ by the actual state from $\mathcal{R}$. These parts of $\mathcal{T}$ will be used to implement the conjugation with $\alpha(d)$ and $\beta(d)$ in $B_{\beta, \alpha}$ (just like in the proof of \cref{thm:uniformPSPACE}). For this, we define the functions $\alpha_0, \beta_0: \mathbb{N} \to Q^{\pm *}$. We let $\alpha_0(d)$ (respectively: $\beta_0(d)$) be the same as $\alpha(d)$ (respectively: $\beta(d)$) with the only difference being that we replace every $r \in R$ by the corresponding state $r_0$ from the appropriate copy of $\mathcal{R}_{0, r}$. Clearly, $\alpha_0$ and $\beta_0$ are \DLINSPACE-computable (since $\alpha$ and $\beta$ are). As an abbreviation, we write $B$ for $B_{\beta, \alpha}$ and $B_0$ for $B_{\beta_0, \alpha_0}$ in the rest of this proof.
      
      This completes the definition of $\mathcal{T}$ and it remains to define the state sequence $\bm{q}$ depending on $w \in \Lambda^*$. For this, we first compute (in logarithmic space) all state sequences $\bm{p}_{i, r}$ with $0 \leq i < D$ and $r \in R$ from \cref{prop:TMmode} (with respect to the appropriate copy $\mathcal{T}_{2, r}$) using $w$ as the input. Here, we may assume that $D = 2^d$ is a power of two (otherwise, we repeat $\bm{p}_{D - 1, r}$ as a new $\bm{p}_{D, r}$ for all $r \in R$ until we reach a power of two, which is possible in logarithmic space). Then, we compute in logarithmic space the elements $\bm{b}_0 = b(D, 0), \dots, \bm{b}\sub{D - 1} = b(D, D - 1) \in R^{\pm *}$ such that $B[\bm{b}\sub{D - 1}, \dots, \bm{b}_0] \neq \idGrp$ in $\mathscr{G}(\mathcal{R})$ (which is possible by the choice of $\mathcal{R}$ above). Note that $\mathscr{G}(\mathcal{R})$ is a subgroup of $\mathscr{G}(\mathcal{T})$ since $\mathcal{R}$ is a sub-automaton of $\mathcal{T}$. Thus, we have $B[\bm{b}\sub{D - 1}, \dots, \bm{b}_0] \neq \idGrp$ in $\mathscr{G}(\mathcal{T})$. Now, for every $0 \leq i < D$, we can compute $\bm{b}_i' = \bm{p}_{i, r_{\ell}} \ldots \bm{p}_{i, r_{1}}$ in logarithmic space from $\bm{b}_i = r_{\ell} \dots r_{1}$ with $r_{1}, \dots, r_{\ell} \in R^{\pm 1}$. Finally, we choose $\bm{q} = B_0[\bm{b}\sub{D - 1}', \dots, \bm{b}_0']$ (which can also be computed in logarithmic space by \cref{fct:BIsLogspaceComputable}).
      
      We need to show $\bm{q} = \idGrp$ in $\mathscr{G}(\mathcal{T})$ if and only if $M$ does \textbf{not} accept $w$. First, assume that $M$ accepts $w$ and consider an arbitrary $0 \leq i < D$. We have $\bm{b}_i = r_{\ell} \dots r_{1}$ for some $r_{1}, \dots, r_{\ell} \in R^{\pm 1}$ and $\bm{b}_i' = \bm{p}_{i, r_{\ell}} \ldots \bm{p}_{i, r_{1}}$. By \cref{prop:TMmode} (first implication, applied to the appropriate copies of $\mathcal{T}'$) and by \cref{rmk:crossDiagarmsInEncoding}, there is some $u \in Y^*$ such that we have the cross diagram\footnote{Strictly speaking, we do not have the states $r_i$ on the right but rather state sequences from $\id^* r_i \id^*$. However, we omit these $\id$ states from the cross diagram to keep it more readable.}
      \begin{center}
        \begin{tikzpicture}[baseline=(m-6-1.base), auto]
          \matrix[matrix of math nodes, ampersand replacement=\&,
                  text height=1.75ex, text depth=0.25ex] (m) {
                            \& u \&    \& \$ \&     \\
            \bm{p}_{i, r_1} \&   \& {} \&    \& r_1 \\
                            \& u \&    \& \$ \&     \\
                     \vdots \&   \&    \&    \& \vdots \\
                            \& u \&    \& \$ \&     \\
            \bm{p}_{i, r_\ell} \&   \& {} \&    \& r_\ell \\
                            \& u \&    \& \$ \&     \\
          };
          \foreach \j in {1,5} {
            \foreach \i in {1,3} {
              \draw[->] let
                \n1 = {int(2+\i)},
                \n2 = {int(1+\j)}
              in
                (m-\n2-\i) -> (m-\n2-\n1);
              \draw[->] let
                \n1 = {int(1+\i)},
                \n2 = {int(2+\j)}
              in
                (m-\j-\n1) -> (m-\n2-\n1);
            };
          };
          
          \draw[decorate, decoration={brace}] (m-6-1.south west) -- node{$\bm{b}_i' ={}$} (m-2-1.north west);
          \draw[decorate, decoration={brace}] (m-2-5.north east) -- node{${} = \bm{b}_i$} (m-6-5.south east);
        \end{tikzpicture}.
      \end{center}
      
      Combining these cross diagrams (for all $0 \leq i < D$), we obtain the black part of the cross diagram
      \begin{center}
        \begin{tikzpicture}[baseline=(m-6-1.base), auto]
          \matrix[matrix of math nodes, ampersand replacement=\&,
                  text height=1.75ex, text depth=0.25ex] (m) {
                            \& u \&    \& \$ \&     \\
                \bm{b}_{0}' \&   \& {} \&    \& \bm{b}_0 \\
                            \& u \&    \& \$ \&     \\
                     \vdots \&   \&    \&    \& \vdots \\
                            \& u \&    \& \$ \&     \\
            \bm{b}_{D - 1}' \&   \& {} \&    \& \bm{b}_{D - 1} \\
                            \& u \&    \& \$ \&     \\
          };
          \foreach \j in {1,5} {
            \foreach \i in {1,3} {
              \draw[->] let
                \n1 = {int(2+\i)},
                \n2 = {int(1+\j)}
              in
                (m-\n2-\i) -> (m-\n2-\n1);
              \draw[->] let
                \n1 = {int(1+\i)},
                \n2 = {int(2+\j)}
              in
                (m-\j-\n1) -> (m-\n2-\n1);
            };
          };
          
          \node[gray, rotate=90, below=0pt of m-6-1.south, anchor=east, inner sep=0pt] (B0) {$B_0[$};
          \node[gray, rotate=90, above=0pt of m-2-1.north, anchor=east, inner sep=0pt] (B0c) {$]$};
          \foreach \j in {2,4} {
            \path let
              \n1 = {int(2+\j)}
            in
              node[gray, rotate=90, anchor=base] at ($(m-\j-1)!0.5!(m-\n1-1)$ |- B0.base) {$,$};
          };
          \node[gray, rotate=90, below=0pt of m-6-5.south, anchor=east, inner sep=0pt] (B) {$B[$};
          \node[gray, rotate=90, above=0pt of m-2-5.north, anchor=east, inner sep=0pt] {$]$};
          \foreach \j in {2,4} {
            \path let
              \n1 = {int(2+\j)}
            in
              node[gray, rotate=90, anchor=base] at ($(m-\j-5)!0.5!(m-\n1-5)$ |- B.base) {$,$};
          };
          
          \draw[gray, decorate, decoration={brace}] ($(B0.north west)+(-8pt,0pt)$) -- node{$\bm{q} ={}$} ($(B0c.north east)+(-8pt,0pt)$);
        \end{tikzpicture}.
      \end{center}
      Since we have $r_0 \cdot u \$ = r$ and $r_0 \circ u \$ = u \$$ by the construction of $\mathcal{R}_0$, we also have the cross diagrams
      \begin{center}
        \begin{tikzpicture}[baseline=(m-6-1.base), auto]
          \matrix[matrix of math nodes, ampersand replacement=\&,
                  text height=1.75ex, text depth=0.25ex] (m) {
                        \& u \&    \& \$ \&          \\
            \alpha_0(d) \&   \& {} \&    \& \alpha(d)\\
                        \& u \&    \& \$ \&          \\
          };
          \foreach \j in {1} {
            \foreach \i in {1,3} {
              \draw[->] let
                \n1 = {int(2+\i)},
                \n2 = {int(1+\j)}
              in
                (m-\n2-\i) -> (m-\n2-\n1);
              \draw[->] let
                \n1 = {int(1+\i)},
                \n2 = {int(2+\j)}
              in
                (m-\j-\n1) -> (m-\n2-\n1);
            };
          };
        \end{tikzpicture}
      \end{center}
      for all $d$ and can add the commutators to the above cross diagram by \cref{fct:commutatorInCrossDiagrams} to obtain the gray additions. As we have $B[\bm{b}\sub{D - 1}, \dots, \bm{b}_0] \neq \idGrp$ in $\mathscr{G}(\mathcal{T})$, there must be some $v \in \Sigma^*$ such that
      \[
        \bm{q} \circ u \$ v = B_0[\bm{b}\sub{D - 1}', \dots, \bm{b}_0'] \circ u \$ v = u \$ (B[\bm{b}\sub{D - 1}, \dots, \bm{b}_0] \circ v) \neq u \$ v \text{,}
      \]
      which concludes this direction.
      
      For the other direction, assume that $M$ does not accept the input $w$. We have to show $\bm{q} \circ u' = B_0[\bm{b}\sub{D - 1}', \dots, \bm{b}_0'] \circ u' = u'$ for all $u' \in \Sigma^*$. First, we show this for all $u' \in Y^* \$ \Sigma^*$ (and, thus, also for all prefixes of such $u'$) where we let $u' = u \$ v$ with $u \in Y^*$ and $v \in \Sigma^*$. Our approach is similar to the one given above for the case that $M$ accepts $w$. Again, consider an arbitrary $0 \leq i < D$ with $\bm{b}_i = r_{\ell} \dots r_{1}$ for some $r_{1}, \dots, r_{\ell} \in R^{\pm 1}$ and, thus, $\bm{b}_i' = \bm{p}_{i, r_{\ell}} \ldots \bm{p}_{i, r_{1}}$. We have the cross diagram
      \begin{center}
        \begin{tikzpicture}[baseline=(m-6-1.base), auto]
          \matrix[matrix of math nodes, ampersand replacement=\&,
                  text height=1.75ex, text depth=0.25ex] (m) {
                            \& u \&    \& \$ \&     \\
            \bm{p}_{i, r_1} \&   \& {} \&    \& \bm{p}_{i, r_1} \cdot u \$ \\
                            \& u \&    \& \$ \&     \\
                     \vdots \&   \&    \&    \& \vdots \\
                            \& u \&    \& \$ \&     \\
            \bm{p}_{i, r_\ell} \&   \& {} \&    \& \bm{p}_{i, r_\ell} \cdot u \$ \\
                            \& u \&    \& \$ \&     \\
          };
          \foreach \j in {1,5} {
            \foreach \i in {1,3} {
              \draw[->] let
                \n1 = {int(2+\i)},
                \n2 = {int(1+\j)}
              in
                (m-\n2-\i) -> (m-\n2-\n1);
              \draw[->] let
                \n1 = {int(1+\i)},
                \n2 = {int(2+\j)}
              in
                (m-\j-\n1) -> (m-\n2-\n1);
            };
          };
          
          \draw[decorate, decoration={brace}] (m-6-1.south west) -- node{$\bm{b}_i' ={}$} (m-2-1.north west);
          \draw[decorate, decoration={brace}] (m-2-5.north east) -- node{${} = \bm{b}_i' \cdot u \$$} (m-6-5.south east);
        \end{tikzpicture}
      \end{center}
      for this $i$ by \cref{prop:TMmode}. However, this time, there is some, $0 \leq i < D$ with $\bm{b}_i' \cdot u \$ \in \id^*$ and, thus, equal to $\idGrp$ in $\mathscr{{G}(\mathcal{T})}$ (this follows from applying the second implication of \cref{prop:TMmode} to all $r_1, \dots, r_\ell$).\footnote{In fact, by the construction of $\mathcal{T}'$, we actually either have $\bm{b}_i' \cdot u \$ = \bm{b}_i$ or $\bm{b}_i' \cdot u \$ = \idGrp$ in $\mathscr{G}(\mathcal{T})$ for all $0 \leq i < D$.}
      
      Again, we can combine these cross diagrams to obtain the black part of\csname @beginparpenalty\endcsname10000
      \begin{center}
        \begin{tikzpicture}[baseline=(m-6-1.base), auto]
          \matrix[matrix of math nodes, ampersand replacement=\&,
                  text height=1.75ex, text depth=0.25ex] (m) {
                        \& u \&    \& \$ \&     \\
            \bm{b}_{0}' \&   \& {} \&    \& \bm{b}_{0}' \cdot u \$ \\
                        \& u \&    \& \$ \&     \\
            \vdots \&   \&    \&    \& \vdots \\
                        \& u \&    \& \$ \&     \\
            \bm{b}\sub{D - 1}' \&   \& {} \&    \& \bm{b}\sub{D - 1}' \cdot u \$ \\
                        \& u \&    \& \$ \&     \\
          };
          \foreach \j in {1,5} {
            \foreach \i in {1,3} {
              \draw[->] let
                \n1 = {int(2+\i)},
                \n2 = {int(1+\j)}
              in
                (m-\n2-\i) -> (m-\n2-\n1);
              \draw[->] let
                \n1 = {int(1+\i)},
                \n2 = {int(2+\j)}
              in
                (m-\j-\n1) -> (m-\n2-\n1);
            };
          };
          
          \node[gray, rotate=90, below=0pt of m-6-1.south, anchor=east, inner sep=0pt] (B0) {$B_0[$};
          \node[gray, rotate=90, above=0pt of m-2-1.north, anchor=east, inner sep=0pt] (B0c) {$]$};
          \foreach \j in {2,4} {
            \path let
              \n1 = {int(2+\j)}
            in
              node[gray, rotate=90, anchor=base] at ($(m-\j-1)!0.5!(m-\n1-1)$ |- B0.base) {$,$};
          };
          \node[gray, rotate=90, below=0pt of m-6-5.south, anchor=east, inner sep=0pt] (B) {$B\Big[$};
          \node[gray, rotate=90, above=2pt of m-2-5.north, anchor=east, inner sep=0pt] {$\Big]$};
          \foreach \j in {2,4} {
            \path let
              \n1 = {int(2+\j)}
            in
              node[gray, rotate=90, anchor=base] at ($(m-\j-5)!0.5!(m-\n1-5)$ |- B.base) {$,$};
          };
          
          \draw[gray, decorate, decoration={brace}] ($(B0.north west)+(-8pt,0pt)$) -- node{$\bm{q} ={}$} ($(B0c.north east)+(-8pt,0pt)$);
        \end{tikzpicture}
      \end{center}
      and, then, add the gray part using \cref{fct:commutatorInCrossDiagrams}. Since we have $\bm{b}_{i}' \cdot u \$ = \idGrp$ in $\mathscr{G}(\mathcal{T})$ for some $i$, we obtain (for the state sequence on the right) $B[\bm{b}\sub{D - 1}' \cdot u \$, \dots, \bm{b}_{0}' \cdot u \$] = \idGrp$ in $\mathscr{G}(\mathcal{T})$ by \cref{fct:commutatorCollapses}. Thus, we have $\bm{q} \circ u \$ v = B_0[\bm{b}\sub{D - 1}', \dots, \bm{b}_{0}'] \circ u \$ v = u \$ v$.
      
      Now, assume that $u'$ is not a prefix of a word in $Y^* \$ \Sigma^*$. Then, it is in $Y^* (\PPre Y) \tilde{\Sigma} \Sigma^*$ by \cref{lem:malformedWords} and we can factorize $u' = u a v$ for some $u \in Y^* (\PPre Y)$, $a \in \tilde{\Sigma}$ and $v \in \Sigma^*$. By \cref{fct:malformedWords}, we obtain\footnote{Again, we actually have state sequences from $\id^*$, rather, but omit them for readability.}
      \begin{center}
        \begin{tikzpicture}[baseline=(m-2-1.base), auto]
          \matrix[matrix of math nodes, ampersand replacement=\&,
                  text height=1.75ex, text depth=0.25ex] (m) {
                            \& u \&    \& a \&        \& v \& \\
            \bm{p}_{i, r_1} \&   \& {} \&   \& \varepsilon \&   \& \varepsilon \\
                            \& u \&    \& a \&        \& v \& \\
                     \vdots \&   \&    \&   \& \vdots \&   \& \vdots \\
                            \& u \&    \& a \&        \& v \& \\
            \bm{p}_{i, r_\ell} \&   \& {} \&    \& \varepsilon \& \& \varepsilon \\
                            \& u \&    \& a \&        \& v \& \\
          };
          \foreach \j in {1,5} {
            \foreach \i in {1,3,5} {
              \draw[->] let
                \n1 = {int(2+\i)},
                \n2 = {int(1+\j)}
              in
                (m-\n2-\i) -> (m-\n2-\n1);
              \draw[->] let
                \n1 = {int(1+\i)},
                \n2 = {int(2+\j)}
              in
                (m-\j-\n1) -> (m-\n2-\n1);
            };
          };
          
          \draw[decorate, decoration={brace}] (m-6-1.south west) -- node{$\bm{b}_i' ={}$} (m-2-1.north west);
        \end{tikzpicture}
      \end{center}
      for $0 \leq i < D$ with $\bm{b}_i = r_\ell \dots r_1$ (where $r_1, \dots, r_\ell \in R^{\pm 1}$). Using the same argumentation, we also obtain $r_0 \circ ua = ua$ and $r_0 \cdot u a = \idGrp$ in $\mathscr{G}(\mathcal{T})$ for all $r_0$ in the copies of $\mathcal{R}_{0, 2}$ and, thus, for all $d$, the cross diagrams
      \begin{center}
        \begin{tikzpicture}[baseline=(m-6-1.base), auto, baseline=(m-2-1.base)]
          \matrix[matrix of math nodes, ampersand replacement=\&,
                  text height=1.75ex, text depth=0.25ex] (m) {
                        \& u \&    \& a \&        \& v \& \\
            \alpha_0(d) \&   \& {} \&   \& \varepsilon \&   \& \varepsilon \\
                        \& u \&    \& a \&        \& v \& \\
          };
          \foreach \j in {1} {
            \foreach \i in {1,3,5} {
              \draw[->] let
                \n1 = {int(2+\i)},
                \n2 = {int(1+\j)}
              in
                (m-\n2-\i) -> (m-\n2-\n1);
              \draw[->] let
                \n1 = {int(1+\i)},
                \n2 = {int(2+\j)}
              in
                (m-\j-\n1) -> (m-\n2-\n1);
            };
          };
        \end{tikzpicture}.
      \end{center}
      
      We can combine this into the cross diagram
      \begin{center}
        \begin{tikzpicture}[baseline=(m-6-1.base), auto]
          \matrix[matrix of math nodes, ampersand replacement=\&,
                  text height=1.75ex, text depth=0.25ex] (m) {
                        \& u \&    \& a \&        \& v \& \\
            \bm{b}_{0}' \&   \& {} \&   \& \varepsilon \&   \& \varepsilon \\
                        \& u \&    \& a \&        \& v \& \\
                 \vdots \&   \&    \&   \& \vdots \&   \& \vdots \\
                        \& u \&    \& a \&        \& v \& \\
            \bm{b}\sub{D - 1}' \&   \& {} \&    \& \varepsilon \&   \& \varepsilon \\
                        \& u \&    \& a \&        \& v \& \\
          };
          \foreach \j in {1,5} {
            \foreach \i in {1,3,5} {
              \draw[->] let
                \n1 = {int(2+\i)},
                \n2 = {int(1+\j)}
              in
                (m-\n2-\i) -> (m-\n2-\n1);
              \draw[->] let
                \n1 = {int(1+\i)},
                \n2 = {int(2+\j)}
              in
                (m-\j-\n1) -> (m-\n2-\n1);
            };
          };
          
          \node[gray, rotate=90, below=0pt of m-6-1.south, anchor=east, inner sep=0pt] (B0) {$B_0[$};
          \node[gray, rotate=90, above=0pt of m-2-1.north, anchor=east, inner sep=0pt] (B0c) {$]$};
          \foreach \j in {2,4} {
            \path let
              \n1 = {int(2+\j)}
            in
              node[gray, rotate=90, anchor=base] at ($(m-\j-1)!0.5!(m-\n1-1)$ |- B0.base) {$,$};
          };
          \node[gray, rotate=90, below=0pt of m-6-5.south, anchor=east, inner sep=0pt] (B) {$B_{\varepsilon, \varepsilon}[$};
          \node[gray, rotate=90, above=0pt of m-2-5.north, anchor=east, inner sep=0pt] {$]$};
          \foreach \j in {2,4} {
            \path let
              \n1 = {int(2+\j)}
            in
              node[gray, rotate=90, anchor=base] at ($(m-\j-5)!0.5!(m-\n1-5)$ |- B.base) {$,$};
          };
          
          \node[gray, rotate=90, below=0pt of m-6-7.south, anchor=east, inner sep=0pt] (B2) {$B_{\varepsilon, \varepsilon}[$};
          \node[gray, rotate=90, above=0pt of m-2-7.north, anchor=east, inner sep=0pt] {$]$};
          \foreach \j in {2,4} {
            \path let
              \n1 = {int(2+\j)}
            in
              node[gray, rotate=90, anchor=base] at ($(m-\j-7)!0.5!(m-\n1-7)$ |- B2.base) {$,$};
          };
          
          \draw[gray, decorate, decoration={brace}] ($(B0.north west)+(-8pt,0pt)$) -- node{$\bm{q} ={}$} ($(B0c.north east)+(-8pt,0pt)$);
        \end{tikzpicture}
      \end{center}
      (where we get the gray part by \cref{fct:commutatorInCrossDiagrams}) and obtain
      \[
        \bm{q} \circ u a v = B_0[\bm{b}\sub{D - 1}', \dots, \bm{b}_{0}'] \circ u a v = u a v \text{.}\tag*{\hbox{\qed}}
      \]
    \end{proof}
    
    \begin{remark}
      To calculate the number of states of $\mathcal{T}$, we assume that all parts of the automaton share the same $\id$ state and the same part for $\checksub{\id}$; the $r$ states are shared anyway. This yields
      \begin{alignat*}{2}
             & 1 &\quad& \text{($\id$ state)} \\
        {}+{}& |R| && \text{(for $\mathcal{R}$)} \\
        {}+{}& 5 |\Gamma| + 15 && \text{(for the encoding of $\checksub{\id}$)} \\
        {}+{}& |R| \cdot ( |\mathcal{T}_2| - 2 - 5 |\Gamma| - 15 ) && \text{(for the $|R|$ many copies of $\mathcal{T}_2$)} \\
        {}+{}& |R| \cdot ( \log |\Gamma| + 4 ) && \text{(for the $|R|$ many encoded $r_0$ automata)} \\
        {}={}& \rlap{$5 |\Gamma| + 16 + |R| \cdot (3 |\Gamma|^3 + 10 |\Gamma|^2 + 15 |\Gamma| + \log |\Gamma| + 41)$}
      \end{alignat*}
      many states, where $|R|$ is the number of states of $\mathcal{R}$ (for example, $|R| = 3$ for the Aleshin automaton from \cref{ex:freeGroup}), $|\Gamma|$ is the sum of the number of states and the number of tape symbols for a Turing machine for a $\PSPACE$-complete problem (which we assume to be a power of two) and where we have used the size of $\mathcal{T}_2$ from \cref{rmk:sizeOfT2}.
    \end{remark}
  \end{section}

  \begin{section}{Compressed Word Problem}\label{sec:compressedWP}
    In this section, we re-apply our previous construction to show that there is an automaton group with an \EXPSPACE-complete compressed word problem. The \emph{compressed word problem} of a group is similar to the normal word problem. However, the input element (to be compared to the neutral element) is not given directly but as a straight-line program. A \emph{straight-line program} is a context-free grammar which generates exactly one word.
    \begin{theorem}\labelx{thm:compressedEXPSPACE}
      There is an automaton group with an \EXPSPACE-complete compressed word problem:
      \problem
        [a \GAut $\mathcal{T} = (Q, \Sigma, \delta)$ with $|\Sigma| = 2$]
        {a straight-line program generating a state sequence $\bm{q} \in Q^{\pm *}$}
        {is $\bm{q} = \idGrp$ in $\mathscr{G}(\mathcal{T})$?}
    \end{theorem}
  
    The hard part of the proof of \cref{thm:compressedEXPSPACE} is (again) that the problem is \EXPSPACE-hard. That it is contained in \EXPSPACE space follows immediately by uncompressing the straight-line program (which yields at most an exponential blow up) and then applying the \PSPACE-algorithm for the normal word problem.

    The outline of the proof for the \EXPSPACE-hardness is the same as for the (normal) word problem: we start with a Turing machine $M$ and construct a \GAut $\mathcal{T} = (Q, \Sigma, \delta)$ from it in the same way as in the proof of \cref{thm:nonuniformPSPACE}. This time, however, the Turing machine accepts an (arbitrary) \EXPSPACE-complete problem and we assume that the length of its configurations is $s(n) = n + 1 + 2^{2n^e}$ for some positive integer $e$ where $n$ is the length of the input for $M$. Additionally, we assume the same normalizations as before.
    
    To keep things a bit simpler here, we do not use an arbitrary \GAut $\mathcal{R} = (R, \Sigma, \rho)$ for the commutators but fix the Aleshin automaton generating the free group $F_3$ (from \cref{ex:freeGroup}) for $\mathcal{R}$ in this section. In fact, it is even a bit more convenient to use the union of the Aleshin automaton $\mathcal{A} = (R_1, \Sigma, \eta)$ (where we also include the inverse states $a^{-1}, b^{-1}$ and $c^{-1}$) with its second power $\mathcal{A}^2$ for $\mathcal{R}$. The second power of an automaton is its composition with itself. Formally, $\mathcal{A}^2$ is the automaton $(R_1^2, \Sigma, \eta^2)$ where the transitions are given by
    \[
      \eta^2 = \{ \trans{r_2 r_1}{a}{c}{r_2' r_1'} \mid \trans{r_1}{a}{b}{r_1'}, \trans{r_2}{b}{c}{r_2'} \in \eta \text{ for some } b \in \Sigma \} \text{.}
    \]
    Note that, with this construction, $r_2 r_1$ acts in the same way when seen as a state of $\mathcal{A}^2$ as when seen as a sequence of states of $\mathcal{A}$.
    With this choice, we have that $b(D, i) = b^{-1} a$ is a state in $R$ and can avoid taking multiple copies of the encoding $\mathcal{T}_2$ over $\Sigma \supseteq \boxedAlph$ of $\mathcal{T}'$ from \cref{prop:TMmode} as we only need the copy for $r = b^{-1} a$.
    
    As in \cref{ex:freeGroup}, we use $\varepsilon$ for $\alpha$ and define $\beta: \mathbb{N} \to R^{\pm *}$ by
    \[
      \beta(d) = 
      \begin{cases}
        c & \text{for $d$ even,} \\
        b & \text{for $d$ odd.}
      \end{cases}
    \]
    This also yields $\alpha_0 = \varepsilon$ and $\beta_0$ (as defined in the proof of \cref{thm:nonuniformPSPACE}). So, effectively, we use $B_{\beta, \varepsilon}$ for $B_{\beta, \alpha}$ and $B_{\beta_0, \varepsilon}$ for $B_{\beta_0, \alpha_0}$ for the commutators. We continue to abbreviate the former by $B$ and the latter by $B_0$.
    
    We want to reduce the word problem of $M$ (which now is \EXPSPACE-complete) to the compressed word problem of $\mathscr{G}(\mathcal{T})$. Here, we cannot simply use the same proof as in the case of the (normal) word problem, however! Recall that we have mapped the input word $w$ of length $n$ to the state sequence
    \[
      \bm{q} = B_0 \big[ \bm{b}\sub{D - 1}', \dots, \bm{b}_0' \big] \text{.}
    \]
    In the case using the free group $F_3$ for $\mathscr{G}(\mathcal{R})$, we have $\bm{b}_i' = \bm{p}_{i, b^{-1}a} = \bm{p}_i$ where $\bm{p}_{i}$ is given by \cref{prop:TMmode}.\footnote{Remember that we only have a single copy of $\mathcal{T}_2$ (resp.\ $\mathcal{T}'$) this time and, thus, can omit the index $r$.} In turn, the $\bm{p}_i$ were given by $f, \bm{q}_{s(n) - 1}, \dots, \bm{q}_0, \bm{c}',\allowbreak \bm{c}_{s(n) - 1}, \dots, \bm{c}_0, z$ and $D$ was (the next power of two after) $3 + 2s(n)$ in the proof of \cref{prop:TMmode}.
    
    This is a problem since we now have exponentially many $\bm{c}_i$ and $\bm{q}_i$ and we, thus, cannot output all of them with a \LOGSPACE\ (or even polynomial time) transducer -- even if we compress every individual $\bm{c}_i$ and $\bm{q}_i$ using a straight-line program. On the positive side, we have that all $\bm{c}_i$ and all except linearly many $\bm{q}_i$ are structurally very similar: we have
    \[
      \bm{c}_i = {\checksub{\id}}^{- i} \, \checksub{\id}^{-1} \checksub{b^{-1}a} \, \checksub{\id}^{i} \quad \text{ and }\quad
      \bm{q}_j = {\checksub{\id}}^{- j} \, q_{\blank} \, \checksub{\id}^{j}
    \]
    for all $0 \leq i < s(n)$ and all $n + 1 \leq j < s(n)$. Due to this structural similarity, we will still be able to output a single straight-line program that generates a word equal to $B_0[\bm{c}_{s(n) - 1}, \dots, \bm{c}_{n + 1}]$ in $\mathscr{G}(\mathcal{T})$ and one generating a word equal to $B_0[\bm{q}_{s(n) - 1}, \dots, \bm{q}_{n + 1}]$ in $\mathscr{G}(\mathcal{T})$.
    
    \paragraph{Twisted Balanced Iterated Commutators.}
    For the construction of these straight-line programs, we use a twisted version of our nested commutators, where the left side has an additional conjugation.
  
    \begin{definition}
      Let $Q$ be some alphabet, $\alpha, \beta: \mathbb{N} \to Q^{\pm *}$ and $\gamma \in Q^{\pm *}$. For $\bm{p} \in {Q}^{\pm *}$, we define $B_{\beta, \alpha}(\bm{p}, D)$ inductively for all $D = 2^d$:
      \begin{align*}
        B_{\beta, \alpha}^\gamma(\bm{p}, 1) &= \bm{p} \text{ and} \\
        B_{\beta, \alpha}^\gamma(\bm{p}, 2D) &= \Big[
          \left( \gamma^{-D}\, B_{\beta, \alpha}^\gamma(\bm{p}, D) \, \gamma^D \right)^{\beta(d)}, \,
          \left( B_{\beta, \alpha}^\gamma(\bm{p}, D)  \right)^{\alpha(d)}
        \Big]
      \end{align*}
    \end{definition}\noindent
  
    The compatibility between commutators and conjugation allows us to use the twisted version to move an iterated conjugation of the commutator entries into the commutator itself. As we have already seen, the check state sequences $\bm{c}_i$ and (most) $\bm{q}_i$ are of this form.
    \begin{lemma}\labelx{lem:twistingIsConjugaction}
      Let $G$ be a group generated by a finite set $Q$, $\alpha, \beta: \mathbb{N} \to Q^{\pm *}$ and $\gamma \in Q^{\pm *}$ such that $\gamma$ commutes in $G$ with $\alpha(d)$ and $\beta(d)$ for all $d$. Furthermore, for some $\bm{p} \in Q^{\pm *}$, let
      \[
        \bm{p}_i = {\gamma}^{-i} \bm{p} \gamma^{i}
      \]
      for $0 \leq i$. Then, we have
      \[
        B_{\beta, \alpha}^\gamma(\bm{p}, D) = B_{\beta, \alpha} \big[ \bm{p}\sub{D - 1}, \dots, \bm{p}_0] \text{ in } G
      \]
      for all $D = 2^d$.
    \end{lemma}
    \begin{proof}
      We simply write $B^\gamma$ for $B_{\beta, \alpha}^\gamma$ and $B$ for $B_{\beta, \alpha}$ and prove the statement by induction on $d$. For $d = 0$ (or, equivalently $D = 1$), we have $B^\gamma[\bm{p}, 1] = \bm{p} = \bm{p}_0 = B[\bm{p}_0]$ and, for the inductive step from $d$ to $d + 1$ (or, equivalently, from $D$ to $2D$), we have in $G$:
      \begin{multline*}
        B^\gamma[\bm{p}, 2D] \\
        \begin{alignedat}[b]{2}
          &= \Big[
            \left( \gamma^{-D}\, B^\gamma(\bm{p}, D) \, \gamma^D \right)^{\beta(d)}, \,
            \left( B^\gamma(\bm{p}, D) \right)^{\alpha(d)}
            \Big] \\
          &= \Big[
            \left( \gamma^{-D}\, B \big[ \bm{p}\sub{D - 1}, \dots, \bm{p}_0 \big] \, \gamma^D \right)^{\beta(d)}, \\
          &\phantom{{}= \Big[}
            \left( B \big[ \bm{p}\sub{D - 1}, \dots, \bm{p}_0 \big] \right)^{\alpha(d)}
          \Big] && \text{ (by induction)} \\
          &= \Big[
            \left( B \big[ \bm{p}\sub{D - 1}^{\gamma^D}, \dots, \bm{p}_0^{\gamma^D} \big] \right)^{\beta(d)}, \,
            \left( B \big[ \bm{p}\sub{D - 1}, \dots, \bm{p}_0 \big] \right)^{\alpha(d)}
            \Big] && \text{ (by \cref{fct:conjugatedCommutator})} \\
          &= \Big[
            \left( B \big[ \bm{p}_{2D - 1}, \dots, \bm{p}\sub{D} \big] \right)^{\beta(d)}, \,
            \left( B \big[ \bm{p}\sub{D - 1}, \dots, \bm{p}_0 \big] \right)^{\alpha(d)}
            \Big] && \text{\rlap{ (by definition of $\bm{p}_i$)}} \\
          &= B \big[ \bm{p}_{2D - 1}, \dots, \bm{p}_0 \big] && \text{ (by definition)}
        \end{alignedat}\hfil\hbox{\qed}\hfilneg
      \end{multline*}
    \end{proof}
    
    The connection in \cref{lem:twistingIsConjugaction} allows us to use $B_{\beta, \varepsilon}^\gamma$ for the check sequences $\bm{c}_i$ and $\bm{q}_i$. The advantage of this approach is that the twisted version can efficiently be compressed into a straight-line program (although the corresponding (ordinary) balanced iterated commutator would have too many entries).
    \begin{lemma}\labelx{lem:straightLineForTwistedCommutator}
      If the functions $\alpha$ and $\beta$ are computable in \DLINSPACE (where the input is given in binary), we can compute a straight-line program for $B_{\beta, \alpha}^\gamma(\bm{p}, D)$ with $D = 2^d$ on input of $\bm{p}$ and $D$ in logarithmic space (where $D$ is given in binary).
    \end{lemma}
    \begin{proof}
      The alphabet of the straight-line program is obviously $Q^{\pm 1}$ and we only give the variables implicitly. Clearly, if we can compute the production rules for a variable $X$ generating a state sequence $\bm{q}$, we can also compute the production rules for a variable $X^{-1}$ generating $\bm{q}^{-1}$. Therefore, we only give the positive version for every variable (but always assume that we also have a negative one).
      
      First, we add the production rules
      \[
        M_{2^{d - 1}} \to M_{2^{d - 2}} M_{2^{d - 2}}, \quad \dots, \quad M_{2^1} \to M_{2^0} M_{2^0}, \quad M_{2^0} \to \gamma
      \]
      because we need blocks of $\gamma$ for the recursion of $B_{\beta, \alpha}^\gamma$. Clearly, $M_{2^i}$ generates $\gamma^{2^i}$ for all $0 \leq i < d$. For the actual commutator, we use the variables $A_{2^i}$ for $0 \leq i \leq d$ and add the production rules
      \begin{align*}
        A_{2^0} &\to \bm{p} \quad \text{and} \\
        A_{2^{i + 1}} &\to
          \begin{array}[t]{c@{\;}c@{\;}c@{\;}c}
            \beta(i)^{-1} & {M}_{2^{i}}^{-1} {A}_{2^{i}}^{-1} M_{2^{i}} &\beta(i) &\; \alpha(i)^{-1} {A}_{2^{i}}^{-1} \alpha(i) \; \\
            \beta(i)^{-1} & {M}_{2^{i}}^{-1} A_{2^{i}} M_{2^{i}} & \beta(i) &\; \alpha(i)^{-1} A_{2^{i}} \alpha(i)
          \end{array}
      \end{align*}
      for all $0 \leq i < d$. Using a simple induction it is now easy to see that $A\sub{D}$ generates $B_{\beta, \alpha}^\gamma[\bm{p}, D]$ for $D = 2^d$. Accordingly, we choose $A\sub{D}$ as the start variable.
      
      By assumption, the functions $\alpha$ and $\beta$ can be computed in \DLINSPACE on input of the binary representation of $d$~-- this means, the required space is logarithmic in $d$.
      To compute the productions rules, we obviously only need to count up to $d$ (in binary) and this can clearly be done in logarithmic space with respect to the binary length of $D$ (which is $d$).\qed
    \end{proof}
    
    With these twisted commutators and the corresponding straight-line programs, we have introduced the missing pieces to adapt the proof for $\PSPACE$ and the (normal) word problem from \cref{thm:nonuniformPSPACE} to $\EXPSPACE$ and the compressed word problem.

    \begin{proof}[Proof of \cref{thm:compressedEXPSPACE}]
      As we have already remarked, we only need to show that the problem is \EXPSPACE-hard and construct the \GAut $\mathcal{T} = (Q, \Sigma, \delta)$ from the machine $M$ for an \EXPSPACE-complete problem in the same way as in the proof of \cref{thm:nonuniformPSPACE}.
      
      However, compared to the case of the (normal) word problem, we need to make some changes to the reduction of the (normal) word problem of $M$ to the compressed word problem of $\mathscr{G}(\mathcal{T})$. We map an input word $w$ of length $n$ for $M$ to a straight-line program for a state sequence $\bm{q}$ equal to\footnote{To be absolute precise: we actually need to repeat one of the entries of the outer commutator in order to get a power of two as the number of entries.}
      \[
        B_0\! \left[ f, B_0[\bm{q}_{s(n) - 1}, \dots, \bm{q}_{n + 1}], \bm{q}_{n}, \dots, \bm{q}_{0}, \bm{c}', B_0[\bm{c}_{s(n) - 1}, \dots, \bm{c}_{n + 1}], \bm{c}_{n}, \dots, \bm{c}_0, z \right]
      \]
      in $\mathscr{G}(\mathcal{T})$. If the Turing machine $M$ accepts the input $w$, this means (by \cref{prop:TMmode} and \cref{fct:commutatorInCrossDiagrams}) that there is some $u \in Y^* = (\{ 0, 1, \# \} \cup \Gamma)^* \subseteq \Sigma^*$ such that $\bm{q} \cdot u \$$ is in $\mathscr{G}(\mathcal{T})$ equal to
      \begin{align*}
        &B \Big[
          b^{-1}a, \,
          B[\underbrace{b^{-1}a, \dots, b^{-1}a}_{\mathclap{s(n) - 1 - n = 2^{2n^e} \text{ many}}}], \,
          \underbrace{b^{-1}a, \dots, b^{-1}a}_{\mathclap{\phantom{2^{2n^e}}n + 1 \text{ many}}}, \\
        &\phantom{B \Big[}
          b^{-1}a, \,
          B[\underbrace{b^{-1}a, \dots, b^{-1}a}_{\mathclap{s(n) - 1 - n = 2^{2n^e} \text{ many}}}], \,
          \underbrace{b^{-1}a, \dots, b^{-1}a}_{\mathclap{\phantom{2^{2n^e}}n + 1 \text{ many}}}, \,
          b^{-1}a \Big]
      \end{align*}
      (because we have $\bm{p}_{i} \cdot u \$ = r = b^{-1} a = b(D, i)$ in this case). The inner commutators are both equal to $B_3(2^{2n^e})$ (from \cref{ex:freeGroup}) and, thus, in $b^{-1} \{ a, b, c \}^{\pm *} a$ (and freely reduced). This means that the outer commutator is a non-empty freely reduced word in $F_3$ (as we have already pointed out in \cref{ex:freeGroup}). If the Turing machine does not accept $w$, we have $\bm{p}_i \cdot u \$ = \idGrp$ in $\mathscr{G}(\mathcal{T})$ for some $i$ and the commutators also collapse to the neutral element by \cref{fct:commutatorCollapses}. This shows that we can use the same argument as in the proof of \cref{thm:nonuniformPSPACE}.
      
      Thus, it remains to describe how we can compute a straight-line program for such a $\bm{q}$ in logarithmic space. If we have straight-line programs for the individual entries, we immediately also obtain straight-line programs for their inverses and can combine everything into a straight-line program for the overall nested commutator. This can be done (on the level of the variables) in logarithmic space by \cref{fct:BIsLogspaceComputable}. For $f$, $\bm{q}_{n}, \dots, \bm{q}_{0}$, $\bm{c}_{n}, \dots, \bm{c}_0$ and $z$, we do not even need straight-line programs but can output the words directly (as in the \PSPACE-case).
      
      For the inner commutators, recall that we have
      \begin{align*}
        \begin{array}{c@{\;}r@{\;}c@{\;}l}
          \bm{c}_{n + 1 + i} &= {\checksub{\id}}^{- i} & \left( \checksub{\id}^{- (n + 1)} \, \checksub{\id}^{-1} \kern-0.4ex \checksub{b^{-1}a} \, \checksub{\id}^{n + 1} \right) & \checksub{\id}^{i} \quad \text{and} \\
          \bm{q}_{n + 1 + i} &= {\checksub{\id}}^{- i} & \left( \checksub{\id}^{- (n + 1)} \, q_{\blank} \, \checksub{\id}^{n + 1} \right) & \checksub{\id}^{i}
        \end{array}
      \end{align*}
      for all $0 \leq i < 2^{2n^e}$. Thus, we have
      \begin{align*}
        B_0[\bm{c}_{s(n) - 1}, \dots, \bm{c}_{n + 1}] &= B_0^{\checksub{\id}}(\checksub{\id}^{- (n + 1)} \checksub{\id}^{-1} \kern-0.4ex \checksub{b^{-1}a} \checksub{\id}^{n + 1}, 2^{2n^e}) \quad \text{and} \\
        B_0[\bm{q}_{s(n) - 1}, \dots, \bm{q}_{n + 1}] &= B_0^{\checksub{\id}}(\checksub{\id}^{- (n + 1)} q_{\blank} \checksub{\id}^{n + 1}, 2^{2n^e})
      \end{align*}
      in $\mathscr{G}(\mathcal{T})$ by \cref{lem:twistingIsConjugaction} where we write $B_0^{\checksub{\id}}$ for $B_{\beta_0, \alpha_0}^{\checksub{\id}}$. In order to apply \cref{lem:twistingIsConjugaction}, we need that $\checksub{\id}$ commutes with ($\varepsilon$ and) $\beta_0(d)$ for all $d$. However, this immediately follows from the construction of $\mathcal{T}$ (as $\checksub{\id}$ only manipulates the TM part of the input word while $b_0$ and $c_0$ only manipulate the commutator part).
      
      Finally, we can compute a straight-line program for the two twisted commutators in \LOGSPACE by \cref{lem:straightLineForTwistedCommutator}. Here, it is important that $n$ is the input length and that, thus, $2^{2n^e}$ can be output in binary (since it has length $2n^e$).
      
      The last remaining part is a straight-line program for 
      \[
        \bm{c}' = {\checksub{\id}}^{-s(n)} c \checksub{\id}^{s(n)} = {\checksub{\id}}^{-2^{2n^e}} \, {\checksub{\id}}^{- (n + 1)} c \checksub{\id}^{n + 1} \, \checksub{\id}^{2^{2n^e}} \text{.}
      \]
      The inner part can be output directly and the outer $\checksub{\id}$-blocks of length ${2^{2n^e}}$ can be generated in the same way as in the proof of \cref{lem:straightLineForTwistedCommutator}.\footnote{In fact, we already have the required production rules and variables up to $M_{2^{2n^e - 1}}$.}\qed
    \end{proof}
    
    We can take the disjoint union of the \GAut over $\Sigma$ with a \PSPACE-complete word problem and the \GAut over $\Sigma$ with an \EXPSPACE-complete compressed word problem. In this way, we obtain an automaton group whose (normal) word problem is $\PSPACE$-complete and whose compressed word problem is $\EXPSPACE$-complete and, thus, provably harder (by the space hierarchy theorem \cite[Theorem~6]{stearns1965hierarchies}, see also e.\,g.\ \cite[Theorem~7.2, p.~145]{papadimitriou97computational} or \cite[Theorem~4.8]{arora2009computational}).
    \begin{corollary}
      There is an automaton group with a binary alphabet whose word problem
      \problem
        [a \GAut $\mathcal{T} = (Q, \Sigma, \delta)$ with $|\Sigma| = 2$]
        {a state sequence $\bm{q} \in {Q}^{\pm *}$}
        {is $\bm{q} = \idGrp$ in $\mathscr{G}(\mathcal{T})$?}\noindent
      is $\PSPACE$-complete and whose compressed word problem
      \problem
        [the same \GAut $\mathcal{T} = (Q, \Sigma, \delta)$ with $|\Sigma| = 2$]
        {a straight-line program generating a state sequence $\bm{q} \in Q^{\pm *}$}
        {is $\bm{q} = \idGrp$ in $\mathscr{G}(\mathcal{T})$?}\noindent
      is $\EXPSPACE$-complete.
    \end{corollary}
  \end{section}
  
  \bibliographystyle{spmpsci}% the mandatory bibstyle
  \bibliography{references}
\end{document}